\documentclass[reqno,10pt]{amsart}
\evensidemargin5mm \textwidth14.8cm \headsep0.8cm \headheight0.4cm\numberwithin{equation}{section}
\usepackage{pinlabel-test}
\usepackage{placeins}
\usepackage{color}
\usepackage{mathrsfs}
\usepackage[normalem]{ulem} 
\usepackage{hyperref}
\usepackage{fullpage}
\usepackage[hmargin=1in,vmargin=3cm]{geometry}

\newtheorem{rem}{Remark}[section]

\renewcommand{\d}[1]{\mathrm{d} #1}
\newcommand{\dvol}{\d\textup{Vol}}
\newcommand{\dvolg}{\d\textup{Vol}} 
\newcommand{\eps}{\varepsilon}
\newcommand{\eb}{\mathbf{e}}
\newcommand{\eu}{\mathbf{e}_1}
\newcommand{\ev}{\mathbf{e}_2}
\newcommand{\ei}{\mathbf{e}_i}
\newcommand{\ej}{\mathbf{e}_j}
\newcommand{\fu}{\mathbf{f}_1}
\newcommand{\fd}{\mathbf{f}_2}
\newcommand{\n}{\mathbf{n}}
\newcommand{\nb}{\boldsymbol{\nu}}
\newcommand{\N}{{\mathbb N}}
\newcommand{\tbf}{\mathbf{t}}
\newcommand{\bbT}{\mathbb{T}}
\renewcommand{\vec}[1]{\mbox{\boldmath$#1$}}
\newcommand{\ub}{\vec{u}}
\newcommand{\vb}{\vec{v}}

\newcommand{\Hut}{H^1_{tan}(\Sigma;\mathbb{S}^2)}
\newcommand{\HT}{H^1_{tan}(\mathbb{T};\mathbb{S}^2)}
\newcommand{\Ht}{H^{1}_{tan}(\Sigma)}
\newcommand{\Lt}{L^2_{tan}(\Sigma)}
\newcommand{\B}{\mathfrak{B}}
\newcommand{\Bd}{\mathfrak{B}^2}
\newcommand{\Rt}{\R^3}
\newcommand{\RN}{\R^{N}}

\newcommand{\R}{\mathbb{R}}
\renewcommand{\S}{{\mathbb S}}
\newcommand{\Z}{\mathbb{Z}}

\newcommand{\argmin}{\hbox{argmin}}

\newcommand{\delj}{\frac{\partial}{\partial x^j}}

\DeclareMathOperator{\Div}{\textup{div}}
\DeclareMathOperator{\curl}{\textup{curl}}
\DeclareMathOperator{\grads}{\textup{grad}_s}
\DeclareMathOperator{\divs}{\textup{div}}	
\DeclareMathOperator{\curls}{\textup{curl}} 

\newtheorem{theorem}{Theorem}
\newtheorem{lemma}{Lemma}[section]
\newtheorem{proposition}{Proposition}[section]

\newenvironment{definition}[1][Definition]{\begin{trivlist}
\item[\hskip \labelsep {\bfseries #1}]}{\end{trivlist}}

\newtheorem{remark}{Remark}[section]

\begin{document}


\title{Analysis of a variational model for nematic shells}

\author{Antonio Segatti, Michael Snarski and Marco Veneroni}

\date{\today}

\maketitle

\begin{abstract}
We analyze an elastic surface energy which was recently introduced by G. Napoli and L. Vergori to model thin films of nematic liquid crystals. We show how a novel approach in modeling the surface's extrinsic geometry leads to considerable differences with respect to the classical intrinsic energy. Our results concern three connected aspects: i) using methods of the calculus of variations, we establish a relation between the existence of minimizers and the topology of the surface; ii) we prove, by a Ginzburg-Landau approximation, the well-posedness of the gradient flow of the energy; iii) in the case of a parametrized torus we obtain a stronger characterization of global and local minimizers, which we supplement with numerical experiments.
\end{abstract}

\tableofcontents





\section{Introduction}

Liquid crystals are an intermediate phase of matter between solid and fluid states which possess
peculiar optical properties and are controllable through electric and magnetic
fields. As a result, they play a fundamental role in the development of many scientific applications and in the 
design of new generation technologies. 

A \emph{nematic shell} is a thin film of nematic liquid crystal coating a
rigid and curved substrate $\Sigma$
which is typically represented as a two-dimensional surface. The basic mathematical description of these shells is given 
in terms of a unit vector field constrained to be tangent to the substrate $\Sigma$. 
This vector field
will be called the {\itshape director}, analogous to the nomenclature for liquid crystals in domains.  The rigorous mathematical
treatment of nematic shells is intriguing since it combines 
tools from diverse fields such as the calculus of variations, partial differential equations, topology,
differential geometry and numerical analysis. Our study
is further motivated by the vast technological applications of nematic shells, as discussed
in \cite{Nelson02}. To the best of our knowledge, the study of these structures has been mostly confined to
the physical literature (see, e.g., \cite{KRV11, LubPro92, NapVer12E, Straley71, VitNel06}) with the sole exception of \cite{Shk02}.

The form of the elastic energy for nematics is well 
established, both in the framework of director theory which is based on the works 
of Oseen, Zocher, and Frank, and in the framework of the order-tensor 
theory introduced by de Gennes (see, e.g., \cite
{DegPro93,Virga94}). On the other hand, when dealing with nematic
shells, there is no universal agreement 
on the form of a \emph{two-dimensional} free energy. 
The differences between the various approaches arise in the choice of the
local distortion element of the substrate, i.e., the 
effect of the substrate's extrinsic geometry on the elastic energy of 
the nematics.  Indeed, as 
observed in \cite{BowGio2009, VitNel06}, the liquid crystal 
ground state (and all its stable configurations, in general) is determined by two
competing, driving principles: on one hand the 
minimization of the ``curvature of the texture" penalized by the elastic 
energy, and on the other the frustration due to constraints of geometrical 
and topological nature, imposed by anchoring the nematic to the surface 
of the underlying particle. 
A new energy model proposed by 
Napoli and Vergori in \cite{NapVer12L,NapVer12E} affects these two aspects, leading to different results with respect to the classical models  \cite{HelPro88,LubPro92, Straley71}.
It is interesting to note that a definitive microscopic justification of these energies is still to be found.

The aim of this paper is to analyze 
the new surface energy for liquid crystal shells proposed in \cite{NapVer12L, NapVer12E}.
To describe our results and to highlight some of the related difficulties, 
let us consider at first the simplest \emph{one-constant approximation} of the surface energy on a two-dimensional surface $\Sigma\subset \mathbb{R}^3$:
\begin{equation}
\label{eq:napoli_intro}
W_{\kappa}(\n):= \frac \kappa2 \int_\Sigma
	| D\n|^2 + |\B \n|^2\,\dvolg,
\end{equation}
where $\n$ is a unit norm and tangent vector field on $\Sigma$ representing, for any 
point on $\Sigma$, the mean orientation of the nematic molecules; here $\kappa$ is a positive constant, the symbol $D$ denotes the covariant derivative on $\Sigma$, and $\B$ is the shape operator (see Section \ref{sec:prelimin} for all the details and definitions).
Our results address 
\begin{itemize}
	\item[(a)] the relation between the topology of the surface and the functional setting,
	\item[(b)] the minimization of \eqref{eq:napoli_intro} and the well posedness of its gradient flow on a general genus one surface,
	\item[(c)] the precise structure of local minimizers on a particular surface: the axisymmetric torus.
\end{itemize}
We pay particular attention to the gradient flow of the energy because, aside from being an interesting mathematical object on its own, it provides an efficient tool for numerical approximations of minimizers. Furthermore, it can be seen as a first step towards the evolutionary study of liquid crystals on surfaces. While Step (a) is necessary to give a rigorous formulation to the problem, Steps (b) and (c) complement each other: The general analysis in (b) has the advantage of being
applicable to any two-dimensional  topologically admissible surface and even,
up to some technical obstacles, to $(N-1)$-dimensional compact and smooth
hypersurfaces embedded in $\R^N$.  In (c) we sacrifice generality in order to
obtain more precise analytical and numerical information on the solutions. In particular, the regularity
issue and the existence of solutions with prescribed winding number, which seem 
difficult to be obtained by working directly on \eqref{eq:napoli_intro},  are more transparent.

\medskip

\noindent \textbf{(a) Topological constraints.}
Given the form of \eqref{eq:napoli_intro}, it would be natural to set its analysis in the ambient space of tangent vector fields such that $|\n|$ and $|D\n|$ belong to $L^2(\Sigma)$. We refer to the quantity $\int_\Sigma |D\n|^2$ as the \emph{Dirichlet energy} of $\n$. However, the topology of the surface may force the subset of vector fields with $|\n|=1$, which would represent our directors, to be empty. This could be heuristically explained as follows.
Let $\vb$ be a smooth tangent vector field on $\Sigma$, with finitely many zeroes. The index $m\in \Z$ of a zero $\bar x\in \Sigma$ is, intuitively, the number of counterclockwise rotations that the vector completes around a small circle around $\bar x$. So, if $m\neq 0$, the corresponding unit-length vector field $\vb/|\vb|$ has a discontinuity at $\bar x$ (see Figure \ref{fig_defects}).  By the Poincar\'e-Hopf index Theorem \cite[Chapter 3]{GuiPol74}, the global sum of the indices of the zeroes of $\vb$ equals the Euler characteristic $\chi(\Sigma)$ and therefore it is possible to find a smooth field $\n$ with $|\n|\equiv 1$ on $\Sigma$ if and only if $\chi(\Sigma)=0$, i.e. if $\Sigma$ is a genus-1 surface (``hairy ball Theorem"). Moreover, a direct computation (say for $m=1$) shows that the Dirichlet energy of $\vb/|\vb|$ in any small enough annulus centered at $\bar x$, with internal radius $\rho$, scales like $|\log(\rho)|$ as $\rho\to 0$. Therefore, one would expect the topological constraint of the hairy ball Theorem to hold also for $H^1$-regular vector fields.
\begin{figure}[h]
 \includegraphics[height=2.5cm]{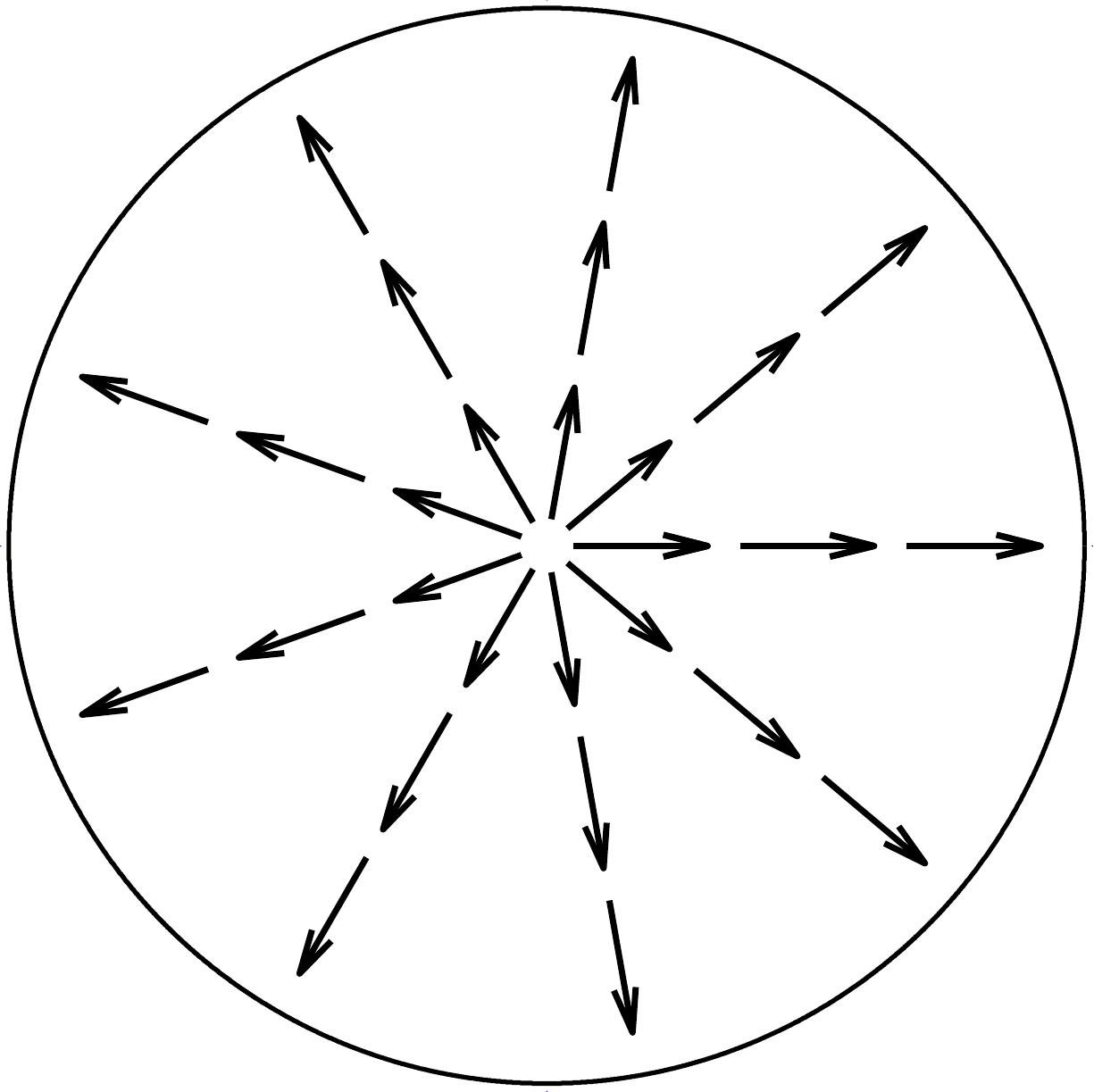}\qquad
  \includegraphics[height=2.5cm]{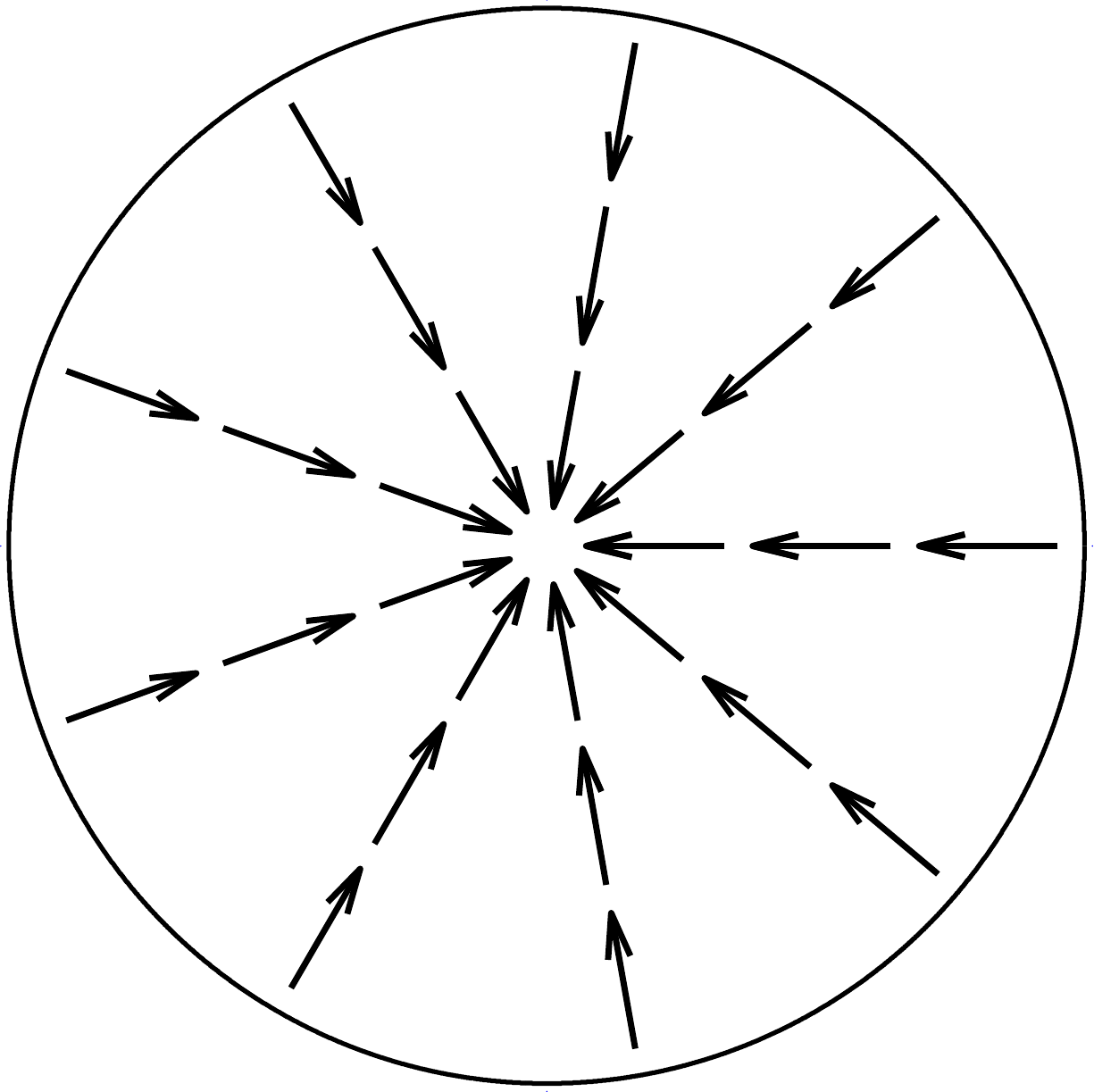}\qquad
 \includegraphics[height=2.5cm]{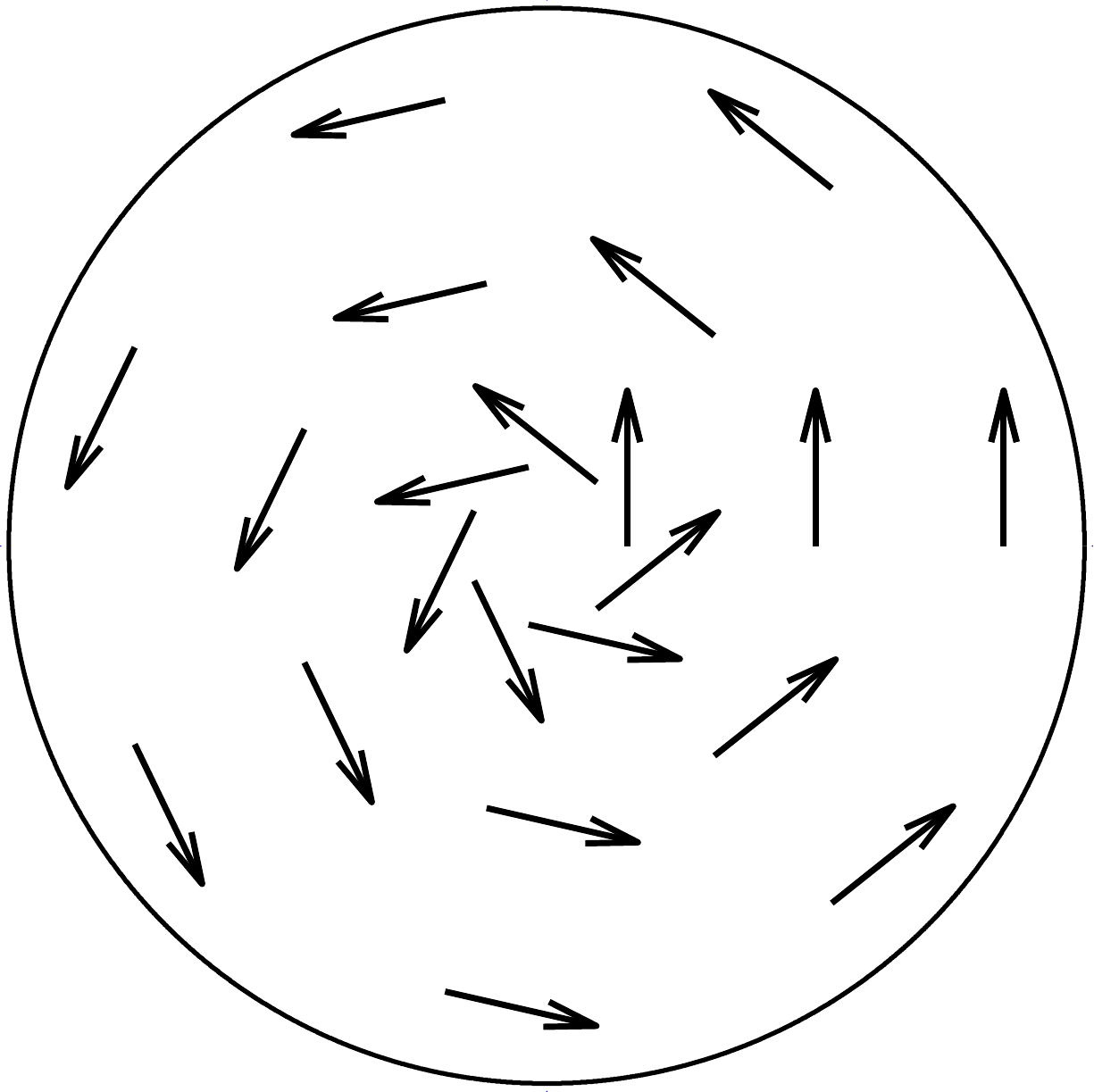}\\  \vspace{-0.2cm}
\caption{Examples of unitary vector fields on a disc in $\R^2$, showing topological defects with index 1. }
\label{fig_defects}
\end{figure}
Indeed, in Theorem \ref{Th:H1empty} we generalize the hairy ball Theorem to $H^1$-regular vector fields. 
Theorem \ref{Th:H1empty} is indeed a consequence of the more general results of \cite{CSV13} in which 
we prove the Morse Index formula for unit norm vector fields with VMO regularity. However, one can prove directly Theorem \ref{Th:H1empty} by using PDEs arguments. For the sake of completeness we decided to include the proof in the paper.
 The idea of the proof is the following. First, using variational methods  we show that if the set of $H^1$-regular tangent unitary fields on $\Sigma$ is not empty, then it includes a minimizer of the Dirichlet energy on $\Sigma$. 
Then, using the local representation introduced in Section \ref{sec:alpha} and regularity theory for elliptic PDEs, we show that this minimizer is continuous. By the classical hairy ball Theorem, we conclude that $\Sigma$ has to be a genus-1 surface.
 Note that the exponent $2$ in \eqref{eq:napoli_intro} is a limit-case, as it is possible to construct unitary fields such that $|D\vb|\in L^{p}(\Sigma)$ for any $p\in [1,2)$, on any smooth compact surface $\Sigma$.
In view of Theorem \ref{Th:H1empty}, we restrict our study to genus-1 surfaces, where the underlying geometry of the substrate does not force the creation of defects. 
A rigorous analysis of the distribution and evolution of defects on nematic surfaces is an interesting problem 
(see \cite{CSV_defects}) which is beyond the scope of this paper.
Due to its
large potential impact on the design of new generation metamaterial structures (see \cite{Nelson02, YiEtal13}), this question has garnered a good deal of interest within the physics community
(see \cite{KRV11,R_V_K_trans,S_K_T_S11,N_G_S_S13,VitNel06}). To the best of our knowledge it still lacks a rigorous mathematical treatment. A different approach to defects, following an approximation of Ginzburg-Landau type, was studied in \cite{Baraket96}.

\medskip

\noindent \textbf{(b) Well-posedness on general surfaces.}
The general form of the surface energy \eqref{eq:napoli_intro}, introduced in \cite{NapVer12L}, is the surface analogue of the well-studied Oseen, Zocher and Frank model (see, e.g., \cite{Virga94}) and is defined as
\begin{equation}
\label{eq:napolif_intro}
	W(\n):=\frac 12 \int_\Sigma K_1 (\divs \n)^2 + K_2 (\n \cdot \curls \n)^2 +K_3 |\n \times \curls \n|^2\,\dvolg.
\end{equation}
In the above display, the subscript $s$ denotes surface operators (see Section \ref{sec:prelimin}) and $K_1, K_2, K_3$ are positive constants known respectively as the splay, twist and bend moduli. Using the direct method of the calculus of variations, in Proposition \ref{propo:ex_min} we prove existence of a minimizer of \eqref{eq:napolif_intro}. We then focus on the $L^2$-gradient flow of \eqref{eq:napoli_intro}, in the case of $\kappa:=K_1=K_2=K_3$. 
The study of the gradient flow for the energy \eqref{eq:napoli_intro} could be 
seen as a starting point for the analysis of an Ericksen-Leslie type model for nematic
shells. This problem has already been addressed in \cite{Shk02} where
various well-posedness and long time behavior results have been obtained
for an Ericksen-Leslie type model on Riemannian manifolds. However, 
it should be pointed out that the model in \cite{Shk02} is purely intrinsic
and does not take into account the way the substrate on which the nematic
is deposited sits in the three-dimensional space.

In Theorem \ref{Th:gfn} we prove the well-posedness of the $L^2$-gradient flow of \eqref{eq:napoli_intro}, i.e.
\begin{equation}
\label{eq:gf_intro}
	\partial_t \n -\Delta_g \n + \Bd \n = \vert D\n\vert^2 \n + \vert \B\n\vert^2 \n\quad \hbox{in  } \Sigma \times (0,+\infty).
\end{equation}
Here $\Delta_g$ is the \emph{rough Laplacian}, $D$ is the covariant derivative and $\B$ is the shape operator on $\Sigma$ (see Section \ref{sec:prelimin}). The right-hand side of \eqref{eq:gf_intro} is a result of the unit-norm constraint on the director $\n$. 
  A proof of the existence relying on  \textit{i) discretization, ii) a priori estimates, iii) convergence of discrete solutions}, would encounter a difficulty here, as the nonlinear term $|D\n|^2$ in the right-hand side of \eqref{eq:gf_intro} is not continuous with respect to the weak-$H^1$ convergence expected from the a priori estimates. We overcome this problem with techniques employed in the study of the heat  flow for harmonic maps (see \cite{chen, chen_struwe}): we first relax the unit-norm 
constraint with a Ginzburg-Landau approximation, i.e.,  we allow for vectors $\n$ with $\vert\n\vert \neq 1$, but we penalize deviations from unitary length at the order $1/\eps^2$, for a small parameter $\eps>0$. In this way, it is possible to build a sequence of  fields $\n^\eps$, with $|\n^\eps|\to 1$ as $\eps \to 0$, which solve an approximation of \eqref{eq:gf_intro}, with zero right-hand side. The crucial remark, in order to recover \eqref{eq:gf_intro} in the limit, is that for a smooth unit-norm field $\n$, \eqref{eq:gf_intro} is equivalent to 
\[
	(\partial_t \n -\Delta_g \n + \Bd \n) \times \n =0.
\]
When passing to the limit, the non-trivial term is $\Delta_g \n \times \n$, which can be treated by a careful surface integration by parts.

\medskip

\noindent \textbf{(c) Parametric representation on a torus.}
A common way to study unit-norm tangent vector fields on a surface $\Sigma$ is to introduce a scalar parameter $\alpha$ which measures the rotation of $\n$ with respect to a given orthonormal frame $\{\eu,\ev\}$, i.e. $\n = \eu \cos(\alpha)+\ev\sin(\alpha)$. The local existence of such a representation is straightforward, but since a global one on $\Sigma$ is in general not possible (even when the topology of $\Sigma$ allows for $H^1$-fields), we first prove that for every $H^1$-regular unit-norm vector field $\n$ there exists a representation $\alpha\in H^1_{\textup{loc}}(\R^2)$ defined on the universal covering of $\Sigma$ (Proposition \ref{prop:alp}). Then, we express the energies \eqref{eq:napoli_intro} and \eqref{eq:napolif_intro}, and the relative Euler-Lagrange equations,  in terms of $\alpha$. With this representation in hand, we focus on a specific parametrization of the axisymmetric torus in $\R^3$.  The main advantages are that we now deal with the scalar quantity $\alpha$, instead of the vector $\n$, that through the parametrization we can reduce to work on a flat domain, e.g. $Q=[0,2\pi]\times [0,2\pi]$, and that the unit-norm constraint does not appear explicitly. The disadvantages are that the representation is not unique (as $\alpha$ and $\alpha+2k\pi$ yield the same field $\n$) and that the parametrization introduces an unusual condition of ``periodicity modulo $2\pi$" on the boundary of $Q$.

In \cite{SSV14} we used this approach to explicitly calculate the value of the energy \eqref{eq:napolif_intro} on constant deviations $\alpha$.  The interest lies in understanding the dependence of the energy on the mechanical parameters $K_i$ and on the aspect ratio of the torus, even on a special set of configurations. The constant configurations $\alpha_m:= 0$ and $\alpha_p:=\pi/2$ (see Figure \ref{fig:extra}) are of particular interest, as, up to an additive constant, the $\alpha$-representation of \eqref{eq:napoli_intro} is
\begin{equation}
\label{eq:walpha_intro}
		W_\kappa(\alpha)=\frac \kappa 2\int_Q \left\{|\nabla_s \alpha|^2 +\eta\cos(2\alpha)\right\}\, \dvol ,
\end{equation}
where $\eta$ is a function which depends only on the geometry of the torus. This structure, a Dirichlet energy plus a double (modulo $2\pi$) well potential, is well-studied in the context of Cahn-Hilliard phase transitions. Depending on the torus aspect ratio, the sign of $\eta$ may not be constant on $Q$, thus forcing a smooth transition between the states $\alpha_m$, where $\eta<0$, and $\alpha_p$, where $\eta>0$.

\begin{figure}[ht]
	\labellist
			\hair 2pt
			\pinlabel $\alpha_m$ at 330 -10
		\endlabellist				
		\includegraphics[height=2cm]{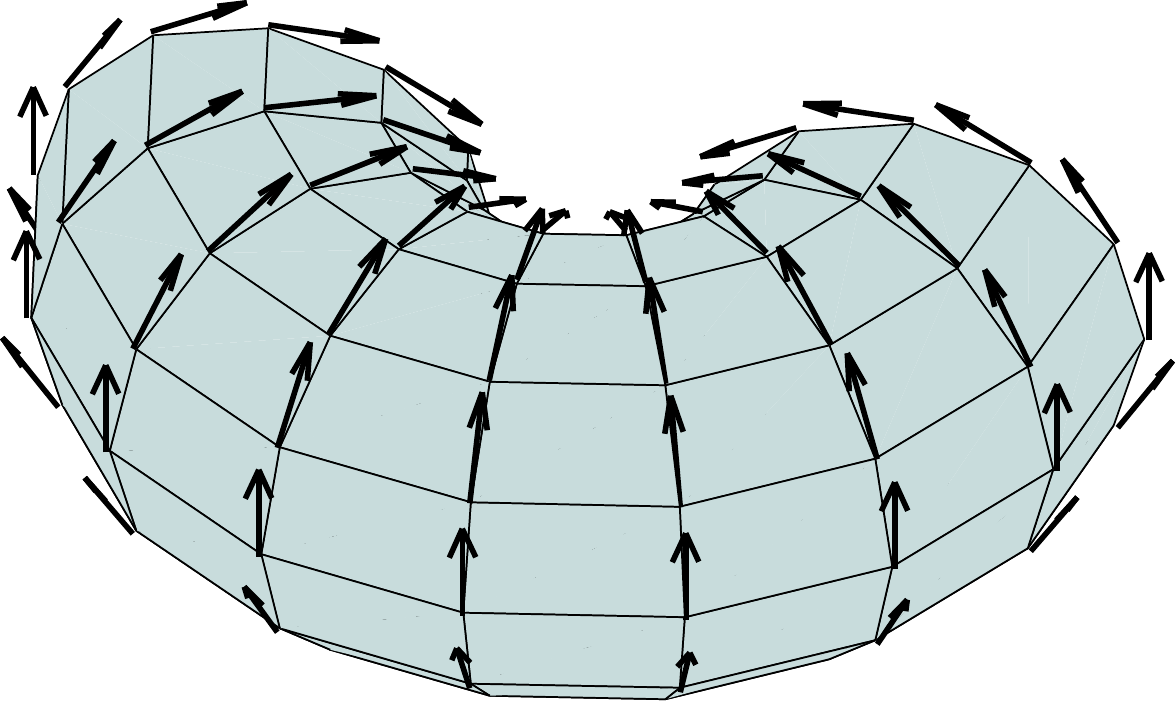} \
		\includegraphics[height=2cm]{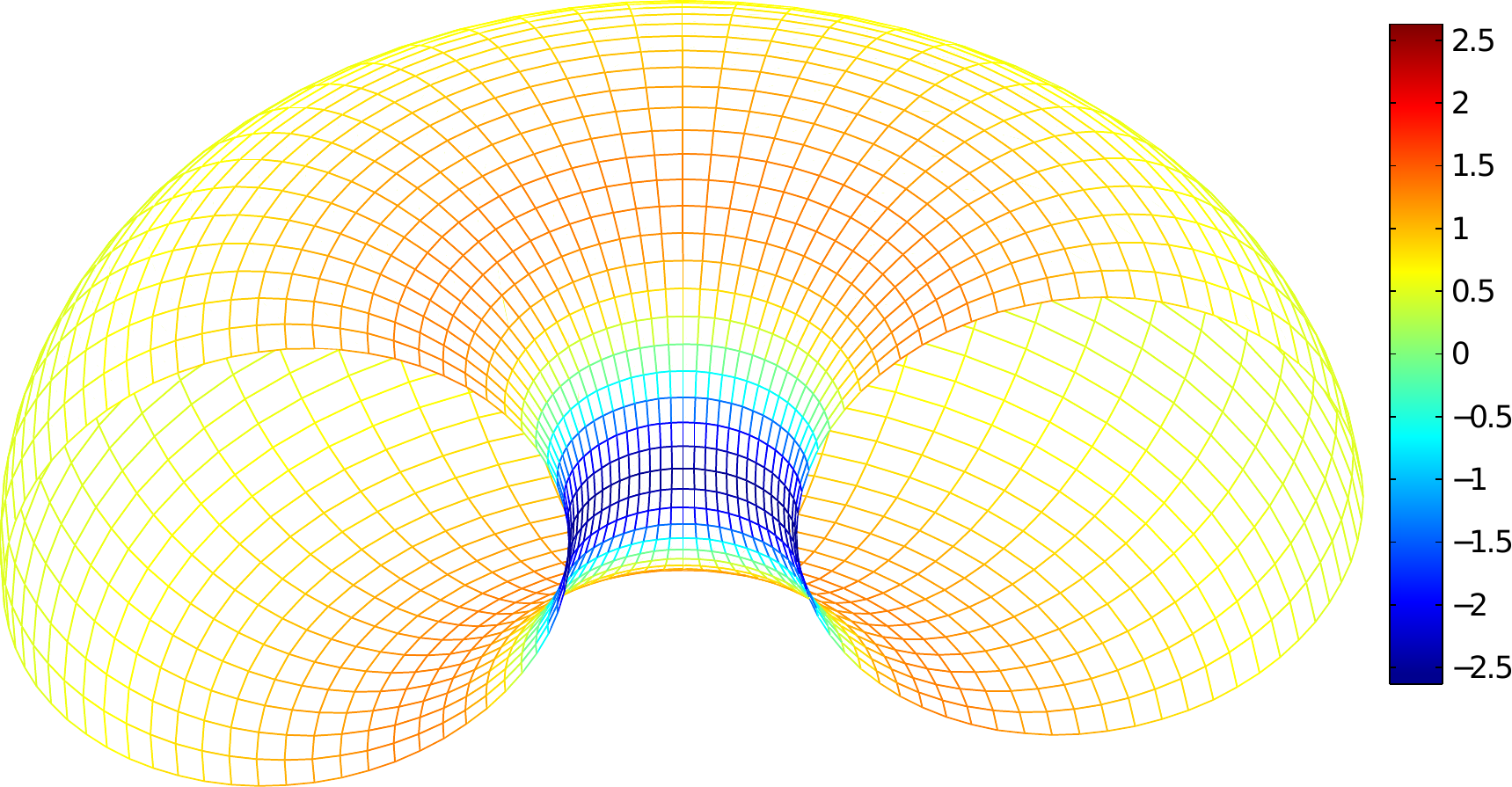} \
			\labellist
			\hair 2pt
			\pinlabel $\alpha_p$ at 330 -10
		\endlabellist				
		\includegraphics[height=2cm]{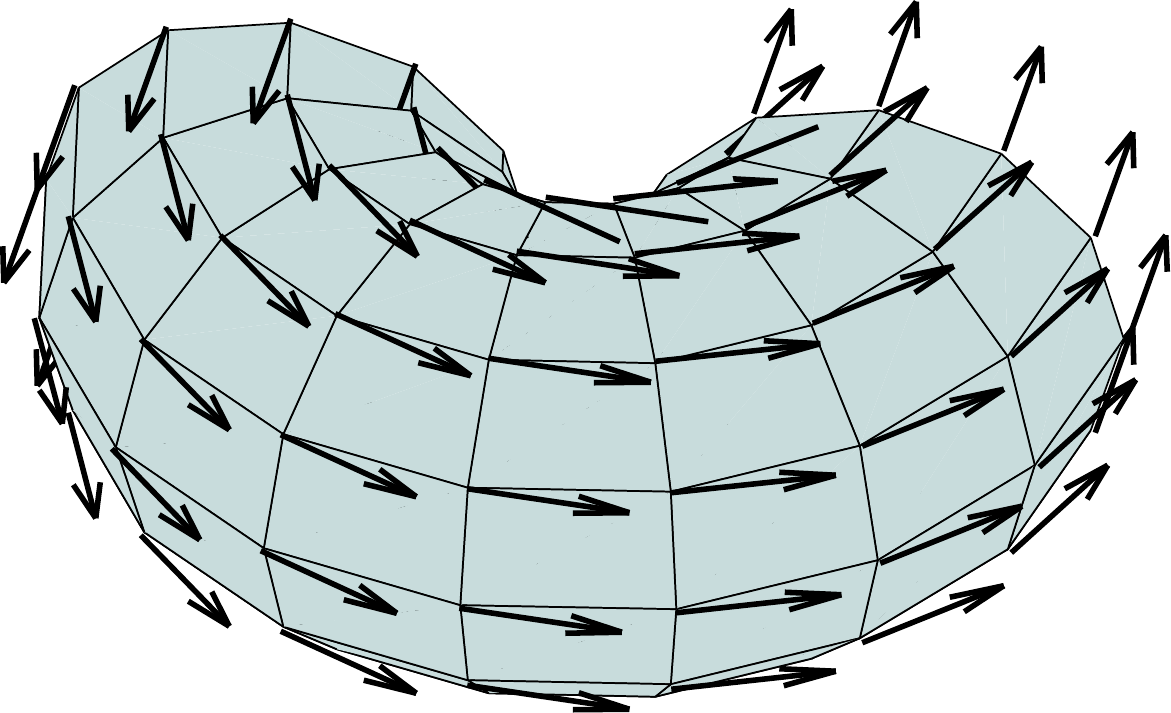} \
		\includegraphics[height=2cm]{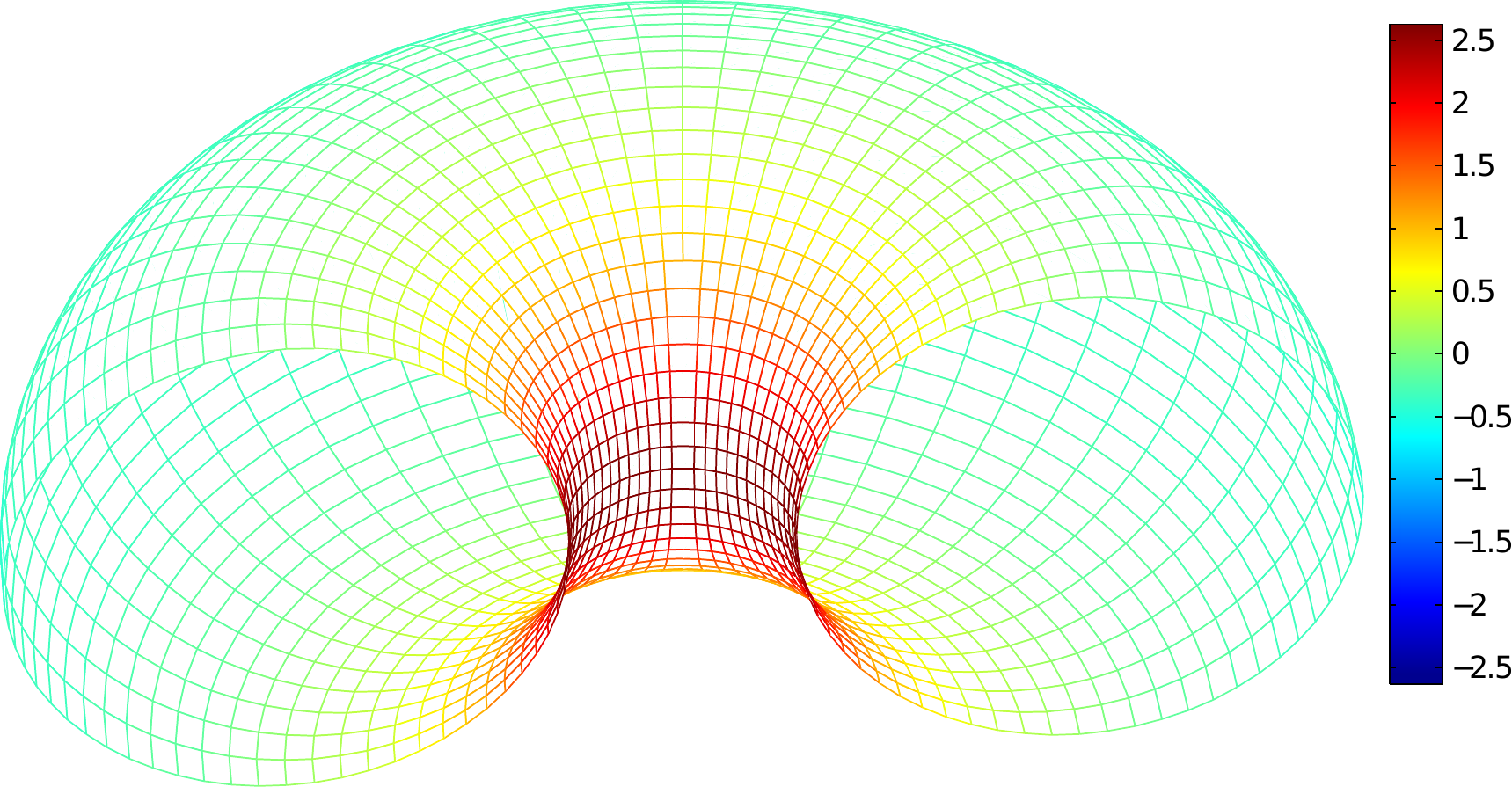}
\caption{The constant states $\alpha_m\equiv 0$ (director oriented along the meridians of the torus), $\alpha_p\equiv\pi/2$ (director oriented along the parallels of the torus) and their respective energy densities.}
\label{fig:extra}		
\end{figure}

In Subsection \ref{ssec:oneconst} we show a correspondence between elements of the fundamental group of the torus ($\Z \times \Z$), classes of functions $\alpha$ with the same boundary conditions, and classes of vector fields $\n$ with the same winding number. In view of this decomposition, in Theorem \ref{Th:gf} we prove that the Euler-Lagrange equation of \eqref{eq:walpha_intro} has a solution for every element of $\Z \times \Z$, and that for every (regular enough) initial datum $\alpha_0$ in a class with fixed boundary conditions, the $L^2$-gradient flow of \eqref{eq:walpha_intro} has a unique classical solution, which converges to a solution of the E.-L. equation as $t\to \infty$.

\subsection{Structure of the paper} In Section \ref{sec:prelimin} we introduce
 the differential geometry notation and tools that we need for our study. In Section \ref{sec:energy} we describe and contextualize the
 energies \eqref{eq:napoli_intro} and \eqref{eq:napolif_intro}. In Section \ref{sec:framework}
 we set up the functional framework and we state the $H^1$-version of the hairy ball Theorem (Theorem \ref{Th:H1empty}).
The existence of minimizers (Proposition \ref{prop:W}) and the gradient flow dynamics (Theorem \ref{Th:gfn})
 on general two-dimensional embedded surfaces are proved in Section \ref{sec:ex_min}. 
 Section \ref{sec:alpha} is devoted to the existence of global $H^1$-representations $\alpha$ (Proposition \ref{prop:alp}), which we use in Subsection \ref{ssec:formulas} to express the energies in terms of $\alpha$ and in Subsection \ref{ssec:proofempty} to prove Theorem \ref{Th:H1empty}.
In Section \ref{sec:torus} we concentrate on the particular case of an axisymmetric parametrised torus. After a short revision of the minimization problem on constant deviations $\alpha$ (Subsection \ref{ssec:constantalpha}), we state and prove the results concerning the correspondence between homotopy classes of the torus, solutions of the Euler-Lagrange equations, and gradient flows (Theorem \ref{Th:gf}). Numerical approximations of these solutions, obtained by evolving the discretized gradient flow, are presented in Section \ref{sec:numerics}. The appendices contain the computations regarding the explicit parametrization of the torus and the derivation of the Euler-Lagrange equation for the full energy \eqref{eq:napolif_intro}, in terms of $\alpha$.


\section{Differential Geometry Preliminaries}
\label{sec:prelimin}

We refer the reader to, e.g., \cite{Lee2013}, for all the material regarding Riemannian geometry. Let $\Sigma\subset \Rt$ be
 an embedded oriented regular surface of $\Rt$. 
We assume that 
$\Sigma$ is compact, connected, smooth and without boundary. For any point $x\in \Sigma$, let 
$T_x\Sigma$ and $N_x\Sigma$ denote the tangent and the normal space to $\Sigma$ in the point $x$, respectively. 
We denote with $\nb$ the inner unit normal field. Let $T\Sigma$ 
denote the tangent bundle of $\Sigma$, i.e.
the (disjoint) union over $\Sigma$ of the tangent planes $T_x\Sigma$. 
Let $\pi:T\Sigma\to \Sigma$ be the (smooth) map that assigns
to any tangent vector its application point on $\Sigma$. 
A vector field $\n$ on a open neighbourhood $A\subset\Sigma$, is a section of $T\Sigma$, i.e.
a map $\n:A\to T\Sigma$ for which $\pi\circ\n$ is the identity on $\Sigma$. 
We denote by $\mathfrak{T}(\Sigma)$ the space
of all the smooth sections of $T\Sigma$. 
For any point $x\in \Sigma$ let $T_x^{*}\Sigma=(T_x\Sigma)^{*}$
be the dual space of $T_x\Sigma$, also named cotangent space. 
Its elements are called covectors. The 
disjoint union over $\Sigma$ of the cotangent spaces $T_x^{*}\Sigma$
is $T^{*}\Sigma$. As we did for vector fields,
we introduce the space of smooth sections of $T^{*}\Sigma$. We denote this space
by $\mathfrak{T}^{*}(\Sigma)$, its elements are the covector fields. 
We denote by $g$ the metric induced on $\Sigma$ by the embedding, i.e. the restriction
of the metric of $\Rt$ to tangent vectors to $\Sigma$.
As a consequence, we can unambiguously
use the inner product notation $(\ub, \vb)_{\Rt}$ instead of $g(\ub,\vb)$ for $\ub,\vb\in T_x\Sigma$, 
$x\in \Sigma$. Similarly, we write $\vert \ub\vert  = \sqrt{(\ub, \ub)_{\Rt}}$ to denote
the norm of a tangent vector $\ub$ to $\Sigma$.  For a two-tensor $\mathbb A=\{a_{i}^j\}$ 
we adopt the norm $|\mathbb A |^2:= \text{tr}(\mathbb A^T \mathbb A)=\sum_{ij}(a_{i}^j)^2$, 
We denote with $\mathbb{A}:\mathbb{B}$ the corresponding scalar product
between the tensors $\mathbb{A}$ and $\mathbb{B}$.
If $\{\eu, \ev\}$ is any local frame for $T\Sigma$, we denote by $g_{ij}=g(\ei,\ej)=(\ei,\ej)_{\R^3}$ 
the components of the metric tensor with respect to $\{\eu, \ev\}$. By $g^{ij}$ and $\bar g$ we denote
the components of the inverse $g^{-1}$ and the determinant of $g$, respectively. 
As it is customary, if $(x^1,x^2)$ is a coordinate system 
for $\Sigma$, then $(\frac{\partial}{\partial x^1}, \frac{\partial}{\partial x^2})$ is the corresponding
local basis for $T\Sigma$ and $(\d x^1,\d x^2)$ is the dual
basis. Given a vector $X$, we denote by $X^\flat$ the covector
such that $X^\flat(\vb) = g(X, \vb)$. In coordinates,  
\[
	X^\flat = X^\flat_{i} \d x^i,\quad \hbox{with}\quad X^\flat_{i} = g_{ij}X^j.   
\]	
Being the flat $^\flat$ operator invertible, we denote 
by the sharp $^\sharp$ symbol its inverse, which acts in the following way: Given a covector $\omega$, let $\omega^\sharp$
be the vector such that $\omega(\vb) = g(\omega^{\sharp},\vb)$. In coordinates,
we have 
\[
	\omega^{\sharp} = (\omega^{\sharp})^i \frac{\partial}{\partial x^i},
	\quad \hbox{with}\quad  (\omega^{\sharp})^i = g^{ij}\omega_j. 
\]

In the formulae above and in the rest of the paper
we use Einstein summation convention: repeated upper and
lower indices will automatically be summed unless otherwise specified.  
In particular, indices with greek letters are summed from $1$ to $3$, while
latin ones are summed from $1$ to $2$. 
\medskip

\paragraph{\textbf{Differential Operators}}
Let $\nabla$ be the connection with respect to the standard metric of $\Rt$,
i.e., given two smooth vector fields $Y$ and $X$ in $\R^3$ (identified with its tangent space),
the vector field $\nabla_{X}Y$ is the vector field whose components are
the directional derivatives of the components of $Y$ in the direction $X$. 
When $\eb_\alpha$ ($\alpha=1,\ldots,3$) is a basis of $\Rt$ we will set $\nabla_\alpha Y := \nabla_{\eb_\alpha} Y$.
 Given $\ub$ and $\vb$ in $\mathfrak{T}(\Sigma)$,  
 we denote with
 $D_{\vb}\ub$ the covariant derivative of $\ub$ in the direction $\vb$,
  with respect to the Levi Civita (or Riemannian) connection $D$ of the metric $g$ on $\Sigma$. 
 Now, if
$\ub$ and $\vb$ are extended arbitrarily to smooth vector fields on $\RN$, we have the Gauss Formula along $\Sigma$: 
\begin{equation}
\label{eq:gauss}
\nabla_{\vb} \ub = D_{\vb} \ub + h(\ub,\vb)\nb.
\end{equation}
 In the relation above, the symmetric bilinear form $h:\mathfrak{T}(\Sigma)\times \mathfrak{T}(\Sigma)\longrightarrow \mathbb{R}$
  is the \emph{scalar second fundamental form} of $\Sigma$. Associated to $h$, there is a linear self adjoint operator, called \emph{shape operator} and denoted with $\B:\mathfrak{T} (\Sigma) \longrightarrow \mathfrak{T} (\Sigma)$, such that $\B \vb = -\nabla_{\vb} \nb$ for any $\vb\in  \mathfrak{T}(\Sigma)$. We recall that the operator $\B$ satisfies the Weingarten relation
 \begin{equation*}
	 (\B \ub,\vb)_{\Rt} = h(\ub,\vb) \,\,\,\,\forall \ub,\vb \in \mathfrak{T}(\Sigma).
 \end{equation*}
 Beside the covariant derivative, we introduce another differential operator for vector fields on $\Sigma$,
 which takes into account also the way that $\Sigma$ embeds in $\Rt$. 
Let $\ub\in \mathfrak{T}(\Sigma)$ and extend it smoothly to a vector field $\tilde \ub$ on $\Rt$; denote its standard gradient by $\nabla \tilde \ub$ on $\Rt$. For $x\in \Sigma$, define the \emph{surface gradient} of $\ub$
 \begin{equation}
 \label{eq:surface_grad_vector}
		\nabla_s \ub(x):=\nabla \tilde \ub(x)P(x),
\end{equation}
	where $P(x):=(Id-\nb \otimes \nb)(x)$ is the orthogonal projection on $T_x \Sigma$. Note that $\nabla_s \ub$ is well-defined, as it does not depend on the particular extension $\tilde \ub$. The object just defined is a smooth mapping $\nabla_s \ub: \Sigma \to \R^{3 \times 3}$, or equivalently $\nabla_s \ub: \Sigma \to \mathcal L(\R^3,\R^3)$ (the space of linear continuous operators on $\R^3$), such that $N_x \Sigma \subseteq \ker \nabla_s \ub(x)$, for all $x \in \Sigma$. In general, $\nabla_s \ub \neq D\ub = P(\nabla \ub)$
 since the matrix product is non commutative. Using the decomposition \eqref{eq:gauss}, it is immediate
to get
\begin{equation*}
	\nabla_s \ub[\vb]= \nabla_{\vb} \ub 
		= D_{\vb}\ub + h(\ub,\vb)\nb,\quad \forall \vb\in T_x \Sigma, \forall x\in \Sigma,
\end{equation*} 
which gives, recalling that the decomposition is orthogonal, 
\begin{equation}
\label{eq:normgrad}
	|\nabla_s \ub |^2 = | D\ub |^2 + |\B \ub |^2,\quad \forall \ub\in T_x \Sigma, \forall x\in \Sigma.
\end{equation}
Having defined $\nabla_s\ub$, we can introduce the related notions of divergence and curl
\begin{equation*}
	\hbox{tr}_g\nabla_s\ub=\hbox{tr}_gD\ub =:\divs \ub,\,\,\,\hbox{ in coordinates, }\divs \ub= \frac{1}{\sqrt{\bar g}} \frac{\partial }{\partial x^i}\left(\sqrt{\bar g}\, \ub^i \right)
\end{equation*}
and $\curls \ub:= -\epsilon \nabla_s \ub,$ where $\epsilon$ is the Ricci alternator: 
\[
	\epsilon_{\alpha \beta \gamma}=
\left\{
	\begin{array}{ll}
		0 	& \text{if any of $\alpha,\beta,\gamma$ are the same,}\\
		+1	& \text{if $(\alpha,\beta,\gamma)$ is a cyclic permutation of (1,2,3),}\\
		-1 	& \text{otherwise.}
	\end{array}
\right.		
\]
Note that the trace operator in the definition of the divergence acts only on 
tangential directions. Moreover, note that, contrary to the so-called covariant
curl (denoted with $\curl_{\Sigma}$, see \cite{Lee2013}) the surface $\curls$ defined above has, unless
the surface $\Sigma$ is a plane, 
also in-plane components. To see this, we introduce the Darboux orthonormal frame (or Darboux trihedron)
$(\n, \tbf, \nb)$, where $\tbf = \nb \times \n$. Let $
\kappa_\n ,\kappa_\tbf$ be the geodesic curvatures of the flux lines of
 $\n$ and $\tbf$, defined as $\kappa_\n:= 	(D_\n \n,\tbf)_{\R^3}$, $\kappa_\tbf:= 	-(D_\tbf \tbf,\n)_{\R^3}$, respectively;
 let $c_\n:=(\B \n,\n)_{\R^3}$ be the normal curvature and let 
$\tau_\n = -(\B\n,\tbf)_{\R^3}$ be the geodesic torsion of the flux lines of $\n$ (see, e.g., \cite{DoCarmo76}). The surface gradient of $\n$, with respect to the Darboux frame, has the simple expression (see, e.g., \cite{Rosso2003})
\begin{equation*}
	\nabla_s \n = 
	\begin{pmatrix}
		0 & 0 & 0 \\
		\kappa_\n & \kappa_\tbf & 0 \\
		c_\n & -\tau_\n & 0 
	\end{pmatrix},
\end{equation*}
from which we read
\begin{equation}
\label{eq:div-curl2}
	\divs \n = \kappa_\tbf\quad \hbox{and}\quad
	\curls \n = -\tau_\n \n - c_\n \tbf + \kappa_\n \nb.
\end{equation}
On the other hand, also the norm of the covariant derivative $D\n$ can 
be expressed in terms of the geodesic curvatures $\kappa_\tbf$
and $\kappa_\n$ as $\vert D\n\vert^2 = \kappa_\tbf^{2} + \kappa_\n^{2}$.
As a result, we have the following useful expression 
\begin{equation}
\label{eq:uc}
(\divs \n)^2 + (\n \cdot \curls \n)^2 + |\n \times \curls \n|^2 = (\divs \n)^2 + |\curls \n|^2 = \kappa_\tbf^{2} + \kappa_\n^{2} +\tau_\n^2 +c_\n^2 =|\nabla_s \n|^2.
\end{equation}

For a smooth scalar function $f:\Sigma \to \R$, with differential application $\d f_x:T_x \Sigma \to T_{f(x)}\R\simeq \R$, we introduce its gradient as
$\grads f = \d f^{\sharp}$, that is, the vector field such that
\[
 	\d f(X) = g(\grads f,X)\quad \hbox{for all } X\in T\Sigma. 
\]	
Since for scalar functions the expressions of $\grads f$ and $\nabla_s f$ coincide, in what follows we replace $\grads$ with the more common notation $\nabla_s f$. In coordinates, denoting $X=X^i \frac{\partial}{\partial x^i}$, the above relation means
\begin{equation*}
	\nabla_s f := g^{ij}\frac{\partial f}{\partial x^j}\frac{\partial}{\partial x^i}.
\end{equation*}
The Laplace Beltrami operator on $\Sigma$ is given by 
\begin{equation*} 
	\Delta := \divs \circ \nabla_s 
	= \frac{1}{\sqrt{\bar g}} \frac{\partial}{\partial x^i}\left(\sqrt{\bar g}g^{ij}\frac{\partial}{\partial x^j} \right).
\end{equation*}
We denote with $\dvolg$ the volume form
of $\Sigma$ (see, e.g., \cite{Lee2013}). 
We recall the following integration by parts formula ($f$ and $h$ are smooth functions on $\Sigma$)
\begin{equation}
\label{eq:int_part}
	-\int_{\Sigma}\Delta f \,h\, \dvolg 
	= \int_{\Sigma}g(\nabla_s f, \nabla_s h)\,\dvolg -\int_{\partial\Sigma}h \,\d f(N)\,\d S', 
\end{equation}
where $f$ and $h$ are smooth functions on $\Sigma$
and $\d S'$
is the element of length of the induced metric on $\partial\Sigma$. For a smooth vector field $\n\in \mathfrak{T}(\Sigma)$, we denote with 
$D^2 \n$ the double covariant derivative of $\n$, i.e. the following tensor
field
\[ 
	D^2\n(X,Y):= D_X(D_Y\n) - D_{D_{X}Y}\n\quad \hbox{for } X, Y\in \mathfrak{T}(\Sigma).
\]	
If $X=\frac{\partial}{\partial x_i}$ and $Y=\frac{\partial}{\partial x_j}$, we set $D^2_{ij}\n := D^2\n(\frac{\partial}{\partial x_i},\frac{\partial}{\partial x_j})$.
Then, 
we denote
with $\Delta_g \n$ the rough laplacian of $\n$, namely the 
vector field defined as
\begin{equation*}
	\Delta_g \n: =g^{ij}(D^2_{ij}\n) = g^{ij}D_{i}(D_j \n) - g^{ij}D_{D_{i}\!\!\!\,\,\frac{\partial}{\partial x_j}}\n.
\end{equation*}
In particular, 
 in a local orthonormal frame $\{\eu,\ev \}$, we have that
\[	
	\Delta_ g \n = \delta^{ij}D_{i}(D_j \n) - \delta^{ij}D_{D_{i}\ej}\n.
\]	
Note that $\Delta_g$ can be expressed in divergence form as $ \Delta_g \n = \hbox{div}_sD\n.$ 
In the flat case, the rough laplacian reduces to the componentwise 
laplacian of $\n$.


\section{Energetics} 
\label{sec:energy}
Let $\Omega\subset \R^3$ be the volume occupied by the crystal and let $\mathbb S^2\subset \R^3$ be the
unit sphere. In the framework of the director theory for nematic liquid crystals,
the configurations of the crystal may be described in terms of the optical
axis, a unit vector field $\n:\Omega \to \mathbb S^2$. A widely used model 
for nematic liquid crystals is the Oseen, Zocher and Frank (OZF) model (see, e.g., \cite{Virga94}), which is
based on the energy
\begin{equation}
\label{eq:energy}
	\begin{split}
		W^{OZF}(\n,\Omega):=\frac 12 \int_\Omega \left[K_1 (\Div \n)^2 
			+ K_2 (\n \cdot \curl \n)^2 +K_3 |\n \times \curl \n|^2\right. \\ 
		\left.+ (K_2 + K_{24})\Div[(\nabla \n) \n -(\Div \n)\n]\right]\,\d x,
	\end{split}	
\end{equation}
 where $K_1$, $K_2$, $K_3$ and $K_{24}$ 
 are positive constants called the splay, twist, bend and saddle-splay moduli, 
 respectively. In what follows, we generally omit the dependence of the energy on the domain. 
 A well-studied case is the so-called \emph{one-constant approximation}, obtained when
  the three constants $K_i$ are equal. In this case \eqref{eq:energy} reduces to
\begin{equation}
\label{eq:energyoc}
	 W^{OZF}_\kappa(\n):=\frac \kappa2 \int_\Omega |\nabla \n|^2\,\d x.
\end{equation} 
This model (both in the general case and the one-constant approximation)
has received considerable attention from the mathematical community. Among the others,
we refer to \cite{Bre_Co_Lieb}, \cite{HKL}.  As it is apparent from the energy \eqref{eq:energyoc},
the analysis of liquid crystals shares some difficulties with the theory of harmonic maps 
into spheres
(see, e.g., \cite{Bre_Co_Lieb}).
More precisely, the study of \eqref{eq:energy} and \eqref{eq:energyoc}
has to face possible topological
obstructions coming from the choice of the boundary conditions.
In particular, choices of the boundary data not satisfying
proper topological constraint lead to the formation of singularities,
named defects, in the director (see \cite{HKL}).

In this paper, we study nematic liquid crystals which are constrained on a surface
$\Sigma\subset \R^3$. We describe their behaviour via a unit norm vector field $\n$ 
tangent to $\Sigma$, that is $\n(x) \in T_x\Sigma$, for $x\in \Sigma$. As in the 
 three-dimensional theory of Oseen, Zocher, and Frank (OZF), the director $\n$
describes the preferred direction of the molecular alignment (which coincides
with the optical axis of the molecule). In all classical models of surface free energy for nematics, the derivatives in \eqref{eq:energyoc} are replaced by the covariant derivative $D\n$ of the surface $\Sigma$ (see \cite{LubPro92,Straley71,TuSei07,VitNel06}, and \cite{Shk02}, where the full hydrodynamic model is considered). 
Consequently, the surface energy in the one-constant approximation is
\begin{equation}
\label{eq:energy2doc}
	 W^{in}_\kappa(\n):=\frac \kappa2 \int_\Sigma |D\n|^2\,\dvolg
\end{equation}
and the surface energy in its full generality is 
\begin{equation}
\label{eq:energy2d}
	W^{in}(\n):=\frac 12 \int_\Sigma K_1 (\divs \n)^2  +K_3 |\curl_{\Sigma} \n|^2 \dvolg,
\end{equation}
where $\curl_{\Sigma}$ is the covariant curl (see \cite{Lee2013}). We adopt the superscript `\emph{in}' in \eqref{eq:energy2doc} and \eqref{eq:energy2d}, referring to the \emph{intrinsic} character of this energy.
A recent approach \cite{NapVer12L,NapVer12E} takes into account also the effects of extrinsic curvature in the deviations of the director. The energy in this case is
\begin{equation}
\label{eq:napoli}
	W(\n):=\frac 12 \int_\Sigma K_1 (\divs \n)^2 + K_2 (\n \cdot \curls \n)^2 +K_3 |\n \times \curls \n|^2\,\dvolg.
\end{equation}
To have a better insight on the extrinsic/intrinsic character of this energy, let us focus on the 
one-constant approximation of \eqref{eq:napoli},  which is given by
\begin{equation}
\label{eq:napolioc}
	 W_\kappa(\n):=\frac \kappa2 \int_\Sigma |\nabla_s \n|^2\,\dvolg,
\end{equation}
where $\nabla_s $ is the operator introduced in \eqref{eq:surface_grad_vector}. Now, 
thanks to \eqref{eq:normgrad} we have 
\begin{equation}
\label{eq:napoliocbis}
	W_\kappa(\n) = \frac \kappa2 \int_\Sigma
	| D\n|^2 + |\B \n|^2\,\dvolg
\end{equation}
which shows a striking difference between the classical energy \eqref{eq:energy2doc} and the newly 
proposed \eqref{eq:napolioc}, namely the presence of the extrinsic term $\B\n$. 
This term takes into account how the surface $\Sigma$, which models the thin substrate on which the liquid crystal is smeared, 
is embedded into the three-dimensional space. The energy \eqref{eq:napoli}
has been derived in \cite{NapVer12L,NapVer12E} starting from the well established
Oseen and Frank's energy $W^{OZF}$ \eqref{eq:energy}. More precisely,
starting from a tubular neighborhood $\Sigma_h$ of thickness $h$ (satisfying a 
suitable constraint related to the curvature of $\Sigma$), 
Napoli and Vergori in \cite{NapVer12L,NapVer12E} obtain that $W(\n)$ in \eqref{eq:napoli} is given by
\begin{equation*}
	W(\n) = \lim_{h\searrow 0} \frac 1h W^{OZF}(\n,\Sigma_h)
\end{equation*}
The limit above holds for any fixed and sufficiently smooth field $\n$ with
the property of being independent of the thickness direction
and tangent to any inner surface of the foliation $\Sigma_h$. As a result,
the null lagrangian related to the coefficient $(K_2+K_{24})$ 
in \eqref{eq:energy} disappears in the limit procedure and hence it
is not considered in \eqref{eq:napoli}.  
More generally, 
the rigorous justification of the surface energy \eqref{eq:napoli} in terms of variational methods 
is an intriguing and interesting problem. There are at least 
two possible approaches: the first one is to directly justify the dimensional reduction in terms 
of Gamma convergence (in the spirit of \cite{ledret_raoult}) 
while the second one consists in obtaining \eqref{eq:napoli} 
starting from a discrete lattice model (cf., e.g., \cite{BCSarxiv}).
In the one constant approximation regime, the rigorous justification of \eqref{eq:napolioc} 
in terms of dimensional reduction and in terms of a micro/macro transition is contained in \cite{LM}. 
It is an open problem to rigorously justify the full energy.


\section{Functional Framework}
\label{sec:framework}
In this Section we introduce the functional framework where to set
the problem. As it will be clear in a moment (see Theorem \ref{Th:H1empty}),
the choice of our functional 
setting reflects the topology of the shell. In particular, we will restrict to surfaces
for which the Poincar\'e - Hopf index Theorem does not force the vector field to 
have defects. 

Here integration is always with respect to the area form of the metric $g$ induced on $\Sigma$ by the euclidean metric of $\R^3$. Let $L^2(\Sigma)$ and  $L^2(\Sigma;\R^3)$ be the standard Lebesgue spaces of square-integrable scalar functions and vector fields, respectively. 
Define the spaces of tangent vector fields
\[
	\Lt:=\left\{ \ub \in L^2(\Sigma;\R^3) : \ub(x)\in T_x\Sigma \hbox{ a.e.}\right\} \quad \text{and}\quad
	\Ht:=\left\{\ub \in L^2_{tan}(\Sigma) : |D\ub|\in L^2(\Sigma) \right\}.
\]	
The latter, endowed with the scalar product
\[
	(\ub,\vb)_{H^1}:= \left( \int_\Sigma \left\{
	 \text{tr}(D\ub^TD\vb) 
	+(\ub,\vb)_{\R^3}\right\} \dvolg \right)^{1/2}
\]
is a separable Hilbert space. Let $\|\ub\|_{H^1}:= \sqrt{(\ub,\ub)_{H^1}}$. We can define another norm by
\[
	\| \ub\|_{H^1_s}:=\left( \int_\Sigma \left\{ | \nabla_s\ub |^2 + |\ub|^2\right\}\dvolg\right)^{1/2}.
\]
Let $\lambda_M$ denote the maximum value attained by the eigenvalues of the shape operator $\B$ on $\Sigma$. Since
\[
 	| D\ub |^2 + (\lambda_M^2+1)  |\ub|^2 \geq | D\ub |^2 + |\B\ub|^2+ |\ub|^2 \geq | D\ub |^2 + |\ub|^2,
\]
by  \eqref{eq:normgrad}  the two norms are equivalent:
\[
	(\lambda_M^2+1)^{1/2} \|\ub\|_{H^1} \geq \|\ub\|_{H^1_s} \geq \|\ub\|_{H^1}.
\]
Finally, the ambient space for the directors $\n$ is defined as
\[
	\Hut :=\left\{\ub \in \Ht : \vert\ub\vert =1 \hbox{ a.e.}  \right\}.
\]	
Since for $\ub \in \Hut$
\[
	\|\ub\|^2_{H^1_s} = \frac{2}{\kappa}W_\kappa(\ub) + \textup{Vol}(\Sigma),
\]
it will often be useful to adopt $\|\cdot\|_{H^1_s}$ instead of $\|\cdot\|_{H^1}$. Note that $\Hut\subset L^2(\Sigma;\R^3)$ with compact embedding, and thus $\Hut$ is a weakly closed subset of  $\Ht$. Note also that $\Hut$ lacks a linear structure, while $\Ht$ is a linear space. 

There are two major problems to address before discussing the existence of minimizers
of \eqref{eq:napoli}:

\begin{itemize}
\item Choice of the topology of the surface $\Sigma$. This is related to the choice
of the functional space thanks to the $H^1$-version of the hairy ball Theorem 
(see the next Theorem \ref{Th:H1empty}).  
\item Choice of the boundary conditions. 
Given a boundary datum $\n_b$ in some functional class,
  we have to show that the set of competitors $\mathcal{A}(\n_{b})$
is not empty, where 
\begin{equation*}
	\mathcal{A}(\n_{b}):=\left\{ \ub \in \Hut: \ub = \n_{b} \hbox{ on } \partial \Sigma \right\}.
\end{equation*}
This fact is related to some precise topological properties of $\n_b$, which we analyze in \cite{CSV13}.
\end{itemize}
Here, we restrict to the case of a smooth surface without boundary. In this context, 
the following $H^1$ form of the classical hairy ball Theorem, clarifies the situation.

\begin{theorem}
\label{Th:H1empty}
Let $\Sigma$ be a compact smooth surface without boundary, embedded in $\R^3$.
Let $\chi(\Sigma)$ be the Euler characteristic of $\Sigma$. Then 
\[
	\Hut \neq \emptyset \,\,\Leftrightarrow \chi(\Sigma) = 0.
\]
\end{theorem}

According to the above Theorem, 
in Section \ref{sec:ex_min} we will make this basic topological assumption:
\begin{equation}
\label{eq:main_ass}
\Sigma 
\hbox{ is a compact and smooth two-dimensional 
surface without boundary, with $\chi(\Sigma) = 0$.}
\end{equation}

We postpone the proof of the above result to the end of Section \ref{sec:alpha}.
In particular, we have that the two-dimensional sphere cannot 
be combed with $H^1$-regular vector fields. 
On the other hand, the above Theorem (as well as its smooth classical
counterpart) does not hold for odd-dimensional spheres as the following
example shows. Take $x=(x_1,\ldots, x_{2N})\in \mathbb{S}^{2N-1}$. The vector
field $\ub$ given by 
\[
\ub(x)=(x_2,-x_1,\ldots,x_{2i},-x_{2i-1},\ldots, x_{2N},-x_{2N-1})
\]
is smooth, tangent, and with unit norm.


\section{Existence of minimizers and gradient flow of the energy}
\label{sec:ex_min}

Now, we come to the question of existence of minimizers of the energy \eqref{eq:napoli}. Choosing $\Sigma$ satisfying \eqref{eq:main_ass}, namely in such a way that $\Hut \neq\emptyset$, we have the following 
(see \cite{HKL} for 
the flat case)

\begin{proposition}
\label{prop:W}
Let $\Sigma$ be a smooth, compact surface in $\R^3$, without boundary, 
satisfying \eqref{eq:main_ass}  and let $W:\Hut \to \R$ be the energy functional defined in \eqref{eq:napoli}.
Set $K_*:=\min\left\{K_1,K_2,K_3\right\}$ and $ K^*:= 3(K_1 + K_2 + K_3)$.
We have that
\[
 	\frac{ K_*}{2} \int_{\Sigma}\big(\vert D \ub(x) \vert^2 
		+ \vert\B \ub(x)\vert^2\big)\dvolg\le W(\ub)\le
 		\frac{ K^*}2 \int_{\Sigma}\big(\vert D \ub(x) \vert^2 + \vert\B \ub(x)\vert^2\big)\dvolg.
\]
Moreover, the energy $W$ is lower semicontinuous with respect to the weak convergence of $H^1(\Sigma;\R^3)$. 
\end{proposition}
\begin{proof}
The upper and the lower bound follow by the \emph{one-constant approximation} (see \eqref{eq:normgrad}) and the equality
\eqref{eq:uc}.
The lower semicontinuity
can be proved by noting that all the terms in  \eqref{eq:napoli} are
indeed weakly lower semicontinuous in $H^1(\Sigma;\R^3)$ and are multiplied by the positive constants $K_1, K_2$
and $K_3$.
\end{proof}
Thus, the existence of a minimizer of the energy $W$ follows from the direct method
of calculus of variations
\begin{proposition}
\label{propo:ex_min}
There exists $\n\in \Hut$ such that $W(\n) = \inf_{\ub \in \Hut} W(\ub)$.
\end{proposition}
\begin{proof}
Let $\ub_n$ be a minimizing sequence uniformly bounded in $\Hut$. This means
that  $|\ub_n|=1$ and that $\left\{\ub_n\right\}$ is uniformly bounded in $\Ht$.  
Thus, up to a not relabeled
subsequence of $n$, we have that there exists a vector field $\n\in \Ht$ with $|\n|=1$ such that
\[
	\ub_n \xrightarrow{n\nearrow +\infty} \n \hbox{ weakly in } \Hut  \hbox{ and strongly in } L^2(\Sigma).
\]
Thus, the lower semicontinuity of $W$ gives that
$ \inf_{\ub\in \mathcal A} W (\ub)=\liminf_{n\nearrow +\infty}W(\ub_n)\ge W(\n)$ which means that $\n$ is a 
minimizer for $W$. 

\end{proof}

Now, in the case of the one-constant approximation,
we compute the Euler Lagrange equation
associated to the minimization of \eqref{eq:napoliocbis}
 (see also). Incidentally, note that up
to technical modifications, the same computations are valid for an $(n-1)$-hypersurface in $\mathbb{R}^n$. 
Thus, let $\n\in \Hut$ be a minimizer for \eqref{eq:napoliocbis}.  Take a smooth
$\vb\in \Hut$ and consider the family of deformations 
$\varphi(t):=\frac{\n + t\vb}{\vert \n + t\vb\vert}$, for $t\in (0,1)$. Note that
$\vert\varphi\vert = 1$ by construction and that $\varphi\in  \Hut$. Moreover, 
$\varphi(0) = \n$ and $\dot\varphi(0) = \vb - (\vb,\n)\n$ and thus $ W_\kappa(\varphi(t))$ 
has a minimum at $t=0$.
Hence, we have
\begin{align*}
	0 = \frac{\d}{\d t}_{|t=0} W_\kappa (\varphi(t))&= \kappa\int_{\Sigma}
		(D\varphi(0),D\dot\varphi(0))_{\R^3}\dvolg + \kappa\int_{\Sigma}(\B\varphi(0),\B\dot\varphi(0))_{\R^3}\dvolg \\
	&=  \kappa\int_{\Sigma}(D \n,D\vb)_{\R^3}\dvolg + \kappa\int_{\Sigma}(\B \n,\B\vb)_{\R^3}\dvolg \\
	&  \quad -\kappa\int_{\Sigma} \vert D\n\vert^2(\n,\vb)_{\R^3} \dvolg 
		- \kappa\int_{\Sigma}\vert \B \n\vert^2 (\n,\vb)_{\R^3}\dvolg,
\end{align*}
where we have used that, being $\vert \n\vert = 1$, there holds that $(D\n,\n)_{\R^3}=0$, and the
fact that $\B[\n (\n,\vb)_{\R^3}] = -\nabla_{\n(\n,\vb)_{\R^3}}\nb = -(\n,\vb)_{\R^3}\nabla_{\n}\nb=(\n,\vb)_{\R^3}\B\n$.
Now, since the shape operator $\B$ is self-adjoint, we may introduce the operator
$\Bd$ given by 
\begin{equation*}
	(\Bd \ub, \vb)_{\R^3}:= (\B \ub,\B\vb)_{\R^3}\quad \hbox{for any } \ub, \vb \in \mathfrak T(\Sigma).
\end{equation*}
Thus, collecting all the computations, we obtain that a minimizer $\n$ of $W_\kappa$ 
is a solution of the following system of nonlinear partial differential equations
\begin{equation}
\label{eq:eul_general}
	-\Delta_g \n +  \Bd\n = \vert D\n\vert^2 \n + \vert \B\n\vert^2 \n\quad \hbox{in } \Sigma.
\end{equation}
Since the equations do not depend on $\kappa$, in the remainder of this section we take $\kappa=1$, but we still write $W_\kappa$, to tell the one-constant energy from the general $W$ with three constants.

\begin{rem} 
\label{rem:decomposition}
As it happens for harmonic maps, a vector field $\n$ solving \eqref{eq:eul_general}
is parallel to $ -\Delta_g \n + \Bd \n$. Viceversa, if $ -\Delta_g \n + \Bd \n$ is parallel to $\n$, 
then there exists a function $\lambda$ on $\Sigma$ (the Lagrange multiplier) such that 
\begin{equation*}
	 -\Delta_g \n + \Bd \n = \lambda \n,
\end{equation*}
from which it follows that (recall that $\vert \n\vert =1$)
\begin{equation*}
	\lambda = \lambda (\n,\n)_{\Rt} = (-\Delta_g \n,\n)_{\Rt} + (\Bd \n,\n)_{\Rt} = \vert D\n\vert^2 + \vert \B\n\vert^2,
\end{equation*}
where we have used 
the general identity
\begin{equation}
\label{eq:ortholan} 
	0 \stackrel{|\n|=1}{=}\Delta_g \vert \n\vert^2 = 2\left\{\vert D\n\vert ^2 + (\Delta_g \n, \n)_{\Rt}\right\},
\end{equation}	
holding for any smooth vector field $\n$ on $\Sigma$. Therefore, a smooth unitary vector field $\n \in \mathfrak T(\Sigma) $ is a solution of \eqref{eq:eul_general} if and only if it solves
\begin{equation}
\label{eq:ortho}
	(-\Delta_g \n + \Bd \n) \times \n   =0.
\end{equation}
\end{rem}

\paragraph{\textbf{Evolution of the energy \eqref{eq:napolioc}
}} 
In this paragraph, we study the $L^2$ gradient flow of the energy 
\eqref{eq:napolioc}, namely the following evolution
\begin{eqnarray}
\label{eq:gfn}
&  \partial_t \n -\Delta_g \n + \Bd \n = \vert D\n\vert^2 \n + \vert \B\n\vert^2 \n\quad \hbox{ a.e. in  } \Sigma \times (0,+\infty),\\
&  \n(0) = \n_0\quad  \hbox{ a.e. in } \Sigma.
\label{eq:initialn}
\end{eqnarray}

We make precise the definition of weak solution to \eqref{eq:gfn}.
\begin{definition}
\label{def:weak_sol}
$\n$ is a {\itshape global weak solution} to \eqref{eq:gfn} if 
\begin{align*}
	&\n\in L^\infty(0,+\infty;\Hut),\,\,\,\,\partial_t\n \in L^2(0,+\infty;\Lt),\\
	& \n \hbox{ weakly solves  \eqref{eq:gfn}, that is} 
\end{align*}
\end{definition}
\begin{equation}
\label{eq:sol_distr}
\int_\Sigma (\partial_t\n, \varphi)_{\Rt}\dvolg + \int_{\Sigma}(D\n,D\varphi)_{\Rt}\dvolg +
\int_{\Sigma} (\Bd \n - \vert D\n\vert^2 \n - \vert \B\n\vert^2 \n,\varphi)_{\Rt}\dvolg = 0,
\end{equation} 
for all $\varphi\in \Ht$.

We are going to prove the following Theorem.
\begin{theorem}
\label{Th:gfn}
Let $\Sigma$ be a two-dimensional surface satisfying \eqref{eq:main_ass}. Given 
$\n_0\in \Hut$ there exists a global weak solution to \eqref{eq:gfn} with $\n(\cdot,0)=\n_0(\cdot)$
in $\Sigma$. 
\end{theorem}

Note that equation \eqref{eq:gfn} has some similarities with 
the heat flow for harmonic maps and it offers similar difficulties. In particular,
the treatment of the quadratic terms in the right hand side requires some
care.
Note that these terms are related to the constraint 
$\n(x)\in \mathbb{S}^2$ for a.a. $x\in \Sigma$. As it happens in the study of the heat 
flow for harmonic maps (see \cite{chen, chen_struwe}), we relax this
constraint with a Ginzburg-Landau type
approximation, i.e.,  we allow for vectors $\n$ with $\vert\n\vert \neq 1$, but we penalize deviations from unitary length. The approximating equation is then obtained as the Euler-Lagrange equation of the \emph{unconstrained} functional
\begin{equation}
\label{def:ginzburg}
	\mathcal E_\eps: \Ht \to\R,\qquad \mathcal{E}_\eps(\vb):=W_{\kappa}(\vb) + \frac{1}{4\eps^2}\int_{\Sigma}(\vert \vb\vert^2-1)^2\dvolg.
\end{equation}
Thus, we approximate the solutions to \eqref{eq:gfn}-\eqref{eq:initialn}
with solutions of ($\eps$ is a small parameter intended to go to zero)
\begin{eqnarray}
\label{eq:gfn_approx}
&  \partial_t \n^\eps -\Delta_g \n^\eps + \Bd \n^\eps + \frac{1}{\eps^2}(\vert \n^\eps\vert ^2 - 1)\n^\eps  = 0\quad \hbox{ a.e. in  } \Sigma \times (0,+\infty),\\
&  \n^\eps(0) = \n_0\quad  \hbox{ a.e. in } \Sigma.
\label{eq:initialn_approx}
\end{eqnarray}
Existence of a global solution to \eqref{eq:gfn_approx}-\eqref{eq:initialn_approx},
with all the terms in $L^2(0,\infty;\Lt)$
follows from the time discretization procedure we are going to briefly outline. 

First of all we introduce a uniform partition $\mathcal P$ of $(0,+\infty)$, i.e.
\begin{equation*}
 \mathcal P:=\left\{ 0=t_0< t_1<t_2<\ldots<t_k<\ldots\right\}, \quad
	\tau:= t_{i}-t_{i-1}, \quad \lim_{k\nearrow +\infty} t_{k} = +\infty.
\end{equation*}
Now, starting from the initial value $\n_0$, we find an approximate solution 
$N_k \approx \n^\eps(t_k)$, $k=1,\ldots,$ by solving iteratively the following problem in
 the unknown $N_k$ (for notational simplicity, we skip for a while the $\eps$ dependence)
\begin{eqnarray}
\label{pb:timediscr}
\hbox{ Find } N_k \hbox{ with } N_k(x) \in T_x\Sigma \hbox{ for a.a. } x\in \Sigma \hbox{ such that }\nonumber\\
\frac{N_k - N_{k-1}}{\tau} -\Delta_g N_{k} + \Bd N_k + \frac{1}{\eps^2}(\vert N_k\vert ^2 - 1)N_k = 0 \hbox{ for a.a. } x\in \Sigma.
\end{eqnarray}
The existence of a solution to the above problem follows by minimization. More precisely, 
given \eqref{def:ginzburg}, it 
is not difficult to show that the solution of the iterative minimization problem
\begin{equation}
\label{pb:timediscr2}
\hbox{ Given } N_{k-1}\in \Ht, \hbox{ find }N_k \in \argmin_{\vb \in \Ht }
\left\{\frac{\vert \vb- N_{k-1}\vert^2}{\tau} +  \mathcal{E}_{\eps}(\vb) \right\},
\end{equation} 
is a solution to \eqref{pb:timediscr}. Problem \eqref{pb:timediscr2} can be easily solved 
using the direct method of calculus of variations as we did in Proposition \ref{prop:W}.
Subsequently, we introduce the piecewise linear $(\hat N_\tau)$ and the piecewise
constant ($\bar N_\tau$) interpolants of the discrete values $\left\{N_k\right\}_{k\ge 1}$. Namely, 
given $\n_0,N_1,\ldots,N_k,\ldots,$ we set
\begin{align*}
	\hat N_{\tau}(0)&:=\n_0,  \qquad \hat{N}_{\tau}(t):=a_k(t) N_{k} + (1-a_k(t))N_{k-1},\\
	\bar N_\tau(0)&:=\n_0,  \qquad \bar{N}_{\tau}(t):=N_k\quad \hbox{ for } t\in ((k-1)\tau,k\tau],\ k\ge 1,
\end{align*}
where $a_k(t):=(t-(k-1)\tau)/\tau$ for $t\in ((k-1)\tau,k\tau]$, $k\ge 1$.
Note that, for almost any $(x,t)\in \Sigma\times (0,+\infty)$, we have that 
$\hat{N}_\tau\in T_x\Sigma$ and  $\bar{N}_\tau\in T_x\Sigma$, being $T_x\Sigma$
a linear space for any fixed $x\in \Sigma$. 
Hence, we can rewrite \eqref{pb:timediscr} in the form
\begin{equation}
\label{eq:discr2}
\partial_t \hat{N}_\tau -\Delta_g \bar{N}_\tau + \Bd\bar N_\tau +
 \frac{1}{\eps^2}(\vert \bar N_\tau\vert^2-1)\bar N_\tau =0 \,\,\,\hbox{ for a.a. } (x,t)\in \Sigma\times (0,+\infty).
\end{equation}
Once we have \eqref{eq:discr2}, we can obtain in a standard way some uniform 
(with respect to $\tau$) 
a priori estimates and we can pass to the limit as $\tau\searrow 0$. As a consequence,
we obtain 
a solution to \eqref{eq:gfn_approx}. Note that this procedure 
provides a map $\n^\eps$ which, besides solving \eqref{eq:gfn_approx} pointwise,
is a tangent vector field, namely for a.a. $x\in \Sigma$ there holds
 $\n^\eps(x)\in T_x\Sigma$. This
property follows from the fact that $N_k(x)\in T_x\Sigma$ and from the fact
that the convergence of the discrete solutions to $\n^\eps$ is strong enough.

The question is clearly to pass to the limit as $\eps\searrow 0$ and to recover a solution of \eqref{eq:gfn}-\eqref{eq:initialn}.
To this end, we perform some a priori estimates on the solutions to $\eqref{eq:gfn_approx}$
that are independent of $\eps$. We take the scalar product of $\mathbb{R}^3$
between the approximate equation and $\partial_t\n^\eps$ and then we integrate
over $\Sigma$. We have
\begin{equation}
\label{eq:energy_est}
	\| \partial_t \n^\eps(t)\|^2 +\frac{\d}{\d t}
	\mathcal{E}_{\eps}(\n^\eps(t))=  \| \partial_t \n^\eps(t)\|^2
	 +\frac{\d}{\d t}W_{\kappa}(\n^\eps(t)) 
		+ \frac{1}{4\eps^2}\frac{\d}{\d t}\int_{\Sigma}
		(\vert \n^\eps(t)\vert^2-1)^2\dvolg = 0.
\end{equation}  
Thus, integrating on $(0,T)$, $T>0$, and using that $\n_0\in \Hut$, we get the following estimate

	\begin{align}
		\| \partial_t\n^\eps\|^2_{L^2(0,T;\Lt)} &+  \| D \n^\eps\|^2_{L^\infty(0,T;\Lt)}
			+ \| \B\n^\eps \|^2_{L^\infty(0,T;\Lt)} \nonumber\\
			&+\sup_{t\in(0,T)}\frac{1}{4\eps^2}\int_\Sigma (\vert \n^\eps(t)\vert^2-1)^2 \dvolg \le 
			 3\mathcal E_\eps(\n_0)= 3W_\kappa(\n_0).\label{eq:energy_est2}
	\end{align}
 Now, the estimate above gives the existence of a vector field $\n\in H^1(0,T;\Lt)\cap
 L^\infty(0,T;\Ht)$ with $\n(0)=\n_{0}$
and of a not relabeled subsequence of $\eps$ such that
\begin{align}
	\displaystyle \n^\eps &\xrightarrow{\eps\searrow 0}\n & &\hbox{weakly star in }  
		L^\infty(0,T;\Ht) \hbox{ and strongly in } L^2(0,T;\Lt), \label{eq:weak_convn}\\
	\displaystyle \partial_t\n^\eps &\xrightarrow{\eps\searrow 0}\partial_t\n & &\hbox{weakly in } 
		L^2(0,T;\Lt), \label{eq:weak_convnt}\\
	\displaystyle \Bd \n^\eps  &\xrightarrow{\eps\searrow 0} \Bd\n & &\hbox{strongly in } 
		L^2(0,T;\Lt),\label{eq:conv3forma}
\end{align}
where the last convergence follows directly from the 
continuity of the shape operator with respect to the strong convergence
in $L^2$ and from 
the definition of
the operator $\Bd$. 
Moreover, from \eqref{eq:energy_est2} we have that
\[
  	\int_\Sigma (\vert \n^\eps(t) \vert^2-1)^2 \dvolg  \le 12W_\kappa(\n_0)\eps^2 \quad \forall \eps>0,\ \forall t\in (0,T),  
\]	
which implies that (up to subsequences)
\begin{equation}
\label{eq:conv_vincolo}
	\vert \n^\eps\vert^2 \xrightarrow{\eps\searrow 0} 1\quad \hbox{a.e. on  } \Sigma\times (0,T).
\end{equation}
As a consequence, we have that $\vert \n\vert =1$ a.e. in $\Sigma$ for any time interval 
$(0,T)$, and hence 
\[
\n \in L^\infty(0,+\infty;\Hut).
\]
 Moreover, integrating \eqref{eq:energy_est} 
between $0$ and $+\infty$, we have
\[
\partial_t \n\in L^2(0,+\infty;\Lt).
\]
To conclude we have to prove that $\n$ solves \eqref{eq:gfn}. 
It is important to observe that the bounds at our disposal on the sequence
$\n_\eps$ are too weak for directly pass to the limit in the singular term in equation 
\eqref{eq:gfn_approx}. 
To this end, we follow the approach devised
in \cite{chen}. 
This is based on the observation that 
a smooth solution of \eqref{eq:gfn} actually solves
\begin{equation}
\label{eq:equivalent_heatflow}
(\partial_t \n -\Delta_g \n + \B^2\n)\times \n = 0 
\end{equation}
and viceversa (see Remark \ref{rem:decomposition} for the stationary case).
This observation is of extreme importance
in the analysis of these kind of problems and it has been noticed
and used in \cite{chen} and \cite{chen_struwe}, for instance. 

In the next Lemma \ref{lemma:weak_equiv}, we prove the equivalence between \eqref{eq:sol_distr}
(namely the weak version of \eqref{eq:gfn}) and a proper weak formulation of \eqref{eq:equivalent_heatflow}. 
In particular, this new formulation will be fundamental in the limit procedure due to its divergence-like structure.
Our argument is inspired by \cite{chen}, \cite[Lemma 7.5.4]{Lin_Wang}.

\begin{lemma}
\label{lemma:weak_equiv}
A vector field $\n\in H^1(0,1;\Lt)\cap L^\infty(0,T;\Hut)$ 
solves \eqref{eq:sol_distr} if and only if it solves
\begin{align}
\label{eq:weak_equiv}
-\int_{\Sigma}(\partial_t \n\times \n, \nb)_{\Rt}\psi\,\dvolg + 
\int_{\Sigma}g^{ij}(D_i\n,\nb\times \n)_{\Rt}\partial_{j}\psi\,\dvolg 
- \int_{\Sigma}(\B^2\n\times\n,\nb)_{\Rt}\psi\,\dvolg = 0
\end{align}
for all smooth functions $\psi:\Sigma\to \mathbb{R}$.
\end{lemma}
 \begin{proof}
 
 {\underline{\noindent Step $1$}}
 \par
 Let $\n$ be a solution of \eqref{eq:sol_distr} with the
 regularity specified in the statement. Choose in the weak formulation \eqref{eq:sol_distr}
 the tangent vector field $\phi=\hat{\nb}\times \n$,
 where $\hat{\nb}:=\psi\,\nb$ and $\psi:\Sigma\to\mathbb{R}$.
 First, the quadratic term in the right hand side disappears, as 
 for any $a,b\in \Rt$
 $(a,b\times a)_{\Rt} = (b, a\times a)_{\Rt}= 0$. Thus, 
 \[
 \int_{\Sigma} (\vert D\n\vert^2 \n + \vert \B\n\vert^2 \n,\hat{\nb}\times \n)_{\Rt}\dvolg
 = \int_{\Sigma} (\vert D\n\vert^2 \n + \vert \B\n\vert^2 \n, \nb\times \n)_{\Rt}\psi\,\dvolg
 =0.
 \]
Then, 
 \begin{align}
 \int_\Sigma(\partial_t \n,\hat{\nb}\times \n)_{\Rt}\dvolg &=
  -\int_{\Sigma}(\partial_t \n\times\n,\nb)_{\Rt}\psi\,\dvolg\label{eq:weak_equiv1}\\
 \int_{\Sigma}(\B^2\n,\hat{\nb}\times \n)_{\Rt}\dvolg &= -\int_{\Sigma}(\B^2\n\times\n,\nb)_{\Rt}\psi\,\dvolg,
 \end{align}
 for all $\psi:\Sigma\to\mathbb{R}$.
Regarding the remaining term, by Lemma \ref{lemma:tecnico},
 there holds  
 \begin{equation}
 \label{eq:deriv_norm}
 D_j (\nb\times \n) = \nb\times D_j \n, \,\,\,\,\forall j=1,2,
\end{equation}
 which implies
 \[
 D_j (\hat{\nb}\times \n) = D_j(\psi\nb\times \n)= (\partial_j\psi)\nb\times\n + \psi\nb\times D_j\n.
 \] 
 Thus, we get 
 \begin{align}
 \int_{\Sigma}D\n: D\phi\,\dvolg &= \int_{\Sigma}g^{ij}(D_{i}\n,D_{j}(\hat{\nb}\times \n))_{\Rt}\dvolg \nonumber\\
 &= \int_{\Sigma}g^{ij}(D_i \n, \nb\times D_j\n)_{\Rt}\psi\,\dvolg 
 + \int_{\Sigma}g^{ij}(D_i \n,\nb\times\n)_{\Rt}\partial_{j}\psi\,\dvolg\nonumber\\
 &= \int_{\Sigma}g^{ij}(D_i \n,\nb\times\n)_{\Rt}\partial_{j}\psi\,\dvolg,\label{eq:weak_equiv3}
 \end{align}
 where the first addendum vanishes since 
 \[
 \int_{\Sigma}g^{ij}(D_i \n, \nb\times D_j\n)_{\Rt}\psi\,\dvolg
 =- \int_{\Sigma}g^{ij}(D_i \n \times D_j\n, \nb)_{\Rt}\psi\,\dvolg
 \]
and since the metric tensor $g$ is symmetric and
the cross product is skew symmetric. \\

\underline{{\noindent Step $2$}}
\par
Now, let be given a vector field $\n$ satisfying \eqref{eq:weak_equiv}  
and the regularity of the statement. 
As above, we indicate with $\hat{\nb}$ the vector field
$\hat{\nb} := \psi \nb$. Thus, $\n$ verifies
\begin{eqnarray*}
\int_\Sigma(\partial_t \n,\hat{\nb}\times \n)_{\Rt}\dvolg  + 
\int_{\Sigma}D\n:D(\hat{\nb}\times \n)\dvolg 
+ \int_{\Sigma}(\B^2\n,\hat{\nb}\times \n)_{\Rt}\dvolg = 0.
\end{eqnarray*}
Choosing $\hat{\nb}$ of the form $\hat{\nb}=\n\times \phi$,
 and recalling that, being $\vert \n\vert=1$,
 $(\n\times\phi)\times \n = \phi - \n(\n,\phi)$,
  we get
 \begin{align*}
 \int_\Sigma(\partial_t \n,\hat{\nb}\times \n)_{\Rt}\dvolg 
 &= \int_\Sigma(\partial_t \n,(\n\times\phi)\times \n)_{\Rt}\dvolg\nonumber\\
 &=\int_{\Sigma}(\partial_t\n, \phi)_{\Rt}\dvolg - 
 \int_{\Sigma}(\partial_t\n,\n)_{\Rt}(\phi,\n)_{\Rt}\dvolg\nonumber\\
 &\stackrel{\vert \n\vert = 1}{=}\int_{\Sigma}(\partial_t\n, \phi)_{\Rt}\dvolg.
 \end{align*}
 Analogously, 
 \begin{align*}
  \int_{\Sigma}(\B^2\n,\hat{\nb}\times \n)_{\Rt}\dvolg = 
  \int_{\Sigma}(\B^2\n, \phi)_{\Rt}\dvolg
  - \int_{\Sigma}(\B^2\n,\n)_{\Rt}(\n,\phi)_{\Rt}\dvolg\nonumber\\
  = \int_{\Sigma}(\B^2\n, \phi)_{\Rt}\dvolg
  -\int_{\Sigma}\vert\B\n\vert^2(\n,\phi)_{\Rt}\dvolg.
 \end{align*}
 Finally,  
 \begin{align*}
 \int_{\Sigma}D\n:D(\hat{\nb}\times \n)\dvolg &= \int_{\Sigma}D\n:D((\n\times\phi)\times\n)\dvolg\nonumber\\
 & = \int_{\Sigma}D\n:D(\phi - \n(\n,\phi)_{\Rt})\dvolg\\
& = \int_{\Sigma}D\n:D\phi\,\dvolg - \int_{\Sigma}\vert D\n\vert^2(\n,\phi)_{\Rt}\dvolg. 
 \end{align*}
Collecting the above computations we get that $\n$ is a solution
of \eqref{eq:sol_distr}.
 \end{proof}
We are now in the position to conclude the proof
of Theorem \ref{Th:gfn}.
We recall that for the moment, we have proven that, 
up to a subsequence, 
\[
\n_\eps \xrightarrow{\eps \searrow 0} \n \,\,\,\hbox{ weakly in } H^1(0,T;\Lt) \hbox{ and weakly star in } L^\infty(0,T;\Ht),
\] 
where $\n$ is such that $\vert \n\vert =1$. Moreover, 
\[
\n_\eps \xrightarrow{\eps\searrow 0 }\n \,\,\,\hbox{ strongly in } L^2(0,T;\Lt).
\]
Now, test \eqref{eq:gfn_approx} with $\phi = \hat{\nb}\times \n_\eps$, where as before
$\hat{\nb}:=\psi\nb$ with $\psi:\Sigma\to\mathbb{R}$ smooth. 
The above computations give 
\begin{align}
\label{eq:weak_equiv_eps}
-\int_{\Sigma}(\partial_t \n_\eps\times \n_\eps, \nb)_{\Rt}\psi\,\dvolg + 
\int_{\Sigma}g^{ij}(D_i\n_\eps,\nb\times \n_\eps)_{\Rt}\partial_{j}\psi\,\dvolg 
- \int_{\Sigma}(\B^2\n_\eps\times\n_\eps,\nb)_{\Rt}\psi\,\dvolg = 0,
\end{align}
where the penalization term has disappeared again thanks to 
 $(a,b\times a)_{\Rt} = (b, a\times a)_{\Rt}= 0$, for $a,b\in \Rt$. 
Since \eqref{eq:weak_equiv_eps} is in divergence form, we can easily
pass to the limit as $\eps\searrow 0$ using the above proved 
convergences. 
Note indeed that all the terms easily pass to the limit 
as they are products of weakly and strongly convergent 
sequences in $L^2$. Consequently, we obtain that the limit $\n$
verifies \eqref{eq:weak_equiv} and thus solves \eqref{eq:sol_distr} thanks
to Lemma \ref{lemma:weak_equiv}.

In the next Lemma, we prove the formula 
\eqref{eq:deriv_norm}. 
\begin{lemma}
\label{lemma:tecnico}
Let us given $\n\in T_p\Sigma$. Let $\nb$ be the outer normal 
vector at the point $p$. Then, for $i=1,2$, there holds
\begin{equation}
\label{eq:tecnico}
D_{\ei}(\nb\times \n) = \nb\times D_{\ei}\n.
\end{equation}

\end{lemma}

\begin{proof}
Let $\left\{\eu,\ev\right\}$ be a basis for the tangent space 
at the point $p$. 
For $i=1,2$, the Gauss formula \eqref{eq:gauss} gives
\[
D_{\ei}(\nb\times \n) = \nabla_i(\nb\times \n) - h(\ei, \nb\times \n)\nb.
\]
Now, 
\[
h(\ei,\nb\times \n) = -(\nabla_i\nb,\nb\times \n)_{\Rt} = (\nabla_i\nu\times \n, \nb)_{\Rt},
\]
by definition. Then,
expanding $\nabla_i(\nb\times \n) = \nabla_i \nb\times \n + \nb\times \nabla_i \n$,
we get 
\[
D_{\ei}(\nb\times \n) = \nb\times \nabla_i\n + \nabla_i\nb\times \n -(\nabla_i\nu\times \n, \nb)_{\Rt}\nb
= \nb\times \nabla_i\n,
\]
being the last two terms normal vectors. 
We conclude if we prove that $\nb\times \nabla_i\n = \nb\times D_{\ei}\n$. 
This follows from the Gauss formula \eqref{eq:gauss} since
\[
\nb\times\nabla_i \n = \nb\times D_{\ei}\n + h(\ei,\n)\nb\times \nb= \nb\times D_{\ei}\n.
\]
\end{proof}

\begin{remark}
It is important to note that the above computations hold
for an hypersurface of dimension $n$, upon replacing the cross 
product $\times$ with the wedge product $\wedge$. Thus, Theorem 
\ref{Th:gfn} remains valid more generally on a hypersurface $N$ of dimension $n$ 
provided the corresponding space $H^1_{tan}(N;\mathbb{S}^2)$ is well defined (see \cite[Proposition 1.1]{CSV13}). 
\end{remark}


\newcommand{\CM}{C^1_{\text{tan}}(\Sigma;\mathbb S^2)}

\section{Representation of vector fields $\n$ via local deviation $\alpha$}   
\label{sec:alpha}

A reference textbook to the material covered in this Section is \cite{Lee2013}. Let $\Sigma \subset \R^3$ be a regular orientable compact surface (with or without boundary) with a maximal parametrization $(V_j,x_j)$, $x_j:V_j\subseteq \R^2\to \Sigma$. A set $U\subset \Sigma$ is said to be open in $\Sigma$ if $x_j^{-1}(U\cap x_j(V_j))$ is open in $\R^2$ for all $j$. For any open set $U \subseteq \Sigma$, let $\{\eu,\ev\}$ be a smooth local orthonormal frame, i.e. a pair of smooth sections of the tangent bundle $T\Sigma$ such that  $\{\eu(p),\ev(p)\}$ is an orthonormal basis for $T_p\Sigma$, for all $p\in U$.  

\subsubsection*{Winding number.} Define the 1-form 
\begin{equation}
\label{eq:windform}
	\omega:= \frac{ x\, \d y- y\,\d x }{x^2+y^2}\qquad \text{on }\R^2\setminus \{\mathbf 0\}.
\end{equation}
Given an oriented curve $\gamma$ in $\R^2\setminus \{\mathbf 0\}$, the line integral $\int_\gamma \omega$ measures the winding of $\gamma$ around $\mathbf 0$ in counterclockwise direction. The winding number of a closed curve $\gamma$ with respect to $\mathbf 0$  is the integer $\mathscr W(\gamma):=(2\pi)^{-1}\int_\gamma \omega$.  
Given a regular parametrization $\gamma:[0,1]\to \R^2\setminus \{\mathbf 0\}$ with components $\gamma(t)=(\gamma^1(t),\gamma^2(t))$,  its winding number can be computed via the pullback of $\omega$: 
\[	
	\mathscr W(\gamma)	= \frac{1}{2\pi}\int_{[0,1]}\gamma^* \omega 
		= \frac{1}{2\pi}\int_0^1 \frac{\gamma^1(t)\dot \gamma^2(t) - \gamma^2(t)\dot \gamma^1(t)}{|\gamma(t)|^2}\d t.
\] 
The relation between degree and winding number for a regular simple closed curve $\gamma:[0,1]\to \R^2\setminus \{\mathbf 0\}$ is then $\mathscr W(\gamma) = \deg(\gamma/|\gamma|,\mathbb S^1,\mathbb S^1)$  (after identifying the endpoints $\{0\}$ and $\{1\}$ in $[0,1]$).

Let now $\mathbf v: U\subseteq \R^2 \to \R^2$ be a smooth vector field $\mathbf v = (v^1,v^2)$. If $\gamma\cap \mathbf v^{-1}(\mathbf 0) = \emptyset$, it is natural to measure the winding of $\mathbf v$ along $\gamma$ by
\begin{equation}
\label{def:wind}
	\mathscr W_\gamma(\mathbf v):=\int_{\gamma} \mathbf v^*\omega.
\end{equation}

\subsubsection*{Winding of fields on surfaces} Assume that $T\Sigma$ admits a global orthonormal frame $\{\mathbf e_1,\mathbf e_2\}$. The mapping
\[
	\iota:T\Sigma \to \Sigma \times \R^2,\quad (p,\mathbf v)\mapsto (p, (v^1,v^2))\qquad  \forall p\in \Sigma, \forall \mathbf v=v^i\ei \in T_p \Sigma
\]	
defines a smooth diffeomorphism (see, e.g., \cite[Corollary 10.20]{Lee2013}). We can then extend the winding $\mathscr W_\gamma(\mathbf v)$ to sections of the tangent bundle $T\Sigma$, i.e., to smooth vector fields  $\Sigma \ni p\mapsto \mathbf v(p)\in T_p\Sigma$. 
For every smooth curve $\gamma$ on $\Sigma$ such that $\gamma\cap (\iota \circ\mathbf v)^{-1}(\mathbf 0) = \emptyset$, we define the winding of $\mathbf v$ along $\gamma$, with respect to $\{\mathbf e_1,\mathbf e_2\}$, by
\begin{equation}
\label{def:winde}
	\mathscr W_\gamma(\mathbf v):=\int_{\gamma}  (\iota \circ \mathbf v)^*\omega,
\end{equation}
where $\omega$ is the angle 1-form defined in \eqref{eq:windform}.

Moreover, the existence of a global orthonormal frame allows, via the mapping $\iota$, to associate to $\n \in \Hut$ an $\S^1$-valued field. A classical result by Bethuel and Zheng ensures then the existence of a local lifting $\alpha$. More precisely: 
\begin{proposition}[Lemma 4, \cite{BetZhe88}]
\label{prop:alp}
Let $\Sigma$ be a smooth surface embedded in $\R^3$ and $Q:=[0,2\pi]\times [0,2\pi] \subset \R^2$. Assume that there exist: 
\begin{itemize}
	\item[\_] a smooth global orthonormal frame $\{\eu,\ev\}$ on $\Sigma$,
	\item[\_] a smooth covering map $\pi_\Sigma: Q \to \Sigma$,
	\item[\_] a vector field $\n \in \Hut$.
\end{itemize}	
  Then there exists $\alpha \in H^1(Q)$ such that 
\begin{equation}
\label{eq:nalpha2}
	\n\circ \pi_\Sigma= \cos (\alpha) (\eu\circ \pi_\Sigma) +\sin (\alpha) (\ev\circ \pi_\Sigma) \qquad \text{a.e. in }Q.
\end{equation}
Moreover, if $\alpha_1,\alpha_2 \in H^1(Q)$ are two maps that by \eqref{eq:nalpha2} define the same vector field $\n \in \Hut$, then there exists $n\in \Z$ such that $\alpha_1 = \alpha_2 +2\pi n$ a.e. on $Q$. 
\end{proposition}
This proposition applies, in particular, to the cases where $\Sigma$ is diffeomorphic to a torus or to a disc. 
It is a common notation, which we adopt in the following sections, to drop ``$\circ \pi_\Sigma$". It will be clear from the context whether $\alpha,\n,\ei$ are defined on $\Sigma$ or parametrized on $Q$.

We notice that, if $\Sigma$ is not simply connected, it may not be possible to define $\alpha$ on the whole surface $\Sigma$. For example, given the standard parametrization of the torus $X:[0,2\pi]\times [0,2\pi]\to \bbT$ \eqref{eq:paramtorus}, $\n(\theta,\phi):=\cos(\theta)\eu(\theta,\phi)+\sin(\theta)\ev(\theta,\phi)$ defines a smooth vector field on $]0,2\pi[\times]0,2\pi[$. The only possible $\alpha$ is clearly $\alpha(\theta,\phi)=\theta +2h\pi$, for $h\in \Z$, which cannot be continuously extended to $[0,2\pi]\times[0,2\pi]$ since $2h\pi =\lim_{t\to0^+}\alpha(t,\phi)\neq \lim_{t\to 2\pi^-}\alpha(t,\phi)= 2\pi(1+h)$.


\subsection{Formulas for the deviation $\alpha$}
\label{ssec:formulas}
In this subsection, we perform the formal computations which lead to the representation of $\nabla_s \n$, in terms of $\alpha$. 

First of all, we introduce the \emph{spin connection} $\mathbb{A}$, which, 
for a two-dimensional manifold $\Sigma$ embedded in $\R^3$, can be expressed using
the  1-form $\omega$ defined as
\begin{equation}
\label{eq:one_form_om} 
\omega(\vb) = (\eu, D_{\vb} \ev)_{\R^3} \quad \forall \vb \in T_p \Sigma,
\end{equation}
where $\left\{\eu, \ev\right\}$ is a local orthonormal frame for $T\Sigma$. 
Deriving the relation $(\ei,\ej)_{\R^3}=\delta_{ij}$ one obtains
\begin{equation}
\label{eq:on}
	0=\partial_k (\ei,\ej)_{\R^3}=(D_k \ei,\ej)_{\R^3}+ (\ei,D_k\ej)_{\R^3}, \,\,\,\hbox{ for }\,\,k=1,2,
\end{equation}
which implies 
that $\omega(\vb) = - (\ev, D_{\vb}\eu)_{\R^3}$ for any $\vb$ tangent
and that $(\eu,D_i \eu)_{\R^3} = (\ev, D_i \ev)_{\R^3} = 0$ for $i=1,2$.
Consequently, we have that 
\begin{align}
& D_i \eu = -\omega(\ei) \ev,\label{eq:spin1}\\
& D_i \ev = \omega(\ei)\eu\label{eq:spin2}.
\end{align} 
The spin connection $\mathbb{A}$ is the tangent vector field 
$\mathbb{A} := \omega^{\sharp}$, that is 
 $\mathbb{A}^{i} = g^{ij}\omega_{j}$.
 In what follows we will unambiguously refer to $\mathbb{A}$ and 
 to $\omega$ as the spin connection. Let  
 $\kappa_1 ,\kappa_2$ be the geodesic curvatures of the flux lines of
 $\eu$, $\ev$, respectively. By the definition of geodesic curvature, \eqref{eq:one_form_om} and \eqref{eq:on}, it is immediate to see that
 \begin{equation}
\label{eq:defspin}
 	\mathbb A =-\kappa_1 \eu -\kappa_2\ev.
\end{equation}

Now we show how the spin connection $\mathbb{A}$ and
its related 1-form $\omega$
change when we change the orthonormal frame. In particular, 
it will be important to be able to choose a local orthonormal frame
with divergence-free spin connection (see \cite[Lemma 3.2.9]{Lin_Wang} for a similar result).
Thus, let 
$\left\{\fu, \fd\right\}$ be another smooth local orthonormal frame
centered $U$. We denote with $\beta$ the angle 
that $\fu$ forms with $\eu$. Thus, have 
\begin{eqnarray*}
\fu &=& \cos\beta \eu + \sin\beta \ev,\\
\fd &=& -\sin\beta\eu + \cos\beta \ev.
\end{eqnarray*}
\begin{lemma}
\label{lemma:newomega}
Let $\omega'$ denote the spin connection in the frame $\left\{\fu, \fd\right\}$, namely
the 1-form $\omega'(\vb) = (\fu, D_{\vb}\fd)_{\R^3}$ for $\vb$ tangent. Then there
holds
\begin{equation}
\label{eq:newomega}
\omega'(\vb) = \omega(\vb) - \d\beta(\vb).
\end{equation}
Moreover, if $\mathbb{A}' = (\omega')^{\sharp}$, we have
\begin{equation}
\label{eq:divomega}
	\textup{div}_s \mathbb{A}' = \textup{div}_s\mathbb{A} -\Delta \beta. 
\end{equation}
\end{lemma}
\begin{proof}
We have
\begin{equation*}
\begin{split} 
	\omega'\left(\delj\right) = ( \fu, D_j\fd)_{\R^3} 
		&= (\cos\beta \eu + \sin\beta \ev, D_j(-\sin\beta\eu + \cos\beta \ev))_{\R^3}\\
		& = \cos\beta ( \eu, D_j(-\sin\beta\eu))_{\R^3} + \sin\beta( \ev, D_j(\cos\beta \ev))_{\R^3}\\
		& \quad+ \cos\beta( \eu, D_j(\cos \beta\ev))_{\R^3} + \sin\beta ( \ev, D_j(-\sin\beta \eu))_{\R^3} \\
		& \stackrel{\eqref{eq:on}}{=}  -\cos \beta \partial_j(\sin \beta) -\sin^2\!\beta\, ( \ev, D_j \eu)_{\R^3} 
			+ \cos^2\!\beta\, (\eu,D_j \ev)_{\R^3} + \sin\beta \partial_j(\cos \beta)\\
		& \stackrel{\eqref{eq:one_form_om}}{=} \omega_j -\partial_j \beta = \omega\left(\delj\right) - \d \beta \left(\delj\right).
\end{split}
\end{equation*}
By linearity of $\omega',\omega,$ and $\d \beta$, we conclude \eqref{eq:newomega}. Now, to prove \eqref{eq:divomega},
we notice that \eqref{eq:newomega} corresponds, after the  $^\sharp$ isomorphism, to
\begin{equation*}
\mathbb{A}' = \mathbb{A} - \nabla_s\beta,
\end{equation*}
thus \eqref{eq:divomega} follows. 
\end{proof}

 We are going to prove the following
\begin{lemma}
\label{lemma:gradalp}
Let $U\subset \Sigma$ be open and simply connected and let $\n\in H^1_{tan}(U;\mathbb{S}^2).$
Then, for a.a. $x\in U$,
\begin{align}
\label{eq:cov_alp}
D \n = (\omega-\d \alpha)\tbf,
\end{align}
where
\[
\tbf = \sin \alpha \eu - \cos\alpha \ev.
\]
Consequently, we have
\begin{align}
	\vert D\n \vert^2 =& \vert \nabla_s \alpha - \mathbb{A}\vert^2, \label{eq:normcovalp}\\
	\vert \nabla_s \n\vert^2 = &\vert \nabla_s \alpha - \mathbb{A}\vert^2 
		+ \vert \B\eu\vert^2 \cos^2\alpha + \vert\B\ev\vert^2\sin^2\alpha
		+2(\B\eu,\B\ev)_{\mathbb{R}^3}\sin\alpha\cos\alpha.
\label{eq:normsgradalp}
\end{align} 
\end{lemma}
\begin{proof}
For a.a $x\in U$, let $\left\{\eu,\ev\right\}$ be a smooth local orthonormal
frame for $T_x\Sigma$, 
Then $\n\in H^1_{tan}(U;\mathbb{S}^2)$
is represented as in \eqref{eq:nalpha2} with $\alpha\in H^1(U)$ being the angle
between $\n$ and $\eu$. We have that, for a.a. $x\in U$ and for $i=1,2$, 
\[
	D_{\ei} \n = (\cos\alpha) D_{\ei}\eu + (\sin\alpha) D_{\ei}\ev
	-(\d \alpha(\ei))\sin\alpha\eu + (\d \alpha(\ei))\cos\alpha\ev.
\]
Thus, using \eqref{eq:spin1} and \eqref{eq:spin2} and grouping in terms of $\d\alpha$, $\omega$ and $\tbf$
we get 
\[
D_i \n = (\omega(\ei)-\d \alpha(\ei))\tbf, 
\]
which is \eqref{eq:cov_alp}.
Then, since $\eu$ and $\ev$ are orthonormal, we have that
\begin{eqnarray*}
 	\vert D_{\eu} \n\vert^2 &=& \vert \omega(\eu) - \d\alpha(\eu)\vert^2,\\
 	\vert D_{\ev} \n\vert^2 &=&  \vert \omega(\ev) - \d\alpha(\ev)\vert^2,
\end{eqnarray*}
which implies \eqref{eq:normcovalp} by recalling again
the orthonormality of $\eu$ and $\ev$.
Once we have \eqref{eq:normcovalp}, we can easily obtain 
\eqref{eq:normsgradalp} using the orthogonal decomposition
\eqref{eq:normgrad} once we have written $\vert \B\n\vert$
in terms of $\alpha$. We have 
$$
\B \n = \B (\cos \alpha \eu + \sin\alpha \ev) = \cos\alpha \B\eu 
+ \sin\alpha\B\ev,
$$
thus, 
$$\vert\B\n\vert^2 = \cos^2\alpha\vert \B\eu\vert^2
+\sin^2\alpha \vert \B\eu\vert^2 + 2(\B\eu,\B\ev)_{\R^3}\cos\alpha\sin\alpha,
$$
which, combined with \eqref{eq:normcovalp}, gives 
\eqref{eq:normsgradalp}.
\end{proof}
The expression \eqref{eq:normsgradalp} further simplifies
if we choose, for any point $x\in \Sigma$, $\left\{\eu,\ev\right\}$ to be the 
{\itshape principal directions}
of $\Sigma$ at $x$. In particular, $\left\{\eu,\ev\right\}$ are orthonormal eigenvectors
of $\mathfrak{B}$. The relative eigenvalues $c_1$ and $c_2$ are
named {\itshape principal curvatures}
of $\Sigma$ at $x$ (\cite{DoCarmo76}). As a result, we have
\begin{align}
\label{eq:normgradalp2}
	\vert\nabla_s \n\vert^2 =&  \vert \nabla_s \alpha - \mathbb{A}\vert^2 
		+ \vert \B\eu\vert^2\cos^2\alpha + \vert\B\ev\vert^2\sin^2\alpha\nonumber\\ 
	= &\vert \nabla_s \alpha - \mathbb{A}\vert^2
		+ \frac{(c_1^{2} - c_{2}^2)}{2} \cos(2\alpha) + \frac{(c_1^{2} + c_{2}^2)}{2}.
\end{align}
An immediate consequence of the representation of the the covariant derivative of $\n$ in terms of $\alpha$ is 
the representation of its Rough Laplacian. We have the following
\begin{lemma}
\label{lem:rough_alpha}
Let $U\subset \Sigma$ be open and simply connected and let $(\eu,\ev)$ a local  be a local orthonormal frame. Let $\n$ be
 sufficiently smooth vector field with unit norm.
 Then, when 
 $\n$ is represented as 
 \[
 \n = \cos\alpha\, \eu + \sin\alpha\,\ev,
 \]
 there holds
\begin{equation}
\label{eq:rough_alpha}
\Delta_ g \n = (\Div \mathbb{A}-\Delta \alpha )\tbf - \vert \nabla_s \alpha - \mathbb{A}\vert^2\n,\quad \hbox{ in } U,
\end{equation}
where $\tbf = \sin\alpha\, \eu - \cos\alpha\,\ev$.
\end{lemma}
\begin{proof}
We compute the second covariant derivative of $\n$ using \eqref{eq:cov_alp}. We have
\begin{equation}
\label{eq:sec_cov_alp}
D^2 \n = (D\omega -\d^2\alpha )\tbf + (\omega-\d\alpha)\otimes D\tbf.
\end{equation}
Now, since $D\tbf = -(\omega - \d \alpha)\n$, we conclude that 
\[
	D^2 \n = (D\omega- \d^2\alpha)\tbf - (\omega-\d\alpha)\otimes(\omega - \d\alpha)\n.
\]
The rough laplacian is the trace of $D^2 \n$. Thus, by tracing the above identity we get
\[
\Delta_g \n = \hbox{tr}D^2\n = (\hbox{tr} D\omega - \Delta \alpha) \tbf - \vert \d\alpha - \omega\vert^2\n,
\]
which is \eqref{eq:rough_alpha}.
\end{proof}
It is worthwhile noting that thanks to Lemma \ref{lemma:newomega} we can always 
choose a local orthonormal frame for which the spin connection is divergence free. As a result, 
\eqref{eq:rough_alpha} simplifies in
\begin{equation}
\label{eq:smoothrough}
\Delta_ g \n = -\Delta \alpha \,\tbf - \vert \nabla_s \alpha - \mathbb{A}\vert^2\n\,\,\,\,\,\hbox{ in } U.
\end{equation}

\begin{lemma}
Let $\{\eu,\ev\}$ be the orthonormal frame provided by the principal directions on $\Sigma$. Let $c_1,c_2$ be the corresponding principal curvatures and let $\kappa_1,\kappa_2$ be the corresponding geodesic curvatures. The energy \eqref{eq:napoli} of a director field $\n$, in terms of the deviation angle $\alpha$ characterized by $\n=\cos(\alpha)\eu+\sin(\alpha)\ev$ and of the spin connection \eqref{eq:defspin} is
 \begin{align} 
 	W(\n) &= \frac 12 \int_\Sigma \left\{K_1 ((\nabla_s \alpha -\mathbb A)\cdot \tbf)^2 +K_2(c_1-c_2)^2\sin^2\alpha\,\cos^2\alpha\right. \nonumber\\
		&\qquad \left.+K_3 ((\nabla_s \alpha -\mathbb A)\cdot \n)^2 +K_3(c_1\cos^2\alpha+ c_2\sin^2\alpha )^2\, \right\}\dvolg. \label{eq:energyb}
	\end{align}	
The corresponding one-constant approximation ($\kappa=K_1=K_2=K_3$) is
		\begin{equation}
	\label{eq:oneconst}
	 W_\kappa(\n) = \frac{\kappa}{4}\int_\Sigma \left\{c_1^2+c_2^2\right\} \dvolg  + \frac \kappa2 \int_\Sigma \left\{ |\nabla_s \alpha -\mathbb A|^2 +\frac12(c_1^2-c_2^2)\cos(2\alpha) \right\} \dvolg.
	\end{equation}	
\end{lemma}
\begin{proof}
The expression in \eqref{eq:oneconst} follows directly from \eqref{eq:normgradalp2}. Regarding \eqref{eq:energyb}, we use Liouville's formula \cite[Proposition 4, Section 4--4]{DoCarmo76} to compute
\begin{align*}
	\kappa_\n &= \kappa_1 \cos(\alpha) + \kappa_2 \sin(\alpha) + \d\alpha (\n) = (\nabla_s \alpha -\mathbb A)\cdot \n,\\
	\kappa_\tbf &= -\kappa_1 \sin(\alpha) + \kappa_2 \cos(\alpha) + \d\alpha (\tbf) = (\nabla_s \alpha -\mathbb A)\cdot \tbf.	
\end{align*}
Using the definitions of $\tau_\n$ and $c_\n$ and the choice of $\eu,\ev$ as principal directions, we get
\begin{align*}
	c_\n & = (\B \n,\n)_{\R^3} = c_1 \cos^2(\alpha) + c_2 \sin^2(\alpha),\\
	\tau_\n &= - (\B \n,\tbf)_{\R^3} = c_1\cos(\alpha)\sin(\alpha) -c_2\cos(\alpha)\sin(\alpha).
\end{align*}
The expression in \eqref{eq:energyb} follows then by \eqref{eq:div-curl2}.
\end{proof}

\subsection{Proof of Theorem \ref{Th:H1empty}}
\label{ssec:proofempty}

\begin{proof}[Proof of Theorem \ref{Th:H1empty}]
Let $\Sigma$ be given, as in the hypothesis of Theorem \ref{Th:H1empty}. Referring to Section \ref{sec:framework}, we consider $E:=\Hut$ as a subset of the Hilbert space $X:=H^1_{tan}(\Sigma)$. Assume that $E \neq \emptyset$, we need to prove that $\chi(\Sigma)= 0$.  We study the minimization problem related to the energy 
\begin{equation}
\label{eq:frakE}
	\mathfrak{E}:X \to \R,\qquad \mathfrak{E}(\ub) := \frac 12 \int_{\Sigma}\vert D \ub\vert^2 \dvolg.
\end{equation}
Since the function $f:\Sigma \times \R^3 \to \R$, $f(x,\xi)=g(x)(\xi,\xi)\sqrt{g(x)}$ is continuous and convex in $\xi$ for all $x\in \Sigma$, the energy 
$\mathfrak{E}$ is weakly lower semicontinuous on $X$. As 
the constraint ``$\vert \ub\vert=1$ a.e. on $\Sigma$" is continuous with respect
to the $L^2$ convergence, we deduce that sublevel sets of $\mathfrak E$ in $E$ are sequentially weakly compact in $X$. Hence, using the direct method of the calculus of variations we can find 
 a field $\ub^*\in E$ which minimizes $\mathfrak E$ on $E$. 
In the next step we prove that $\ub^*$ is actually continuous,
so that we can apply the classical Poincar\'e-Hopf Theorem (see \cite{milnor}). 
Thanks to the local representation of tangent vectors in Proposition \ref{prop:alp}, 
for any given point $x\in \Sigma$ we can find an open neighbourhood 
$U\subset \Sigma$ and a real function $\alpha:U\to \mathbb{R}$
such that any vector field $\ub\in E$ can be locally represented as
$ \ub = \cos\alpha\, \eu +\sin\alpha\, \ev $ a.e. in $U$. Here
$\left\{\eu, \ev\right\}$ is a smooth local orthonormal frame for
$T_x \Sigma$ for all $x\in U$, and 
$\alpha\in H^1(U)$ is the angle that $\ub$ forms with $\eu$. 
Owing to Lemma \ref{lemma:newomega}, it is not restrictive to assume that
the spin connection $\mathbb A$ corresponding to $\{\eu, \ev \}$ is divergence-free: indeed
if $\divs \mathbb A \neq 0$,
we can define a new orthonormal frame by rotating $\{\eu, \ev \}$ of an angle $\beta$ such that
$\Delta\beta = \divs \mathbb{A}$ in $U$. The spin connection $\mathbb A'$ in the new frame, owing to \eqref{eq:divomega}, satisfies then $ \divs \mathbb{A}' = \divs \mathbb{A} - \Delta\beta=0$.
  
Now, since $\ub^*$ minimizes \eqref{eq:frakE} on $E$, by Lemma \ref{lemma:gradalp} any function $\alpha^*\in H^1(U)$, such that $\ub^*:=\cos\alpha^*\, \eu +\sin\alpha^*\, \ev$ on $U$, minimizes
\begin{equation}
\label{eq:enalpha}
 	\mathfrak F:H^1(U)\to\R,\qquad \mathfrak{F}(\alpha) := \frac 12\int_{U}\vert \nabla_s\alpha-\mathbb{A}\vert^2 \dvolg,
\end{equation}
on the set $\{\alpha \in H^1(U): \alpha^{\phantom *}_{|\partial U}=\alpha^*_{|\partial U}\}$.  As a result, $\alpha^*$
is a stationary point of \eqref{eq:enalpha}, with respect to variations in $H^1_0(U)$, and hence it solves
\[
	\Delta \alpha^* = 0\quad \hbox{ in } U.
\]
As the Laplace Beltrami operator on a smooth compact manifold is an elliptic operator with smooth coefficients, we have that $\alpha^*$, hence $\ub^*$, is smooth in $U$. Being 
the choice of the point $x$ completely arbitrary, we have proved that $\ub^*$ is 
a unit norm vector field which is smooth
everywhere in $\Sigma$. Thanks to the classical Poincar\'e-Hopf Theorem, $\Sigma$ must be a genus-1 surface, i.e. $\chi(\Sigma) = 0$. The opposite implication is straightforward. More precisely,  assuming that
$\chi(\Sigma) =0$, classical results give the existence of a smooth 
vector field on $\Sigma$ with unit norm, which, in particular, belongs to $\Hut$.

\end{proof}


\section{Energy minimizers on a torus} 
\label{sec:torus}
 
In this section we study the problem of minimizing the surface energy \eqref{eq:napoli} and its one-constant approximation \eqref{eq:napolioc} in the particular case of an axisymmetric torus $\bbT\subset \R^3$. Given the radii $0<r<R$ (see Figure \ref{fig1}), we consider the parametrization $X:\R^2\to \R^3$ defined by
\begin{equation}
\label{eq:paramtorus}
	X(\theta,\phi) = 
		\begin{pmatrix} 
			(R+r\cos \theta)\cos \phi \\ 
			(R+ r\cos \theta)\sin \phi \\ 
			r\sin \theta
		\end{pmatrix}.
\end{equation}
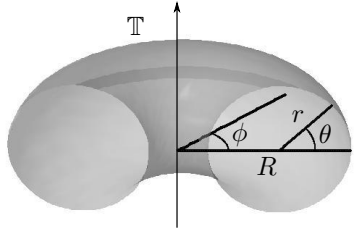
\begin{figure}[h]
\begin{center}
\labellist
		\hair 2pt
		\pinlabel $\phi$ at 265 200
		\pinlabel $\bbT$ at 150 320	
		\pinlabel $\theta$ at 360 200
		\pinlabel $R$ at 295 165		
		\pinlabel $r$ at 330 220				
	\endlabellist				
\centering{
 \includegraphics[height=4.8cm]{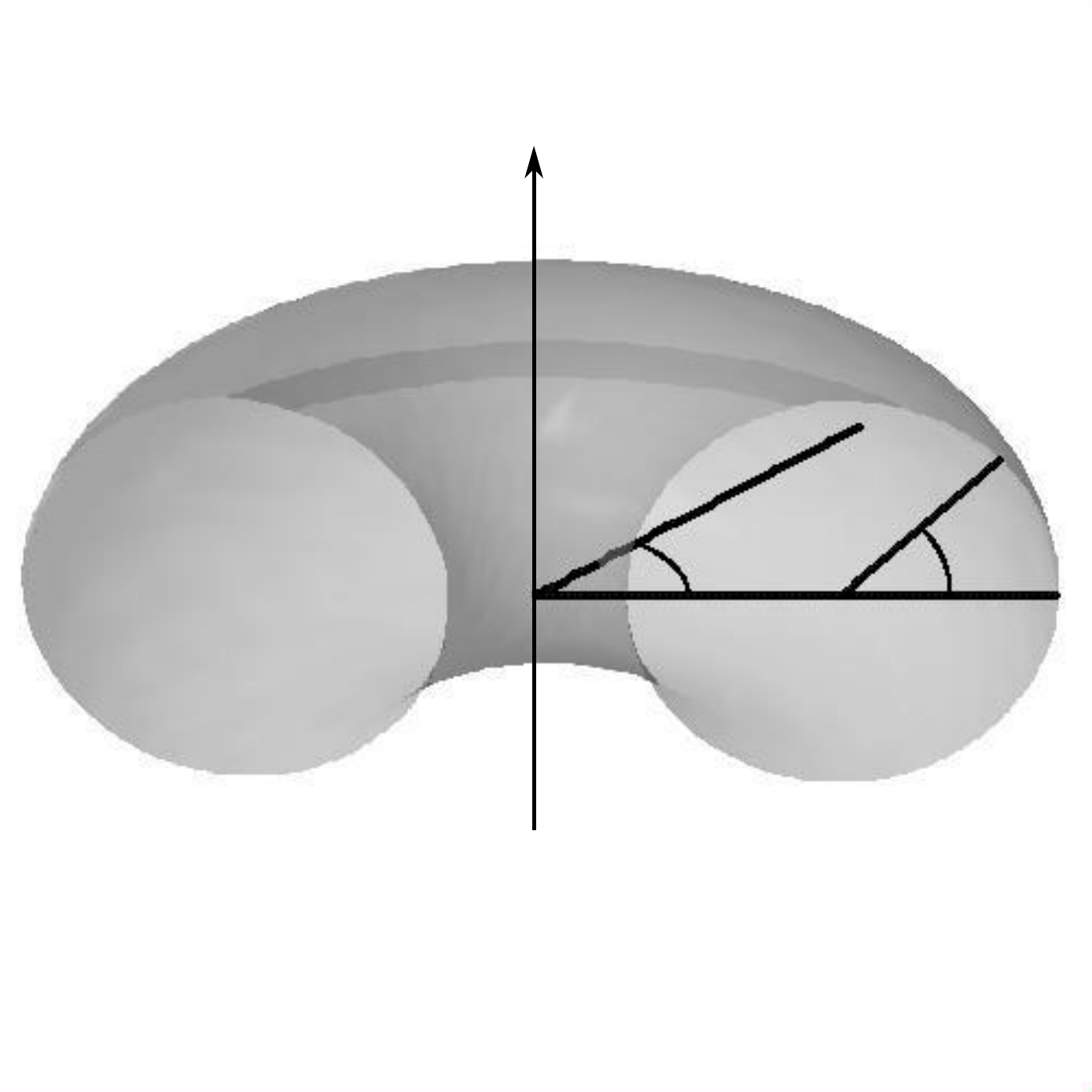}\\  \vspace{-1cm}
}
\end{center}
\caption{Schematic representation of the torus $\bbT$ parametrized by \eqref{eq:paramtorus}.}
\label{fig1}
\end{figure}
Let $\{\eu,\ev\}$ be the orthonormal frame associated to $X$ (see Appendix A).  By Proposition \ref{prop:alp}, any vector field $\n \in \HT$ can be represented by a scalar deviation $\alpha$, with respect to $\eu$, such that
\[
	\n\circ X=\cos(\alpha)\eu +\sin(\alpha)\ev.
\] 
Moreover, since $X$ is $2\pi$-periodic in both variables, we can assume that $\alpha \in H^1(Q)$,  for $Q:=[0,2\pi]\times [0,2\pi] \subset \R^2$. (Note that, since $\bbT$ is not simply connected, we cannot define $\alpha$ directly on $\bbT$.)

\subsection{A toy-problem: constant deviation}
\label{ssec:constantalpha}

In this section, with a slight abuse of notation, we let $W(\alpha):=W(\n)$, for $\n=\cos\alpha \eu +\sin\alpha \ev$. We study the simpler case of $\alpha\equiv const$, where the energy $W(\alpha)$ in \eqref{eq:energyb} reduces to
 \begin{align*} 
 	W(\alpha) &= \frac 12 \int_Q \left\{K_1 \cos^2\alpha \left(\kappa_2 \right)^2 +K_2(c_1-c_2)^2\sin^2\alpha\,\cos^2\alpha\right. \\
		&\qquad \qquad \left.+K_3 \sin^2\alpha \left(\kappa_2 \right)^2 +K_3(c_1\cos^2\alpha+ c_2\sin^2\alpha )^2\, \right\}\dvol. 
	\end{align*}	
Here $K_1,K_2,K_3$ are positive constants and (see Appendix A)
\[
		c_1=\frac{1}{r^2},\qquad c_2=\frac{\cos\theta}{R+r\cos\theta},\qquad  \kappa_2
			=-\frac{\sin\theta}{R+r\cos\theta},\qquad \dvol=r(R+r\cos\theta)\d\theta \d\phi.
\]
\begin{lemma}
\label{lemma:expl}
Let $b:=R/r$. In the case of constant deviation $\alpha$, the energy $W$ has the explicit expression
\begin{align*}
	W(\alpha) &=\pi^2\left[ (K_1 +K_3)\left(b - \sqrt{b^2-1}\right) 
					+ \frac{K_2+K_3}{2}\left(\frac{b^2}{\sqrt{b^2-1}}\right)\right] \\
		 	&\qquad + \pi^2\cos(2\alpha)\left[ (K_1-K_3)\left(b - \sqrt{b^2-1}\right) 
					+ K_3 \left(2b-\frac{b}{\sqrt{b^2-1}}\right)\right]\nonumber \\
		&\qquad + \pi^2\cos^2(2\alpha)\left[ \frac{K_3-K_2}{2}\left(\frac{b^2}{\sqrt{b^2-1}}\right)\right]. 
\end{align*}
\end{lemma}
The proof relies on algebraic manipulations and integration of trigonometric functions, which are detailed in \cite{SSV14}.
There are four parameters which influence the minimizers of $W$, that is $R/r$, $K_1, K_2,K_3$. In Figure \ref{fig2} we plot the graph $\{(\alpha,W(\alpha)/\pi^2\}$ for some especially meaningful choices of these parameters. The rescaling by $\pi^2$ is just for plotting purposes.	

\begin{figure}[h]
\begin{center}
\begin{minipage}{7cm}
\centering{
\labellist
		\hair 2pt
		\pinlabel $\displaystyle\frac{W(\alpha)}{\pi^2}$ at -25 180
		\pinlabel $\alpha$ at 380 -5
	\endlabellist				
 \includegraphics[height=3.5cm]{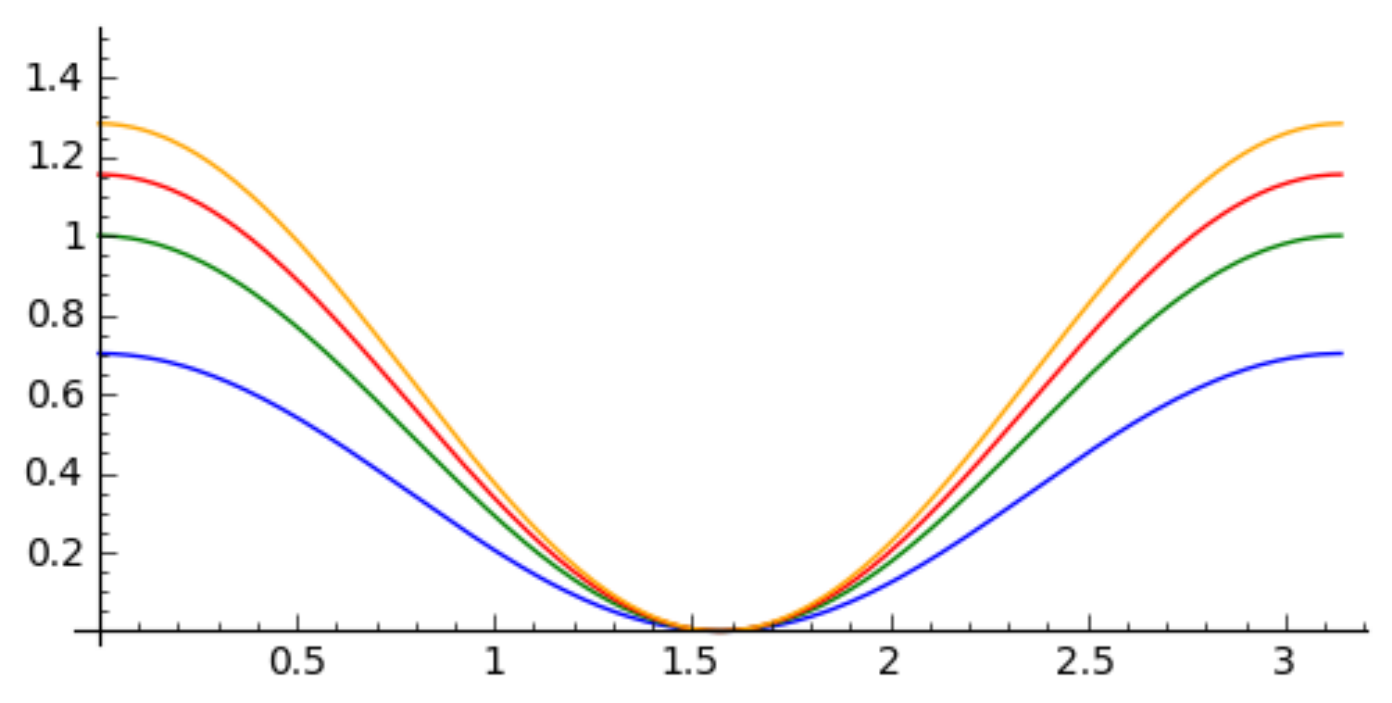}\\ 
\small    Pure splay\ ($K_1=1,$ $K_2=K_3=0$).}
\end{minipage}
\hspace{1cm}
\begin{minipage}{7cm}
\labellist
		\hair 2pt
		\pinlabel $\displaystyle\frac{W(\alpha)}{\pi^2}$ at -25 180
		\pinlabel $\alpha$ at 380 -5
	\endlabellist				
\centering{
 \includegraphics[height=3.5cm]{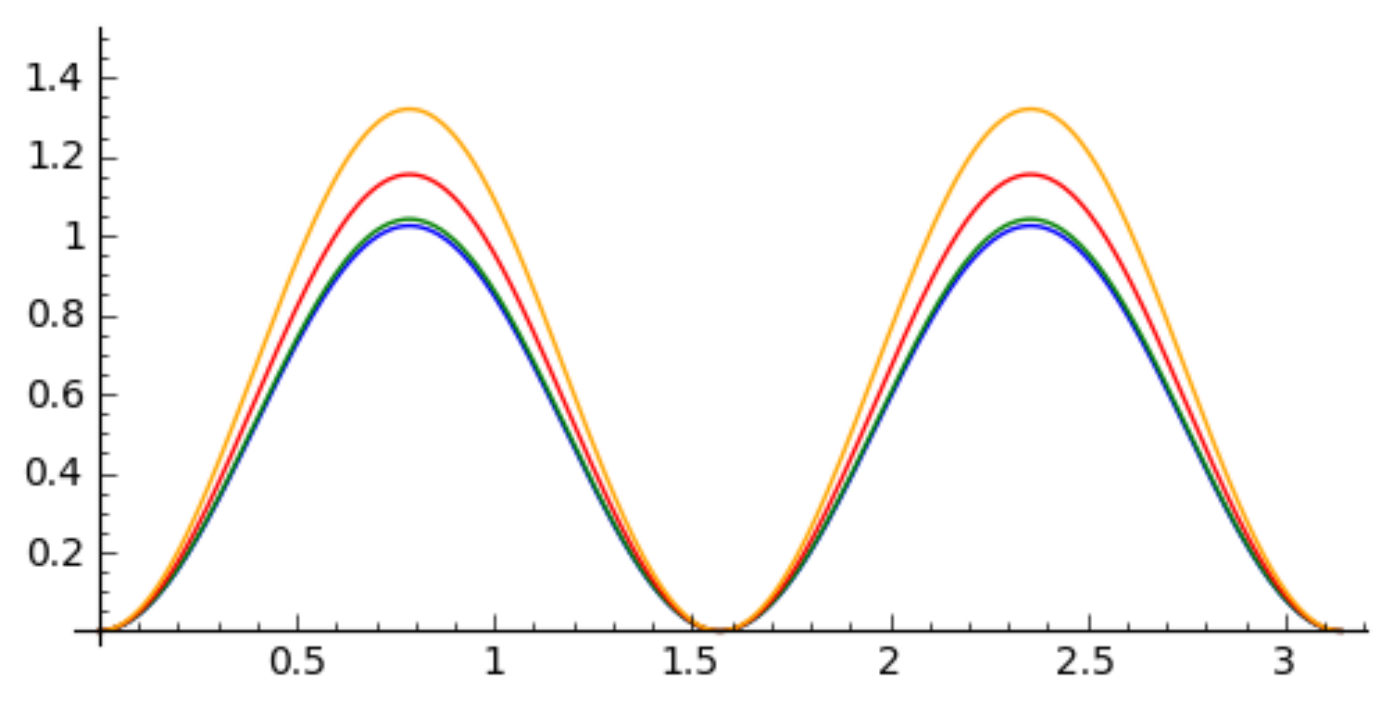}\\  
\small    Pure twist \ ($K_2=1,$ $K_1=K_3=0$).}\normalsize
\end{minipage}
\vspace{1cm}

\begin{minipage}{6.5cm}
\centering{
\labellist
		\hair 2pt
		\pinlabel $\displaystyle\frac{W(\alpha)}{\pi^2}$ at -25 380
		\pinlabel $\alpha$ at 350 -5
	\endlabellist				
	 \includegraphics[height=6.5cm]{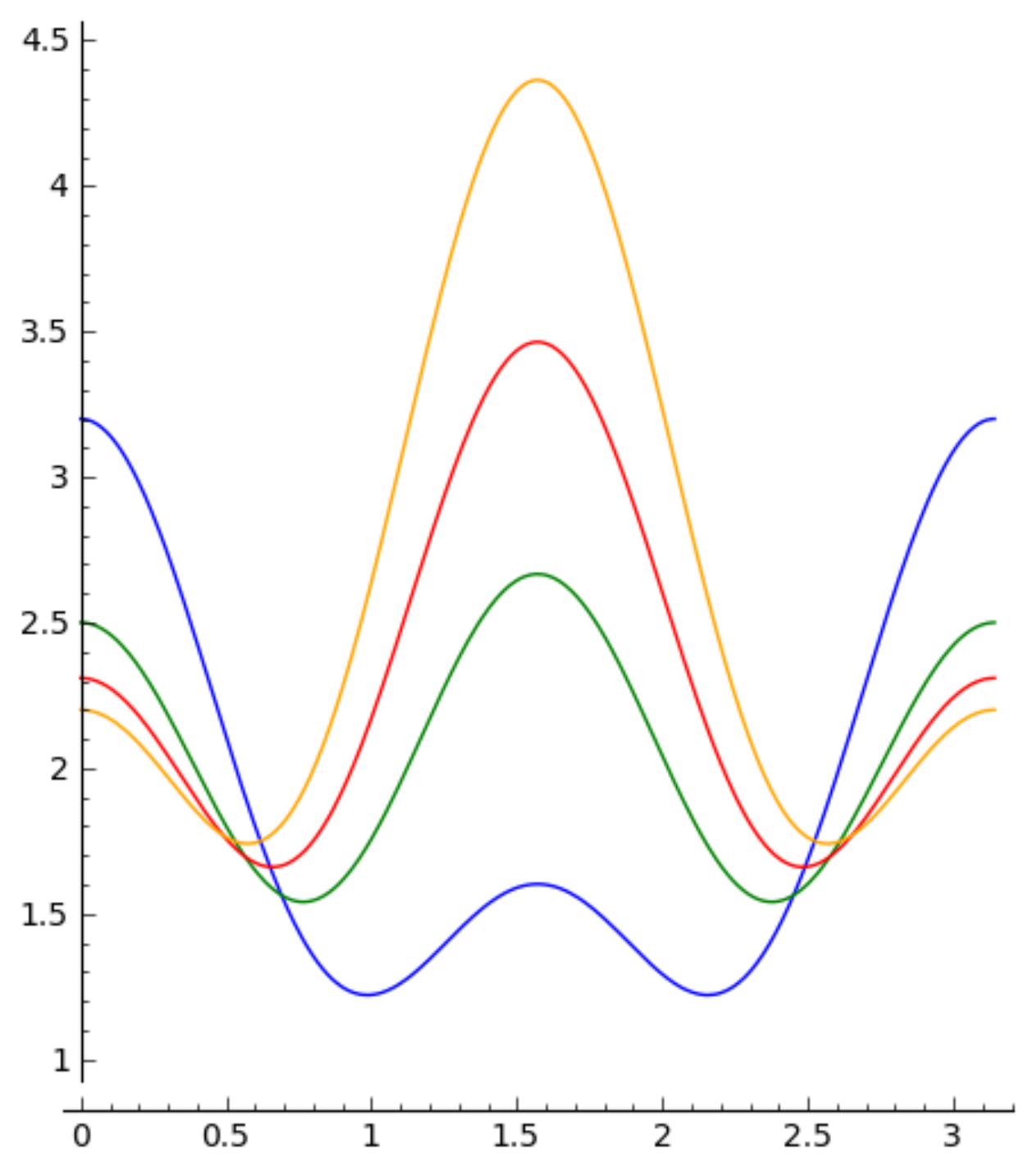}\\ 
\small    Pure bend\  ($K_3=1,$ $K_1=K_2=0$).}
\end{minipage}
\hspace{1cm}
\begin{minipage}{6.5cm}
\labellist
		\hair 2pt
		\pinlabel $\displaystyle\frac{W(\alpha)}{\pi^2}$ at -25 380
		\pinlabel $\alpha$ at 350 -5
	\endlabellist				
\centering{
 \includegraphics[height=6.5cm]{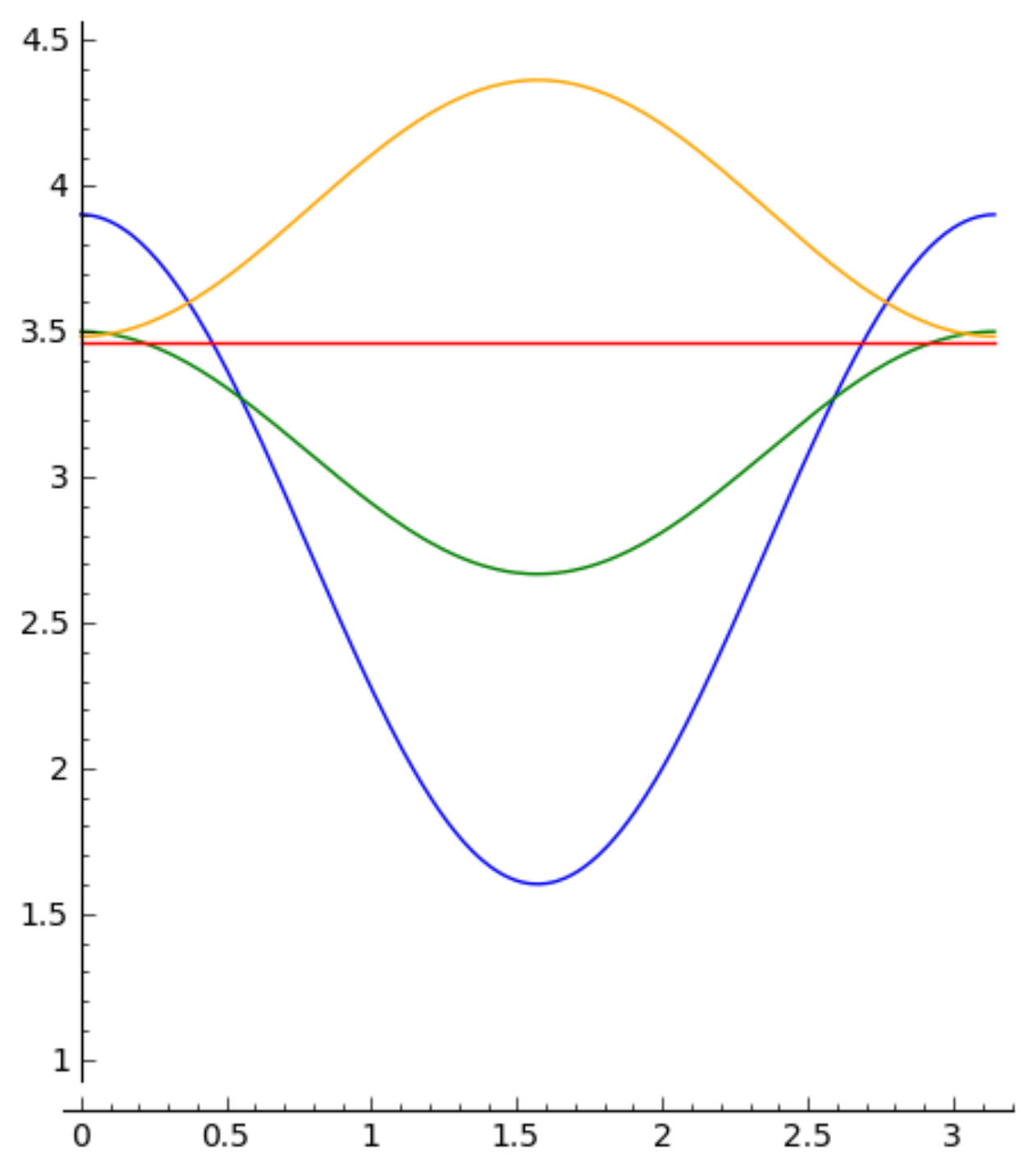}\\  
\small    One-constant approximation\\ \quad ($K_1=K_2=K_3=1)$.}\normalsize
\end{minipage}
\end{center}
\caption{Frank energy $W$ (rescaled by $\pi^2$) as a function of deviation $\alpha$ from $\eu$, for different choices of the parameters $K_i$. The four colours represent four different choices of the ratio $R/r$, namely: $R/r=1.1$ (orange), $R/r= 2/\sqrt 3$ (red), $R/r=1.25$ (green), $R/r=1.6$ (blue).}
\label{fig2}
\end{figure}
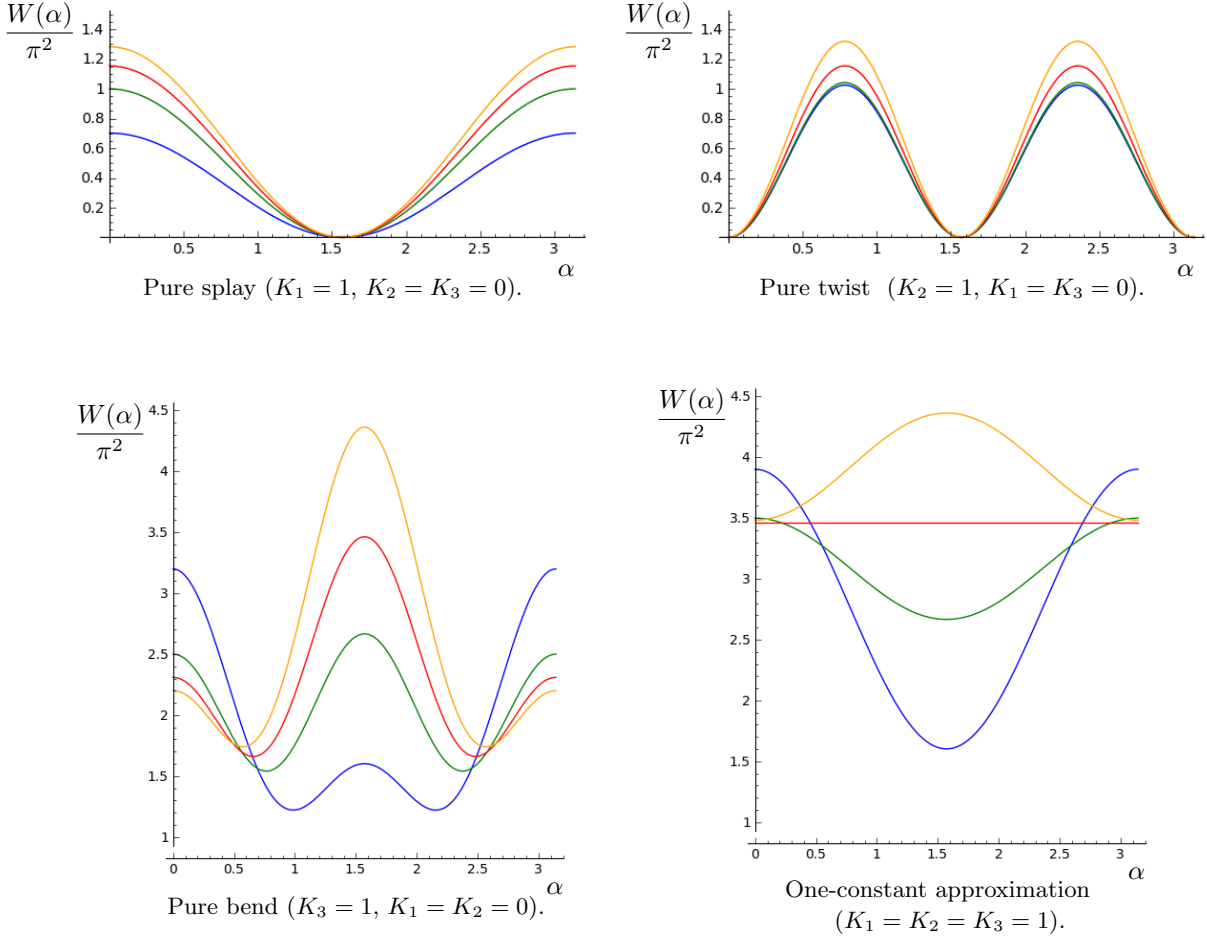

Since we are assuming that $\alpha=const$, instead of the first variation of $W$ we can just take the first derivative with respect to $\alpha$:
\begin{align*}
\frac{\d}{\d\alpha}W(\alpha) &= 2 \pi^2 \sin(2\alpha) \big[A(K_3-K_1) -CK_3\big]+ 2B(K_2-K_3)\cos (2\alpha) \sin (2\alpha)\\
			&=2\sin(2\alpha)\Big( A(K_3-K_1)+B\cos(2\alpha)(K_2-K_3)-CK_3\Big),
\end{align*}
where
\[
	A:=b - \sqrt{b^2-1},\qquad B:=\frac{b^2}{\sqrt{b^2-1}} ,\qquad C:=2 b - \frac{b^2}{\sqrt{b^2-1}}.
\]	
Therefore, $W'(\alpha)=0$ if and only if 
\[
	\sin(2\alpha)=0\qquad \text{or}\qquad  \cos(2\alpha) = \frac{CK_3-A(K_3-K_1)}{B(K_2-K_3)},
\]	
i.e.
\[ 
	\alpha = m\frac \pi 2 \qquad \text{or}\qquad \alpha = \pm \frac 12 \arccos \left( \frac{CK_3-A(K_3-K_1)}{B(K_2-K_3)} \right)+m\pi,
\]	
for $m\in \mathbb Z$, provided the argument of the arccos function is in $[-1,1].$ For short, we refer to the critical points obtained via the arccos function as to points of the \emph{second type}.

To check stability, we compute the second derivative of $W$
\begin{align*}
	\frac{1}{\pi^2}\frac{d^2}{d\alpha^2}W(\alpha) &=4\cos(2\alpha)\Big( A(K_3-K_1)+B\cos(2\alpha)(K_2-K_3)-CK_3\Big) -4B\sin^2(2\alpha)(K_2-K_3)\\
		&=4 A(K_3-K_1)\cos(2\alpha) + 4 B(K_2-K_3)\cos(4\alpha) -4 C K_3\cos(2\alpha).
\end{align*}
Therefore, 
\begin{itemize}
	\item critical points of type $\alpha=m\pi$ are stable local minimizers if
$$ A(K_3-K_1) +B(K_2-K_3)-CK_3>0$$
i.e. if
$$ K_1(\sqrt{b^2-1}-b) +K_2\frac{b^2}{\sqrt{b^2-1}} -K_3(\sqrt{b^2-1}+b) >0,$$
	\item critical points of type $\alpha=(2m+1)\frac\pi 2$ are stable local minimizers if
$$ -A(K_3-K_1) +B(K_2-K_3) +CK_3>0,$$
	\item critical points of the second type are (stable local) minimizers if $K_3>K_2$. 
\end{itemize}	

We make now a special choice of the parameters, in order to be able to plot a stability diagram for the minimizers. Namely, we assume that $K_1=K_3$, $K_2\neq 0$, and we introduce the variables 
\[
	\lambda:=\frac{K_3}{K_2},\qquad \eta:=\frac CB = 2\frac{\sqrt{b^2-1}}{b}-1,
\]	
so that second type minimizers take the form
\[
	\alpha = \pm \frac 12 \arccos \left( \frac{C K_3}{B(K_2-K_3)} \right)=\pm \frac 12 \arccos\left( \eta \frac{\lambda}{1-\lambda}\right).
\]
Note that $\lambda\geq 0$ and, since $b=R/r>1$, then $\eta \in (-1,1)$ and $\eta=0$ if and only if $R/r=2/\sqrt 3$. A necessary condition for $\alpha=m\pi$ to be a stable local minimum for $W$ is then
\[ 
	\frac{B}{B+C}=\frac{2b}{\sqrt{b^2-1}}=\frac{1}{1+\eta}>\lambda.
\]	
A necessary condition for a second type $\alpha$ to be a critical point of $W$ is that
\[ 
	\left| \frac{\lambda}{1-\lambda}\right|\leq \frac{1}{|\eta|},
\]	
while a sufficient condition for a critical point to be a stable local minimum is that $\lambda>1$. Finally, $ \lambda_1:=\frac{1}{1+\eta}$ is a bifurcation point for the unstable critical points of $W$, while $\lambda_2:=\frac{1}{1-\eta}$ is a bifurcation point for the stable global minimizers of $W$.

\begin{figure}[h]
\centering{
\labellist
		\hair 2pt
		\pinlabel $\alpha$ at -10 255
		\pinlabel $\displaystyle\frac{\pi}{2}$ at 16 342		
		\pinlabel $-\displaystyle\frac{\pi}{2}$ at 9 33		
		\pinlabel $\lambda_1$ at 100 171
		\pinlabel $\lambda_2$ at 178 171		
		\pinlabel $\lambda$ at 368 171
		\pinlabel $0$ at 20 171
	\endlabellist				
 \includegraphics[height=6cm]{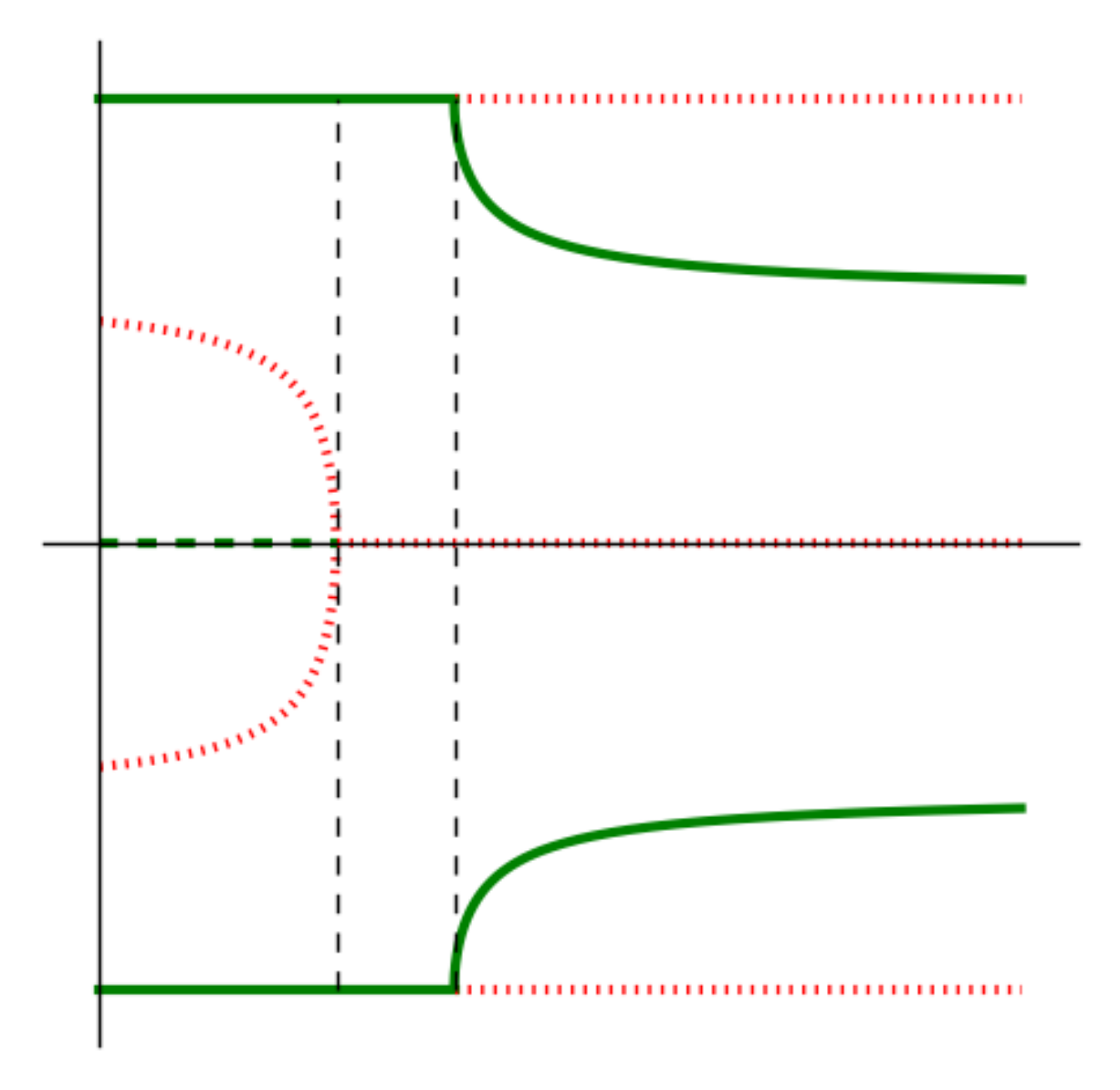}\\ 
}
\caption{Bifurcation diagram for minimizers $\alpha$ of $W$ as a function of $\lambda=K_3/K_2$, for $\lambda \in (0,3.25)$. The other parameters are chosen as $K_1=K_3$, $R/r=1.25$. The diagram shows the stable global minimizer (green continuous line), the stable local minimizer (green dashed line) and the unstable critical points (red dotted lines).}
\label{fig3}
\end{figure}
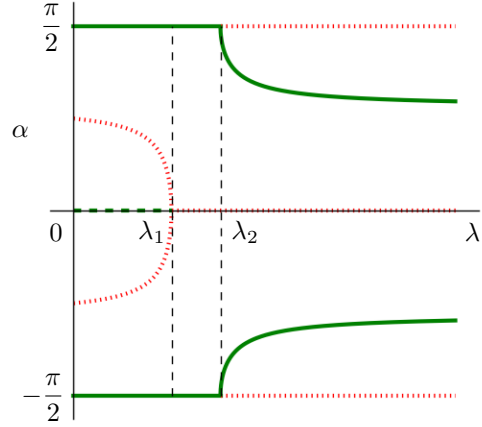

\begin{figure}[h]
\centering{
\labellist
		\hair 2pt
		\pinlabel $\alpha$ at 383 8
		\pinlabel $\displaystyle\frac{W(\alpha)}{\pi^2}$ at 165 215		
		\pinlabel $\lambda\!=\!0.7\lambda_1$ at 330 195
		\pinlabel $\lambda\!=\!\lambda_1$ at 268 165
		\pinlabel $\lambda\!=\!\lambda_2$ at 281 111		
		\pinlabel $\lambda\!=\!2\lambda_2$ at 304 56				
	\endlabellist				
 \includegraphics[height=4.5cm]{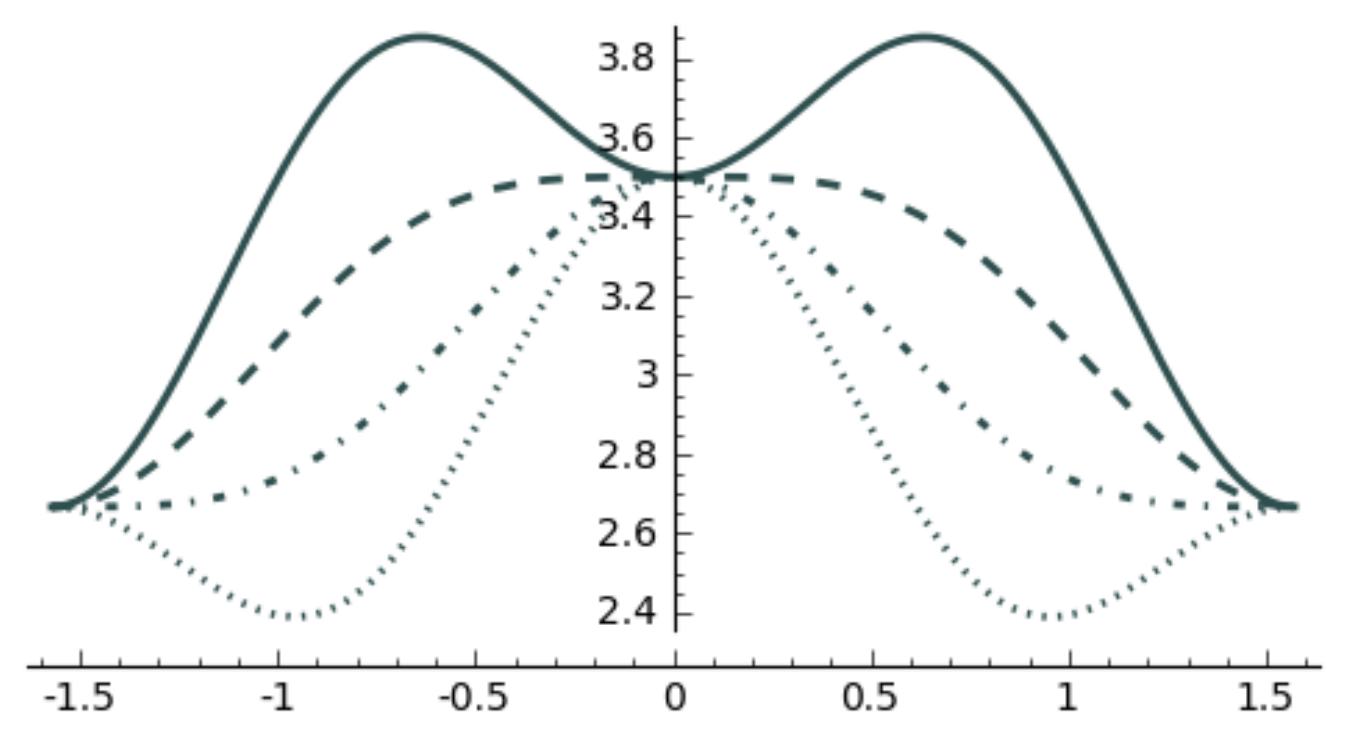}\\ 
}
\caption{Graphs of the energy $W$ (rescaled by $\pi^2$) as a function of $\alpha$, for $R/r=1.25$, $K_1=K_3=1$, and different choices of $\lambda=K_3/K_2$.}
\label{fig3b}
\end{figure}
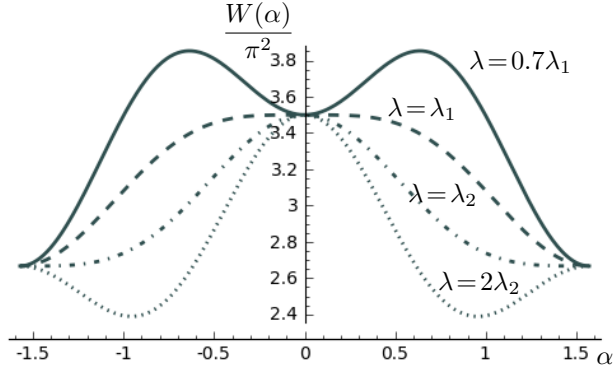


\subsection{The one-constant approximation for the torus}
\label{ssec:oneconst}

Since not every function in $H^1(Q)$ corresponds to a vector field on the torus, before proceeding with the analysis of energy \eqref{eq:napolioc}, we study the structure of the space of configurations $\alpha$. This will enable us to study the gradient flow of the energy functional and to give a geometrical interpretation to its solutions. 

Let $\n\in \HT$ be fixed, and let us assume that $\n$ is also continuous. In general, we cannot expect the corresponding $\alpha$ to be periodic on $Q=[0,2\pi]\times [0,2\pi]$. We observe that the vector field $\n$ is continuous if and only if there exist $m,n\in \Z$ such that
\begin{equation}
\label{eq:stab}
	\alpha(2\pi,\phi)=\alpha(0,\phi) +2m\pi,\qquad 
	\alpha(\theta,2\pi)=\alpha(\theta,0) +2n\pi,\qquad \forall\,(\theta,\phi)\in Q.
\end{equation}
By continuity of $\n$, $m$ and $n$ do not depend on the choice of $\theta$ and $\phi$. Moreover, since $\alpha$ is unique up to an additive constant, $m$ and $n$ are also independent of the choice of $\alpha$ which represents $\n$. In order to extend \eqref{eq:stab} to $H^1$-regular fields, we define the \emph{winding number of $\n$ on $\bbT$} as the couple of indices $h(\n) =(h_\theta,h_\phi)$ given by
\begin{equation}
\label{def:ind}
	h_\theta:= \int_Q\frac{\partial_\theta\alpha(\theta,\phi)}{4\pi^2}\d\theta\,\d\phi,\qquad
	h_\phi:=\int_Q\frac{\partial_\phi\alpha(\theta,\phi)}{4\pi^2}\d\theta\,\d\phi. 	
\end{equation}
\begin{lemma}
Let $\n \in \HT$, then $h(\n) \in \mathbb Z \times \mathbb Z$.
\end{lemma}
\begin{proof}
We prove that $h_\theta\in \Z$, being the proof for $h_\phi$ identical. Let $\n \in \HT$ and let $\n_k$ be a sequence of $C^1$-regular vector fields such that $\n_k \to \n$ in $\HT$.  Let $\alpha_k, \alpha$ be a choice of corresponding representations, such that $\alpha_k \to \alpha$ in $H^1(Q)$. Since $\alpha_k \in C^1(Q)$, it satisfies \eqref{eq:stab} and
\begin{equation}
\label{eq:hint}
	(h_\theta)_k:= \int_0^{2\pi}\left( \int_0^{2\pi}\frac{\partial_\theta\alpha_k(\theta,\phi)}{4\pi^2}\d\theta\right)\d\phi 
		= \frac{1}{2\pi}\int_0^{2\pi} \frac{\alpha(2\pi,\phi)-\alpha(0,\phi)}{2\pi}\d\phi=:
m_k \in \Z.
\end{equation}
 As convergence in $H^1(Q)$ implies weak convergence in $L^1(Q)$ of the derivatives,
\begin{equation}
\label{eq:hint2}
	\lim_{k \to \infty} (h_\theta)_k= \lim_{k \to \infty} \int_Q\frac{\partial_\theta\alpha_k}{4\pi^2}\d\theta\,\d\phi= 	\int_Q\frac{\partial_\theta\alpha}{4\pi^2}\d\theta\,\d\phi = h_\theta.
\end{equation}
By \eqref{eq:hint2} there exists $\bar k\in \N$ such that for all $k>\bar k, |m_k - h_\theta |\leq 1/3$. Therefore for all $k_1,k_2>\bar k$, $|m_{k_1}-m_{k_2}|\leq |m_{k_1}-h_\theta|+|h_\theta - m_{k_2}|\leq 2/3.$ By \eqref{eq:hint}, $m_{k_1},m_{k_2} \in \Z$ and thus $m_{k_1}=m_{k_2}$. We conclude that there exists $m\in \Z$ such that $m_k \equiv m$ for all $k>\bar k$ and finally that $m=\lim_{k \to \infty}m_k=h_\theta \in \Z$.
\end{proof} 
Definition \eqref{def:ind} is also consistent with that of winding of a vector field along a curve given in Section \ref{sec:alpha}: indeed, let $\gamma:[0,2\pi]\to Q$ be given by $\gamma(\theta):=(\theta,0)$, then
\[
	\mathscr W_\gamma(\n)\stackrel{\eqref{def:winde}}{=}\int_\gamma (\iota \circ\n)^* \omega 
		= \int_0^{2\pi} \left\{ \cos(\alpha)\partial_\theta\{\sin(\alpha)\} -\sin(\alpha)\partial_\theta\{\cos(\alpha)\}\right\}_{|\phi=0}\,\d \theta = \int_0^{2\pi} \partial_\theta \alpha(\theta,0)\, \d\theta = 2\pi h_\theta,
\]
that is, $h_\theta = \deg(\iota\circ\n\circ X_{|{\phi=0}},\mathbb S^1, \mathbb S^1)$. An analogous computation holds for $h_\phi$. 
Moreover, 
if $\n,\mathbf{v}\in H^1_{tan}(\bbT;\S^2)$ are homotopic, then $h(\n)=h(\mathbf v)$.

Let $h=(h_\theta,h_\phi)\in \Z^2$, define 
\begin{equation}
\label{def:Ah}
	\mathscr A_h:=\left\{ \alpha\in H^1(Q) : \alpha_{|\{x_j=2\pi\}}=\alpha_{|\{x_j=0\}}+2\pi h_{x_j},\ 
	\text{for }x_j=\theta,\phi\right\},\qquad \mathscr A:= \bigcup_{h\in \Z^2} \mathscr A_h,
\end{equation}
where the equality is in the sense of traces of $H^1$-regular functions. Note that $\mathscr A_0$ and $\mathscr A$ are linear vector spaces, while each $\mathscr A_h$ is an affine space. 
Indeed, for $h=(h_\theta,h_\phi),\ m=(m_\theta,m_\phi) \in\Z^2$, $\alpha \in \mathscr A_h$ and $\beta \in \mathscr A_m$, the function $u(x):=\alpha(x) + \beta(x) \in H^1(Q)$ satisfies
\begin{align*} 
	u_{|\{x_j=2\pi\}}&=\alpha_{|\{x_j=2\pi\}} + \beta_{|\{x_j=2\pi\}} \\
		&\stackrel{\eqref{def:Ah}}{=} \alpha_{|\{x_j=0\}}+2\pi h_{x_j} + \beta_{|\{x_j=0\}}+2\pi m_{x_j}\\
		&=u_{|\{x_j=0\}}+2\pi (h_{x_j} + m_{h_j})
\end{align*}	
in the sense of traces, which implies that $ u= \alpha+\beta \in \mathscr A_{h+m}$, for $h+m=(h_\theta+m_\theta,h_\phi+m_\phi)$. As norm we choose
\begin{equation}
\label{eq:normAh}
	\|\alpha\|_{\mathscr A}:= \left( \int_Q \left\{|\nabla_s \alpha|^2 +\alpha^2\right\}\, \d \textup{Vol}\right)^{\frac12},
\end{equation}
where $\d \textup{Vol}=\sqrt{g}\,\d \theta \d\phi =r(R+r\cos\theta)\d \theta \d\phi$ is the area element induced by the metric $g$ (see Appendix A).
\begin{rem} 
Owing to definition \eqref{def:Ah}, this choice of norm yields $	(\mathscr A_0,\|\cdot \|_\mathscr{A})=H^1_{per}(Q;\textup{Vol}).$ In the remainder of this section, we will alternate between the notations $\mathscr A_0$ and $H^1_{per}(Q)$, depending on the context.
\end{rem} 
Owing to Proposition \ref{prop:alp}, the map $\Phi:\alpha \mapsto \eu\cos \alpha +\ev \sin \alpha $  defines a bijection $ \Phi:\mathscr A/2\pi \Z \to \HT$, and by definition \eqref{def:ind} we have $\Phi^{-1}[\n]\subset \mathscr A_{h(\n)}$.

The Euler-Lagrange equation for the one-constant approximation \eqref{eq:oneconst} can be obtained, of course, by setting $K_1=K_2=K_3=\kappa$ in the corresponding equation for the full energy (see Appendix C). We prefer, though, to derive it from \eqref{eq:oneconst}, which is shorter and more direct. The equations, in the case of the sphere and the cylinder, were derived in \cite{NapVer10, NapVer12L}. Since on $\bbT$ every geometric quantity can be computed explicitly (see Appendix A), we first reduce \eqref{eq:oneconst} to a simpler form.

\begin{lemma} The energy $W_\kappa$, for $\Sigma=\mathbb T$ has the explicit representation
\begin{equation}
\label{eq:expltorus}
	W_\kappa(\alpha)=\frac \kappa 2\int_Q \left\{|\nabla_s \alpha|^2 +\eta\cos(2\alpha)\right\} \dvol  + \kappa \pi^2 \left(\frac{2-b^2}{\sqrt{b^2-1}}+2b\right),
\end{equation}
where $\eta:=\frac{c_1^2-c_2^2}{2}$ and $b:=\frac Rr$.
\end{lemma}
\begin{proof}
Let $\alpha \in \mathscr A_h$, $h=(h_\theta,h_\phi)$. Since $\mathbb A =\frac{\sin \theta}{R+r\cos \theta}\ev$,
\begin{align*}
	\int_Q \nabla_s \alpha \cdot \mathbb A\, \dvol &= \int_0^{2\pi}\int_0^{2\pi} \partial_\phi \alpha(\theta,\phi) \frac{\sin \theta}{R+r\cos \theta} r(R+r\cos \theta)\, \d\theta\,\d\phi\\
	&=\int_0^{2\pi} \left[\int_0^{2\pi} \partial_\phi \alpha(\theta,\phi)\,\d\phi\right] r\sin\theta\,\d\theta 
	=2\pi r h_\phi \int_0^{2\pi} \sin\theta\,\d\theta =0.
\end{align*}
Thus, letting $b=R/r$, 
\begin{align}
	\int_Q |\nabla_s \alpha -\mathbb A|^2 \dvol &= 	\int_Q |\nabla_s \alpha|^2 \dvol + 	\int_Q |\mathbb A|^2 \dvol -2	\int_Q \nabla_s \alpha \cdot \mathbb A\, \dvol \nonumber\\
	&=\int_Q |\nabla_s \alpha|^2 \dvol + 	2\pi\int_0^{2\pi} \frac{r\sin^2 \theta}{R+r\cos \theta} \d \theta 
	\stackrel{\eqref{eq:integr1}}{=}\int_Q |\nabla_s \alpha|^2 \dvol + 4\pi^2\left(b - \sqrt{b^2-1}\right).\label{eq:p1}
\end{align}
Recall that (by Gauss-Bonnet Theorem or by direct computation)
\begin{equation}
\label{eq:gbt}
	\int_{\mathbb T}K\, \dvol = \int_Q c_1c_2\,\dvol=0.
\end{equation}
Using the value of Willmore's functional computed in Lemma \ref{lemma:expl} we get
\begin{equation}
\label{eq:p2}
		\int_Q \frac{c_1^2+c_2^2}{4}\dvol
		\stackrel{\eqref{eq:gbt}}{=}\int_Q \left(\frac{c_1 + c_2}{2}\right)^2\dvol= \pi^2  \left(\frac{b^2}{\sqrt{b^2-1}}\right).
\end{equation}
Substituting \eqref{eq:p1} and \eqref{eq:p2} into \eqref{eq:oneconst} we obtain 
\[
		W_\kappa(\alpha)=\frac \kappa 2\int_Q \left\{|\nabla_s \alpha|^2 +\frac{c_1^2-c_2^2}{2}\cos(2\alpha)\right\}\, \dvol  + \kappa \pi^2 \left(\frac{b^2}{\sqrt{b^2-1}}+2b - 2\sqrt{b^2-1}\right),
\]
using \eqref{eq:ci} and simplifying the last term, we get \eqref{eq:expltorus}.
\end{proof}

\begin{lemma}
\label{lemma:el} 
The Euler-Lagrange equation of \eqref{eq:expltorus} is
 \begin{equation}
\label{eq:el}
	\Delta \alpha +\frac12(c_1^2-c_2^2)\sin (2\alpha)=0.
\end{equation}
 \end{lemma}
 \begin{proof}
 In order to find the Euler-Lagrange equation of \eqref{eq:oneconst}, we compute the first variation in the direction $\omega\in \mathscr A_0$
\begin{align*}
	 \frac{\d}{\d t}W_\kappa(\alpha+t\omega)\Big|_{t=0} &= \frac{\d}{\d t}\frac \kappa2 \int_Q |\nabla_s (\alpha +t\omega)|^2 +\frac12(c_1^2-c_2^2)\cos(2\alpha +2t\omega)\, \dvol\Big|_{t=0}\\
	 &= \kappa \int_Q (\nabla_s \alpha)\cdot \nabla_s \omega -\frac12(c_1^2-c_2^2) \sin (2\alpha)\omega\, \dvol\\
	 &=-\kappa \int_Q \divs (\nabla_s \alpha) \omega -\frac12(c_1^2-c_2^2) \sin (2\alpha)\omega\, \dvol,
\end{align*}	 
which, after integration by parts, yields \eqref{eq:el}.
\end{proof}
We compute also the second variation, in the direction $\omega$
\begin{equation}
\label{eq:secvar}
	 \frac{d^2}{dt^2}W_\kappa(\alpha+t\omega)\Big|_{t=0}	 = \kappa \int_Q | \nabla_s \omega|^2 -(c_1^2-c_2^2) \cos (2\alpha)\omega^2\, \dvol.
\end{equation}

\begin{proposition}
\label{prop:stabil}
Let $b:=R/r$. There exists $b^* \in (2/\sqrt 3 ,2]$ such that the constant values $\alpha=\pi/2+ m\pi$, $m\in \Z$, are local minimizers for $W_\kappa$ in $\mathscr A_0$ if and only if $b \geq b^*$. Moreover, if $b \geq 2$, there exists no non-constant solution $w$ to \eqref{eq:el} such that
\begin{equation}
\label{eq:nononconst}
	\frac{\pi}{2} +m\pi \leq w \leq \frac \pi 2 +(m+1)\pi.
\end{equation}
\end{proposition}

\begin{proof}
Owing to the periodicity of the functions involved, it is not restrictive to assume $m=-1$. By \eqref{eq:secvar}, the second variation of $W_\kappa$, in $\alpha=\pi/2$, in the direction $\omega\in \mathscr A_0$, is  positive if and only if
\begin{equation}
\label{eq:secondvar}
	 \int_Q |\nabla_s\omega|^2 +(c_1^2-c_2^2)\omega^2\, \dvol >0.
\end{equation}
Let $b=R/r>1$, since 
\begin{equation}
\label{eq:cucd}
	c_1^2-c_2^2 \stackrel{\eqref{eq:ci}}{=} \left(\frac{1}{r^2}- \frac{\cos^2\theta}{(R+r\cos\theta)^2}\right) =\frac{b}{r^2(b+\cos \theta)^2}(b+2\cos \theta),
\end{equation}
we see immediately that if $b\geq 2$ then $c_1^2-c_2^2 \geq 0$ everywhere in $Q$, and $c_1^2-c_2^2 = 0$ if and only if $b=2$ and $\theta=\pi$. Therefore, if $b\geq 2$, the integral in \eqref{eq:secondvar} is nonnegative for all $\omega \in \mathscr A_0$  (equal to zero if and only if $\omega=0$) and we can conclude that the stationary point $\alpha =\pi/2$ is a local minimum.  Restricting to constant variations $\omega$, \eqref{eq:secondvar} is satisfied if and only if
\[	
	0 <  \int_Q (c_1^2-c_2^2)\, \dvol = 2\pi b \int_0^{2\pi} \frac{b+2\cos \theta}{b+\cos \theta} \, d\theta =4\pi^2b \left(2- \frac{b}{\sqrt{b^2-1}}\right)
\]	
(see Appendix B for the integration formula), that is, if and only if $b> 2/\sqrt 3$. If $b= 2/\sqrt 3$, then all configurations with constant angle $\alpha(x)=\bar \alpha$ have the same energy, while for $b<2/\sqrt 3$, $W_\kappa(\alpha\equiv 0) < W_\kappa(\alpha\equiv \pi/2)$. The uniqueness of the bifurcation point $b^*$ follows from the monotonicity of $(c_1^2-c_2^2)\text{Vol}$ with respect to $b$:
$$ \frac{\partial}{\partial b}(c_1^2-c_2^2)\text{Vol} = \frac{\partial}{\partial b}\left\{ \frac{b^2+2b\cos \theta}{b+\cos \theta}\right\} = 1+ \frac{\cos^2 \theta}{(b+\cos \theta)^2} >0,\qquad \forall\, \theta\in [0,2\pi],\quad \forall\,b>1.$$   
The proof of the last step of the statement of Proposition \ref{prop:stabil} is inspired by \cite[Theorem 2.4]{ChenEtal04}. Assume that $b\geq 0$ and let $w$ be a solution to \eqref{eq:el}, satisfying \eqref{eq:nononconst} for $m=-1$. Then $v(x):=\pi/2 -w(x)$ satisfies
\begin{equation}
\label{eq:chen}
	\Delta v = -\Delta w =\frac12 (c_1^2-c_2^2)\sin(2w)
		=\frac12 (c_1^2-c_2^2)\sin(\pi -2v)=\frac12 (c_1^2-c_2^2)\sin(2v).   
\end{equation}
Multiplying the first and the last member of \eqref{eq:chen} by $v$, and integrating on $Q$ with respect to $\dvol$, after integration by parts we obtain
\begin{align*}
	-\int_Q |\nabla_s v|^2\, \dvol &= \int_Q \frac12 (c_1^2-c_2^2)\sin(2v)v\, \dvol
		\stackrel{\eqref{eq:cucd}}{\geq} \frac{b(b-2)}{2(b+1)}\int_Q \sin(2v)v\, \d\theta\,\d\phi  
			\stackrel{\eqref{eq:nononconst}}{\geq} 0.
\end{align*}	   
Thus,
\[ 
	\int_Q |\nabla_s v|^2 \,\dvol= 0,\quad \text{and}\quad \int_Q \sin(2v)v\, \d\theta\,\d\phi=0,
\]	
implying $v\equiv 0$, $v\equiv - \pi/2$ or $v\equiv \pi/2$, as we wanted to prove. 
\end{proof}
\bigskip

In order to find a numerical minimizer of $W$, we study the $L^2$-gradient flow of \eqref{eq:oneconst}, that is, we want to find $\alpha \in C^0([0,+\infty);\mathscr A)$ such that
\begin{equation}  
\label{eq:gradient}
	\partial_t \alpha = \kappa\Delta \alpha +\frac \kappa2(c_1^2-c_2^2)\sin (2\alpha),\quad \hbox{on } \mathbb{R}^2\times (0,+\infty)
\end{equation}
with suitable initial data $\alpha_0 \in \mathscr A$. This equation can also be obtained directly from \eqref{eq:gfn}, through the substitution $\n=\eu\cos\alpha +\ev\sin\alpha$, by using Lemma \ref{lem:rough_alpha} and recalling that $\Div \mathbb{A} = 0$, as we detail at the end of Appendix A. As above, denote $\Phi:\alpha \mapsto \n=\eu\cos\alpha +\ev\sin\alpha.$ Since the index of a vector field $h(\Phi[\alpha])$ is invariant under homotopy, if $\alpha_0 \in \mathscr A_h$, then $\alpha(t)\in \mathscr A_h$ for all $t>0$. The spaces $\mathscr A_h$ (see \eqref{def:Ah}) are constructed to take care of the correct boundary conditions, which require some attention, since in general we cannot expect a periodic solution.

Exploiting the affine structure of $\mathscr A$,  for any $h\in \Z^2$, for any fixed $\psi_h \in \mathscr A_h$, it holds $\mathscr A_h=\mathscr A_0 +\psi_h$, i.e., any $\alpha \in \mathscr A_h$ can be decomposed as
\[ 
	\alpha(x) = u(x) + \psi_h(x), \qquad \text{with }u\in \mathscr A_0.
\]	
Using the decomposition $ \alpha(t,x) = u(t,x) + \psi_h(x)$, we see that problem \eqref{eq:gradient} is equivalent to finding $u\in C^0([0,+\infty);\mathscr A_0)$ such that
\begin{equation}
\label{eq:gradient2}
	\partial_t u - \kappa\Delta u= \kappa\Delta \psi_h +\frac \kappa2(c_1^2-c_2^2)\sin (2u +2\psi_h) \hbox{ on } Q\times (0,+\infty),
\end{equation}
with initial condition $u_0 \in \mathscr A_0$ and where $h\in \Z^2$ is the constant degree of the mappings $\Phi[\alpha(t)]$. Equation \eqref{eq:gradient2} can be further simplified by choosing a $\Delta$-harmonic function $\psi_h$, so that the term $\kappa \Delta \psi_h$ vanishes. 

\begin{lemma} 
\label{lemma:harmonic}
Let $h:=(h_\theta,h_\phi)\in \Z^2$, and let $b=R/r$, where $R>r>0$ are the radii of the torus, as in \eqref{eq:paramtorus}.  Define
\begin{equation}
\label{eq:psih}
	\psi(\theta,\phi):=h_\theta\sqrt{b^2-1}\int_0^\theta \frac{1}{b+\cos(s)}\d s +h_\phi \phi.
\end{equation}
Then $\psi \in C^\infty(\R^2),$ $\psi_{|Q}\in \mathscr A_h $, and $\Delta\psi=0$.
\end{lemma}
\begin{proof}
Since $b>1$, $\psi \in C^\infty(\R^2)$ 
 and a simple check, using the explicit expression of the Laplace-Beltrami operator
  on the torus \eqref{eq:deltas} shows that $\Delta \psi=0$. In order to check
   that $\psi \in \mathscr A_h$, according to definition \eqref{def:Ah}, we use
    the $2\pi$-periodicity of $1/(b+\cos(s))$ and the explicit integration $\sqrt{b^2-1}=\int_0^{2\pi}1/(b+\cos(s))$ (see Appendix B)  to compute
\begin{align*}
	\psi(\theta+2\pi,\phi+2\pi) &=h_\theta\sqrt{b^2-1}\int_0^{\theta+2\pi} \frac{1}{b+\cos(s)}\d s +h_\phi (\phi+2\pi)\\
		&=h_\theta\sqrt{b^2-1}\int_0^{2\pi} \frac{1}{b+\cos(s)}\d s +h_\phi 2\pi + h_\theta\sqrt{b^2-1}\int_{2\pi}^{2\pi+\theta} \frac{1}{b+\cos(s)}\d s +h_\phi \phi \\
		&=h_\theta 2\pi +h_\phi 2\pi + h_\theta\sqrt{b^2-1}\int_{0}^{\theta} \frac{1}{b+\cos(s)}\d s +h_\phi \phi \\
		&=h_\theta 2\pi +h_\phi 2\pi + \psi(\theta,\phi).
\end{align*}
\end{proof}
We now have all the ingredients to state and prove the result regarding solutions to the $L^2$-gradient flow of the one-constant approximation of the surface elastic energy $W_\kappa$.  
\begin{theorem}
\label{Th:gf}
Let $X$ be the parametrization of the torus \eqref{eq:paramtorus} with radii $R,r$, embedded in $\R^3$. Let $\mathscr A_h, \mathscr A$ be the spaces defined in \eqref{def:Ah}, endowed with the norm \eqref{eq:normAh}. Then
\begin{itemize}
	\item[(0)] For all $h\in \Z^2$ there exists a classical solution $\alpha\in \mathscr A_h \cap C^\infty(Q)$ to the stationary problem
\begin{equation}
\label{eq:stationary}
	-\kappa\Delta \alpha =\frac \kappa2(c_1^2-c_2^2)\sin (2\alpha).
\end{equation}
	Moreover, $\alpha$ is odd on any line passing through the origin. 
	\item[(i)] \textup{(Weak well-posedness)} For any $\alpha_0\in \mathscr A$, for all $T>0$, there exists a unique 
		mild solution $\alpha$ to \eqref{eq:gradient} 	and it satisfies
 		\[
 			\alpha\in C^0([0,T);\mathscr A).
 		\]
		Moreover, if $\alpha_0 \in \mathscr A_h$, then $\alpha(t) \in \mathscr A_h$ for all $t>0$.
	\item[(ii)] \textup{(Strong well-posedness)} For any $m\in \N$, for any $\alpha_0\in H^{2m}(Q)\cap \mathscr A$, 	for all $T>0$, the unique solution $\alpha$ to \eqref{eq:gradient} satisfies 
		\begin{equation}
		\label{eq:regostrong}
			\alpha \in \bigcap_{k=0,\ldots,m}C^k([0,T];H^{2m-2k}(Q)).
		\end{equation}
		In particular, if $\alpha_0\in C^\infty(Q)\cap \mathscr A$, then $\alpha \in C^\infty([0,T]\times  Q)$.
	\item[(iii)] \textup{(A maximum principle)} Under the hypothesis of step (ii),	
		\begin{equation}
		\label{eq:unifbound}
			\alpha \in L^\infty(0,+\infty;\mathscr A)\quad \text{and}\quad \partial_t \alpha \in L^2(0,+\infty;L^2(Q)).
		\end{equation}
	\item[(iv)] \textup{(Long-time behaviour)} Define the omega-limit set of a solution $\alpha$ to \eqref{eq:gradient} by
		\begin{equation*}
			\omega(\alpha):=\left\{\alpha_\infty \in \mathscr A : \hbox{ there exists } t_n\nearrow +\infty \hbox{ with }
			\alpha(t_n)\to \alpha_\infty \hbox{ in } L^2(Q) \right\}.
		\end{equation*}
	 	Under the hypothesis of (ii), 
		the omega-limit set is nonempty and it is contained in the set 
		of solutions to \eqref{eq:el}, namely if 
		$\alpha_\infty\in \omega(\alpha)$ then $\alpha_\infty$ is a solution 
		of \eqref{eq:el}.
\end{itemize}
\end{theorem}

\begin{proof}
The idea of the proof is that using the decomposition $\mathscr A_h= \mathscr A_0 +\psi_h$, we can reduce the problem of finding a solution to \eqref{eq:gradient} in $C^0([0,T];\mathscr A_h)$, with initial value $\alpha_0$, to the simpler problem of finding $u\in C^0([0,T];\mathscr A_0)$ such that
\begin{equation}
\label{eq:gradient3}
	\partial_t u - \kappa\Delta u= \frac \kappa2(c_1^2-c_2^2)\sin (2u +2\psi_h) \hbox{ on } Q\times (0,+\infty),
\end{equation}
and $u(0)=\alpha_0-\psi_h$. The term $\Delta \psi_h$ disappears by choosing the harmonic function $\psi_h$ defined in \eqref{eq:psih}. Therefore, through the proof, let $h\in\Z^2$ be fixed, and let $\psi_h$ be given by \eqref{eq:psih}. Moreover, in order to make the symmetry properties of the involved functions more visible, we redefine $Q:=(-\pi,\pi)\times (\pi,\pi)$.


\paragraph{Step (0)} Define the Hilbert space $H:=L^2(Q)$, with the scalar product 
$${( u,v )}_H:=\int_Q uv\, \dvol,\qquad u,v\in L^2(Q),$$
and denote the average of a function $u\in H$ by $\langle u\rangle:= \frac{1}{4\pi^2 R r}\int_Q u\, \dvol. $
Let $V:=\{v\in H^1_{per}(Q): \langle v\rangle=0\}$, then, by Wirtinger's inequality, the bilinear form $a:V \times V \to \R$
$$ a(u,v):=\kappa\int_Q \nabla_s u \cdot \nabla_s v\,\dvol$$
defines a scalar product on $V$, such that the induced norm is equivalent to the standard Sobolev norm of $H^1(Q)$ defined in \eqref{eq:normAh}. By Riesz-Fr\'echet representation Theorem \cite[Theorem 5.5]{Brezis2011}, for all $f\in H$ there exists a unique $u\in V$ such that
\begin{equation}
\label{eq:lax}
	a(u,v)={( f,v)}_H\qquad \forall\,v\in V
\end{equation}
and there exists a constant $C_a>0$, depending only on $Q,\kappa$, and on the ellipticity constant of $a$, such that ${\|u\|}_{\mathscr A}\leq C_a{\|f\|}_H$. Moreover (see, e.g., \cite[Section 9.6]{Brezis2011}), if $f\in H^m_{per}(Q)$, then $u\in H^{m+2}_{per}(Q)$; in particular, if $m>1$, then $u\in C^2(\overline{Q})$. For $f\in V$, the solution $u$ to \eqref{eq:lax} satisfies, for all $w \in H^1_{per}(Q)$
\begin{align*}
	\int_Q \nabla_s u \cdot \nabla_s(w-\langle w\rangle) \dvol &= \int_Q f(w-\langle w\rangle)\dvol\\
		&=\int_Q fw\, \dvol -\frac{1}{|Q|} \int w\,\dvol \int f\,\dvol\\
		&=\int_Q (f-\langle f\rangle)w\, \dvol,		
\end{align*}
that is,
$$ -\Delta u =f-\langle f\rangle\qquad \text{on }Q.$$
Let $\eta(x):=\frac{\kappa}{2}\left(c_1^2-c_2^2(x)\right)$, note that $\eta\in C^\infty_{per}(\overline Q)$. Let $f:H^1_{per}(Q) \to H^1_{per}(Q)$ be defined by $ f(u)(x):=\eta(x)\sin(2u(x)+2\psi_h(x)),$ and consider the operator $\mathcal T : H^1_{per}(Q) \to H^1_{per}(Q)$ which maps $v$ into the unique solution $u \in V$ to
$$-\Delta u =f(v)-\langle f(v) \rangle.$$
By a standard bootstrapping argument (see, e.g., \cite[Section 9.6]{Brezis2011}), $u\in C^\infty(\overline Q)$. In order to find a stationary solution to \eqref{eq:gradient3}, we need to find a fixed point $u^*=\mathcal T(u^*)$, such that $\langle f(u^*) \rangle=0$. 
We say that a function $F:Q\to \R$ is \emph{2-even} if $F(\theta,\phi)=F(-\theta,-\phi)$, and we say that it is \emph{2-odd} if $F(\theta,\phi)=-F(-\theta,-\phi)$, for all $(\theta,\phi)\in Q$. It is immediate to check that 
\begin{enumerate}
	\item if $F\in L^1(Q)$ is 2-odd, then $\int_Q F(\theta,\phi)\,\d\theta\,\d\phi=0$;
	\item if $F$ is 2-odd and $G$ is 2-even, then $FG$ is 2-odd;
	\item if $F$ is 2-odd and $\varphi:\R\to\R$ is odd, then $\varphi \circ F$ is 2-odd;
	\item if $F\in C^1(Q)$ is 2-odd (even), then $\partial_i F$ is 2-even (odd).
\end{enumerate}	
(To check the last property, note that a function is 2-odd (even) if and only if its restriction to a line passing by the origin is odd (even). Denote $x=(\theta,\phi)$, $\nb:=x/|x|$, then if $F$ is odd $\nabla F(x)\cdot \nb=-\nabla F(-x)\cdot \nb$, owing to the corresponding property for 1-d functions.) By the definitions of  $\dvol$ \eqref{eq:normAh}, $\psi_h$ \eqref{eq:psih}, $c_1,c_2$ \eqref{eq:ci}, and $\Delta$ \eqref{eq:deltas}, we see that $\eta$ and $\dvol$ are 2-even, $\psi_h$ is 2-odd, and if $u$ is 2-odd, then $f(u)$ and $\Delta u$ are 2-odd. Instead, if $\Delta u$ is 2-odd, we cannot conclude that $u$ is 2-odd. We resort to the projection $\mathcal P$ of a function onto its 2-odd part
$$ \mathcal Pu (\theta,\phi):=\frac{u(\theta,\phi)-u(-\theta,-\phi)}{2}.$$
Then, letting $Id$ be the identity operator in $H$,
$$ (Id-\mathcal P)u (\theta,\phi)=\frac{u(\theta,\phi)+u(-\theta,-\phi)}{2}$$ 
is 2-even. $\mathcal P$ is linear and continuous with respect to the topology of $H$:
$$ {\|\mathcal P v_1 -\mathcal Pv_2\|}_H\leq {\| v_1 -v_2\|}_H,\qquad \forall\,v_1,v_2\in H.$$
Note that for all $v\in V$ 
\begin{equation}
\label{eq:boundf}
	{\|f(v)\|}_H=\left(\int_Q (\eta \sin(2v+2\psi_h))^2\dvol\right)^{1/2}\leq {\|\eta\|}_H =:M.
\end{equation}
Let $\mathcal K:= \left\{ v\in V : {\|v\|}_{\mathscr A}\leq C_a M\right\}.$ The set $\mathcal K$ is a convex and nonempty subset of $H$, moreover, by Rellich-Kondrachov Theorem \cite[Theorem 9.16]{Brezis2011}, it is compactly embedded in $H$. The operator $\mathcal T \circ \mathcal P$ maps a function $v\in \mathcal K$ into the function $u\in  H$ which is the unique solution (in $V$) to
$$ -\Delta u =f(\mathcal P v) -\langle f(\mathcal P v)\rangle= f(\mathcal P v).$$
Moreover, by \eqref{eq:boundf}, $u\in \mathcal K$. The mapping $\mathcal T\circ \mathcal P$ is also continuous, with respect to the topology of $H$:
 \begin{align*}
 	{\|\mathcal T \circ \mathcal P (v_1) - \mathcal T \circ \mathcal P (v_2)\|}_{\mathscr A} 
		&\leq C_a{\|f(\mathcal P v_1)-f(\mathcal P v_2)\|}_H\\
		&=C_a\left(\int_Q (\eta \sin(2\mathcal P v_1+2\psi_h)-\eta \sin(2\mathcal Pv_2+2\psi_h))^2\dvol\right)^{1/2}\\
		&\leq C_a{\|\eta\|}_\infty\left(\int_Q (2\mathcal P v_1+2\psi_h- 2\mathcal Pv_2 -2\psi_h)^2\dvol\right)^{1/2}\\
		&\leq 2C_a{\|\eta\|}_\infty{\|v_1-v_2\|}_H. 		
\end{align*}	
By Schauder fixed point Theorem \cite[p. 56]{Zeidler86}, we conclude that there exists $u^*\in \mathcal K$ such that $u^* = \mathcal T \circ \mathcal P(u^*)$, that is
$$ -\Delta u^* = f(\mathcal P u^*).$$
Since $\Delta$ and $\mathcal P$ commute (owing to the symmetry of $\Delta$), we have
$$ 0 =(Id-\mathcal P)f(\mathcal P u^*)=-(Id -\mathcal P)\Delta u^* = -\Delta (Id-\mathcal P)u^*,$$
that is, we can decompose $u^*$ into a 2-odd and a 2-even function 
$$u^*=\mathcal P u^* + (Id-\mathcal P)u^*$$
such that
$$ -\Delta \mathcal P u^* = f(\mathcal P u^*),\qquad -\Delta(Id-\mathcal P) u^*=0.$$
By the strong maximum principle \cite[Section 6.4.2, Theorem 3]{Evans02} and the periodicity of $u^*$ on $Q$, $(Id-\mathcal P) u^*$ is constant. 
Since 
$$ \langle(Id-\mathcal P) u^*\rangle = \langle u^* \rangle - \langle \mathcal P u^*\rangle=0,$$
we conclude that $(Id-\mathcal P)u^*=0$, hence $\mathcal P u^* =u^*$. We have thus proved that 
$$ -\Delta  u^* =f( u^*).$$
The function $\alpha_k:=u^* +\psi_h +k\pi \in \mathscr A_h$ is a solution to the stationary problem for every $k\in \Z$. The regularity of $\alpha_k$ follows directly from the $C^\infty$ regularity of $u^*$ and $\psi_h$.


\paragraph{Step (i)} The Laplace-Beltrami operator  on the torus, defined in \eqref{eq:deltas}, is a linear second order differential operator, with $C^\infty$-regular and bounded coefficients. It is uniformly  elliptic, with ellipticity constant
$\mu:=\min\{1/r^2,1/(R-r)^2\}$. In order to prove existence and uniqueness of solutions to \eqref{eq:gradient3}, we exploit the powerful machinery of analytic semigroups, developed in \cite{Lunardi95}. We only need to show that our problem fits in the framework. 

Let $D(A):=H^2_{per}(Q)$. Note that $D(A)$ is dense in $H$ and in $H^1_{per}(Q)$, and that the realization of the Laplace-Beltrami operator $A:D(A)\to H$, $Au:=\kappa\Delta u$, is self-adjoint and dissipative. Therefore, $(A,D(A))$ is a sectorial operator (\cite[Proposition 2.2.1]{LLMP05}) and it generates the analytic semigroup $e^{tA}:H\to H$.  
 For all $u,v \in H^1_{per}(Q)$ 
\begin{align}
	{\|f(u)-f(v)\|}_{H}&\leq {\| \eta\|}_\infty {\|u-v\|}_H \leq C{\|u-v\|}_{\mathscr A}. \label{eq:lunardi2}
\end{align}
For $T>0$, a continuous function $u:(0,T]\to H^1_{per}(Q)$ such that $t\mapsto f(u(t)) \in L^1(0,T;H)$ is said to be a \emph{mild solution} of 
\begin{equation}
\label{eq:abstract}
	\partial_t u =Au+f(u),\qquad u(0)=u_0\in H,\qquad \text{on }(0,T),
\end{equation}
if 
$$ u(t)=e^{tA}u_0 +\int_0^te^{(t-s)A}f(u(s))\,\d s,$$
where integration is in the sense of Bochner (see, e.g. \cite[Chapter 3]{Showalter97}). By \cite[Theorem 7.1.3 (i) and Proposition 7.1.8]{Lunardi95}, \eqref{eq:boundf} and \eqref{eq:lunardi2}, for every initial datum $u_0\in H^1_{per}(Q)$, for every $T>0$, there exists a unique mild solution $u\in C^0([0,T];H^1_{per}(Q))$.  The winding number of the vector field $\n(t)=\cos(u(t)+\psi_h)\eu + \sin(u(t)+\psi_h)\ev$ is then $\mathscr W(\n(t))\equiv h$ along the flow.


\paragraph{Step (ii)} If $u_0 \in D(A)$ and $Au_0+f(u_0)\in H$, then 
\begin{equation}
\label{eq:maxreg1}
	u\in C^0([0,T];D(A))\cap C^1([0,T];H)
\end{equation}
 and $u$ solves \eqref{eq:abstract} pointwise, for all $t\in [0,T]$ \cite[Proposition 7.1.10 (iii)]{Lunardi95}. 

More in general, parabolic equations governed by a strongly uniformly elliptic operator $A$ with $C^\infty$-regular coefficients obey the following maximal regularity principle: the terms $\partial_t u$ and $Au$ have \emph{independently} the same regularity as $f(u)$, provided that the initial datum and the boundary conditions (if present) are smooth enough, see, e.g., \cite[Theorem 6, Section 7.1]{Evans02}, \cite[p.341--343]{Brezis2011}, or \cite{CanVes86}. In our case, since $f$ is Lipschitz-continuous and bounded, from \eqref{eq:maxreg1} we read that  $f(u)\in C^0([0,T];H^2(Q))\cap C^1([0,T];L^2(Q))$, and by the maximal regularity principle we obtain that $u\in C^0([0,T];H^4(Q))\cap C^1([0,T];H^2(Q))$. Iterating this process we obtain the regularity \eqref{eq:regostrong} and eventually, provided we choose an initial datum $u_0\in C^\infty(Q)\cap H^1_{per}(Q)$, for all $T>0$ we obtain a $Q$-periodic function $u\in C^\infty([0,T]\times \overline Q)$ (see, e.g., \cite[Theorem 7, Section 7.1]{Evans02}). Reconstruction of $\alpha$ is done as before by $\alpha(t,x):=u(t,x)+\psi_h(x)$.


\paragraph{Step (iii)}  Let $u\in C^2([0,T]\times \overline Q)$ be a solution to \eqref{eq:gradient3}, as in the previous step. We prove the uniform bound \eqref{eq:unifbound} by showing that there exists a constant $C>0$, independent of time, such that 
\begin{equation}
\label{eq:ubound2}
	\sup_{T>0} \|u(T)\|_\infty<C\qquad \text{and}\qquad   \sup_{T>0} \left\{ {\|\partial_t u\|}_{L^2(0,T;H)}+  {\|\nabla_su(T)\|}_H\right\}\leq C.
\end{equation}
Note that, if $u\in \mathscr A_0 \cap C^2(Q)$ has a local maximum in $x_0 \in \overline Q$, then $\nabla u(x_0)=0$ and $\Delta u(x_0)\leq 0$. We remark that the inequality is valid also in points belonging to $\partial Q$, owing to the periodicity of $u$. Since the coefficients of the second-order derivatives of the Laplace-Beltrami operator are positive,  $\Delta u(x_0)\leq 0$. 

Equipped with this regularity, we can use the maximum principle for parabolic semilinear problems \cite[Proposition 6.2.5]{LLMP05} to establish boundedness of $u$. Let $u_0\in C^2(\overline Q)$ be the initial datum for $u$. Let $u^*$ be a solution to the stationary problem \eqref{eq:stationary} as in Step (0). Since $u_0$ and $u^*$ are bounded, there exist $m_1,m_2\in\N$ such that
$$ u^*(x) +m_1\pi \leq u_0(x) \leq u^*(x) +m_2\pi,\qquad \forall\,x\in \overline Q.$$ 
Define $v_1(t,x):=u^*(x) +m_1\pi$, $v_2(t,x):=u^*(x) +m_2\pi$. Then $v_1$ and $v_2$ satisfy
\begin{equation*}
	\partial_t v_1 = Av_1 +f(v_1),\qquad \partial_t v_2 = Av_2 +f(v_2)\qquad \forall\,t\in [0,T].
\end{equation*}	
By \cite[Proposition 6.2.5]{LLMP05}, 
$$v_1(t,x)\leq u(t,x)\leq v_2(t,x)\qquad \text{for all}\quad t,x \in [0,T] \times \overline\Omega.$$ 
Since the estimate does not depend on $T$, we obtain the first half of \eqref{eq:ubound2}.

Regarding the second half, we take the scalar product (in $H$) of \eqref{eq:gradient3} times $\partial_t u$, obtaining
$$ \int_Q (\partial_t u(t))^2\,\dvol+ a(u(t),\partial_t u(t)) = \int_Q f(u(t))\partial_t u(t)\, \dvol.$$
By the linearity of $a$ and the regularity of $u$, integrating in time between $0$ and $T$ we get
\begin{equation}
\label{eq:ubound3}
	\int_0^T\int_Q (\partial_t u(t))^2\,\dvol\,\d t+ \frac12 a(u(T),u(T)) 
		= \frac12 a(u(0),u(0))+ \int_0^T\int_Q f(u(t))\partial_t u(t)\, \dvol\,\d t.
\end{equation}
Recalling the definition of $f(u)$ and exploiting its regularity, we compute
\begin{align*}
	\int_0^T\int_Q f(u(t))\partial_t u(t)\, \dvol\,\d t &= \int_Q\int_0^T f(u(t))\partial_t u(t)\,\d t\, \dvol\\
		&=\int_Q \eta \int_0^T \sin(2u(t)+2\psi_h)\partial_t u(t)\,\d t\, \dvol\\
		&=\int_Q \eta (\cos(2u(0)+2\psi_h)-\cos(2u(T)+2\psi_h))\, \dvol
\end{align*}	
and therefore
$$ \left|\int_0^T\int_Q f(u(t))\partial_t u(t)\, \dvol\,\d t\right| \leq \int_Q 2|\eta| \, \dvol \leq C.$$
Using this estimate in \eqref{eq:ubound3}, we get
$$
	2{\|\partial_t u\|}_{L^2(0,T;H)}^2+ \kappa{\|\nabla_s u(T)\|}_H^2 
		\leq \kappa{\|\nabla_s u_0\|}_H^2+ 2C.
$$
Since the estimate does not depend on $T$, we obtain the second half of \eqref{eq:ubound2}.


\paragraph{Step (iv)} Now, we come to the issue of the long-time behaviour. 
 First, note that the above regularity implies that the set
 $\left\{ \alpha(t), t\in (0,+\infty)\right\}$ is bounded in $H^1(Q)$,
 hence compact in $H$. As a consequence, we
 have that $\omega(\alpha)$ is a nonempty compact set of $H$. 
 Moreover, since by interpolation, $\alpha\in C^0(0,+\infty;H)$,
 a classical dynamical systems argument (see, e.g., \cite{haraux}) shows that 
 $\omega(\alpha)$ is connected in $H$. 
 Consider now an element $\alpha_\infty\in\omega(\alpha)$  and a sequence of times $t_n$ such that 
 $t_n\nearrow +\infty$ for $n\nearrow +\infty$ and $\alpha(t_n)\to \alpha_\infty$ in $H$.
  For any $t\ge 0$, 
 set $\alpha_n(t):=\alpha(t+t_n)$. Note that
 $\| \partial_t\alpha_n\|_{L^2(0,T;H)}\le \|\partial_t \alpha\|_{L^2(t_n,+\infty;H)}.$
Hence, we have that $\forall\, T>0$
$$ \lim_{n\nearrow +\infty}\big( \| \alpha_n-\alpha_\infty\|_{C^0(0,T;H)}
 + \|\partial_t\alpha_n \|_{L^2(0,T;H)} \big) = 0.$$
Thus, passing to the limit with respect to $n$ in \eqref{eq:gradient}, written for $\alpha_n$,
we immediately conclude that $\alpha_\infty$ is a solution of \eqref{eq:el}.
\end{proof}

\subsection{Comparison with the classical energy}
\label{sec:remarks}
In view of the results of this Section, it is worthwhile to compare the predictions of the so called {\itshape intrinsic energy} \eqref{eq:energy2doc} with the ones of the Napoli-Vergori energy \eqref{eq:napolioc} on a torus.
In particular, in \cite{LubPro92} it is found that the Euler-Lagrange equation for \eqref{eq:energy2doc} on a torus is $\Delta \alpha=0$ and simple explicit solutions are given by
\begin{equation}
\label{sol:simpl}
 \alpha(\theta,\phi)=m(\phi-\phi_0)+\theta_0
\end{equation}
for all constants $\theta_0,\phi_0\in \R$, for all $m\in \mathbb Z$. These solutions coincide with those characterised by winding numbers $(0,m)$ in the general family presented in Lemma \ref{lemma:harmonic}. 
Among the solutions \eqref{sol:simpl},  energy \eqref{eq:energy2doc} is minimized by choosing $m=0$ \cite{LubPro92}.
 The second variation of \eqref{eq:energy2doc} is 
\begin{equation*}
	\frac{\d^2}{\d t^2}W^{in}_\kappa(\alpha +t\omega)\big|_{t=0}= \kappa \int_\Sigma |\nabla_s \omega|^2\,\d S.
\end{equation*}
Since it is always nonnegative, and zero if and only if $\omega$ is constant, the conclusion is that every constant $\alpha \equiv\overline \alpha_0 \in \R$ is a global minimizer for $W^{in}_\kappa$, independently of the ratio $R/r$. The scenario depicted by the Napoli-Vergori energy \eqref{eq:napolioc} is quite different. In fact, the presence of the extrinsic term
related to the shape operator acts as a selection principle for equilibrium configurations. More precisely, when $R/r$ is sufficiently large (numerics indicate that the threshold ratio $b^*$ should be between $1.51$ and $1.52$) then 
(see Proposition \ref{prop:stabil}) the only constant solution is $\alpha = \pi/2 +m\pi$ ($m\in \mathbb{Z}$). Moreover, when $R/r< b^*$ a new class of non constant solution appears (see Figures \ref{fig6} and \ref{fig7}). With respect to
 the heuristic principle expressed in \cite{NapVer12L}, that ``the nematic elastic energy promotes the alignment of the flux lines of the nematic director towards geodesics and/or lines of curvature of the surface", we make the following observation: This new solution tries to minimize the effect of the curvature by orienting the director field along the meridian lines ($\alpha = 0$),
which are geodesics on the torus, near the hole of the torus, while near the external equator the director is oriented along the parallel lines $\alpha=\pi/2$, which are lines of curvature. A smooth transition occurs between $\alpha= \pi/2$ and $\alpha = 0$. In this sense, the new solution can be understood as an interpolation between $\alpha= \pi/2$ and $\alpha = 0$, which are the two constant stationary solutions of the system.


\section{Numerical experiments}
\label{sec:numerics}
In this section we report on some simple numerical experiments carried out to approximate minimizers of the one-constant approximation energy \eqref{eq:napolioc} on the axisymmetric torus with radii $0<r<R$ parametrized by \eqref{eq:paramtorus}. Regarding numerics, we note that Monte Carlo methods with simulated annealing were employed in \cite{Santangelo2012,N_G_S_S13, S_K_T_S11} and finite elements on surfaces were developed in \cite{BartEtal12}, in order to study defects evolution and variable surfaces. Since the problem we study is considerably easier, we can afford to use simpler methods. The discussion in sections \ref{sec:alpha} and \ref{sec:torus} shows that instead of studying the minimization on $\Ht$, constrained to the nonconvex subset $\Hut$, we can look at the simpler energy \eqref{eq:expltorus} on $H^1(Q)$, with suitable boundary conditions. Theorem \ref{Th:gf}, in particular, shows that the $L^2$-gradient flow of \eqref{eq:expltorus} is well-posed and its winding number is constant along the flow. Therefore, there exist infinite local minimizers of \eqref{eq:napolioc}, at least one for every element of the fundamental group of the torus $\pi(\bbT)=\mathbb Z \times \mathbb Z$. We actually conjecture that, if $h\neq (0,0)$, there is a unique local minimizer for every $h\in\mathbb Z \times \mathbb Z$ (uniqueness, up to the group of symmetries of $\bbT$, of course).  

\begin{figure}[h!] 
\begin{minipage}{8cm}
\centering{
\labellist
		\hair 2pt
		\pinlabel $\alpha$ at 2920 1670
	\endlabellist				
 \includegraphics[height=6.5cm]{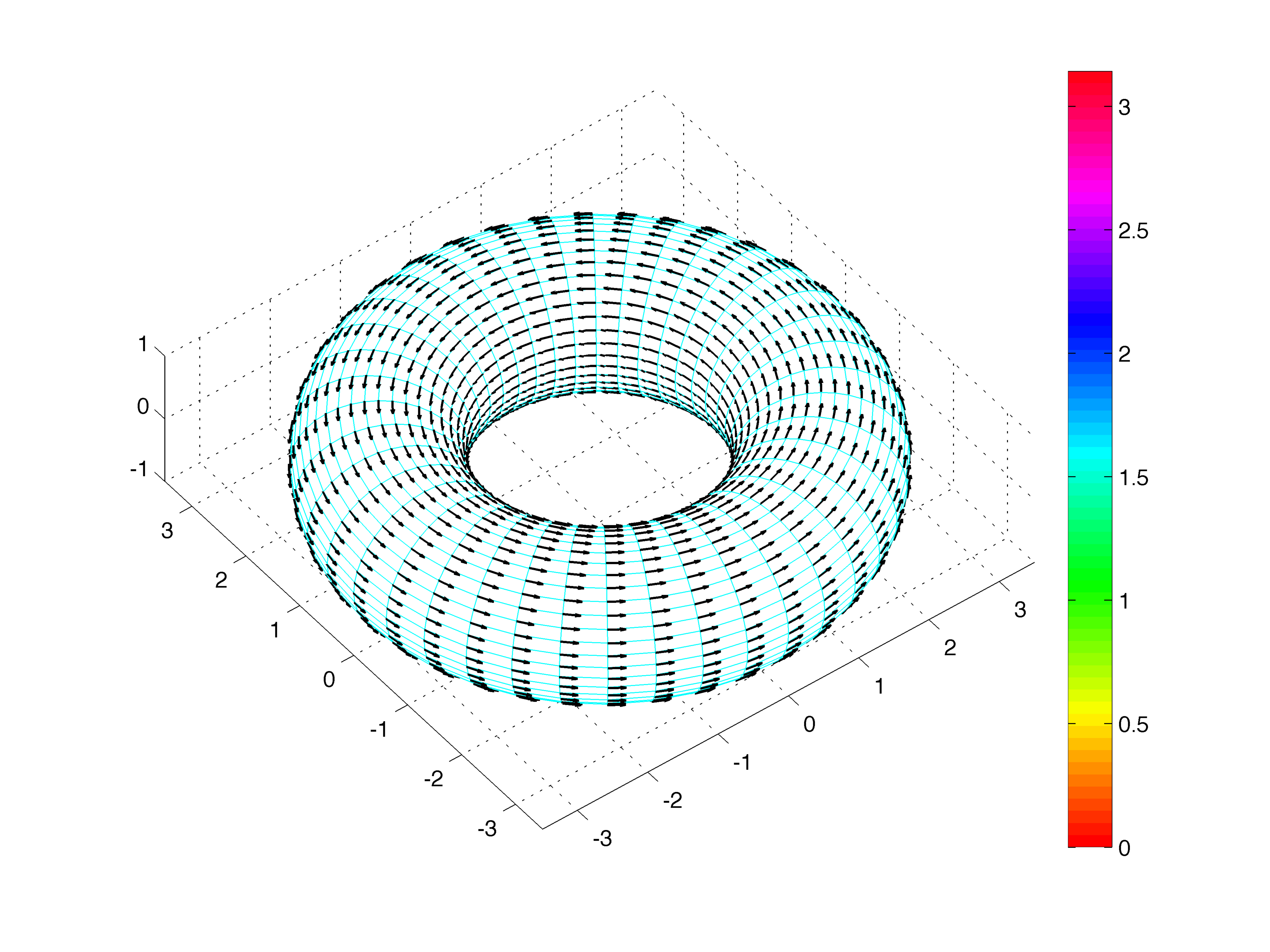}\\ 
}
\end{minipage} \quad
\begin{minipage}{8cm}
 \includegraphics[height=4.5cm]{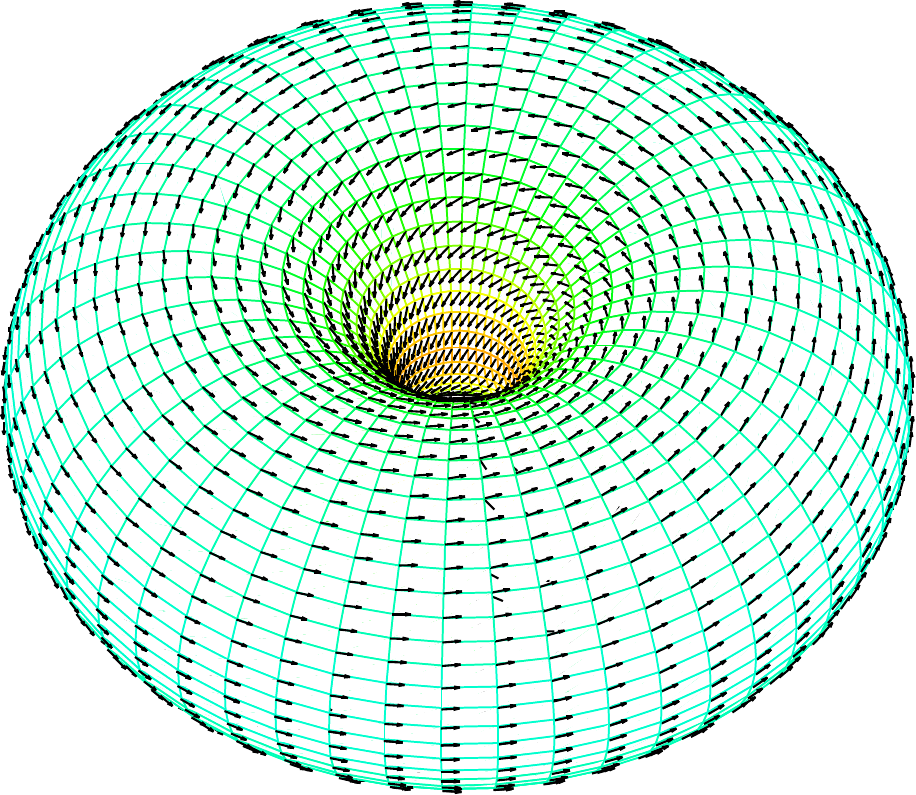}
\end{minipage}\vspace{-0.9cm}
\caption{Configuration of a numerical solution $\alpha$ of the gradient flow.  If  $R/r=2.5$, then $\alpha=\pi/2$, $W(\alpha)=11.61\cdot \pi^2$ (left). When $R/r=1.33$, $W(\alpha)=9.95\cdot \pi^2 < W(\pi/2)=10.22\cdot \pi^2$ (right). The colour represents the angle $\alpha\in[0,\pi]$, the arrows represent the corresponding vector field $\n$.}
\label{fig6}
\end{figure}

For sake of completeness, we detail the method we used in our experiments, but we remark that once the original problem is reduced to the formulation \eqref{eq:gradient}, then any standard method would produce the same results. We discretize the gradient flow \eqref{eq:gradient} with finite differences in space and the Euler
 forward method in time, stopping the evolution when the difference between the (discrete) energy at two consecutive steps is less than $10^{-4}$ (the energy is of the order of 10). Convergence of this discretization scheme is classic, as long as the time step is sufficiently fine with respect to the size of the space grid, according to Von Neumann's stability analysis. The scheme is implemented in Matlab and carried out on a standard laptop (Intel Core i7 \textregistered\,CPU @ 2.8 GHz).  Figures \ref{fig6}--\ref{fig:otto} have relatively rough grids (40x40, 64x64) for graphical purposes, however, refining up to 512x512 yields the same qualitative results. The CPU-time needed for the calculation of one time step on a 256x256 grid, for example, is around 0.02 seconds. 
 
\paragraph{I. Case $h=(0,0)$} As expected from Proposition \ref{prop:stabil}, the numerical experiments indicate that for $R/r\geq 2$, there is one constant global minimizer, given by $\alpha=\pi/2$ (Figure \ref{fig6}), i.e. $\n\equiv \pm\ev$. 
Numerically $\alpha=\pi/2$ remains a minimizer for $R/r\geq 1.52$, while for $R/r < 1.51$, the director field in the inner part of the torus bends in order to follow the geodesics oriented like $\eu$, and the bending becoming steeper and deeper as the ratio $R/r$  decreases (Figures \ref{fig6}-right, \ref{fig7}). Numerical evidence thus suggests that the bifurcation point $b^*$ considered in Proposition \ref{prop:stabil} satisfies $1.51 < b^* <1.52$.

\begin{figure}[h] 
\begin{minipage}{8cm}
\centering{
\includegraphics[trim = 190mm 170mm 100mm 160mm, clip, scale=0.1]{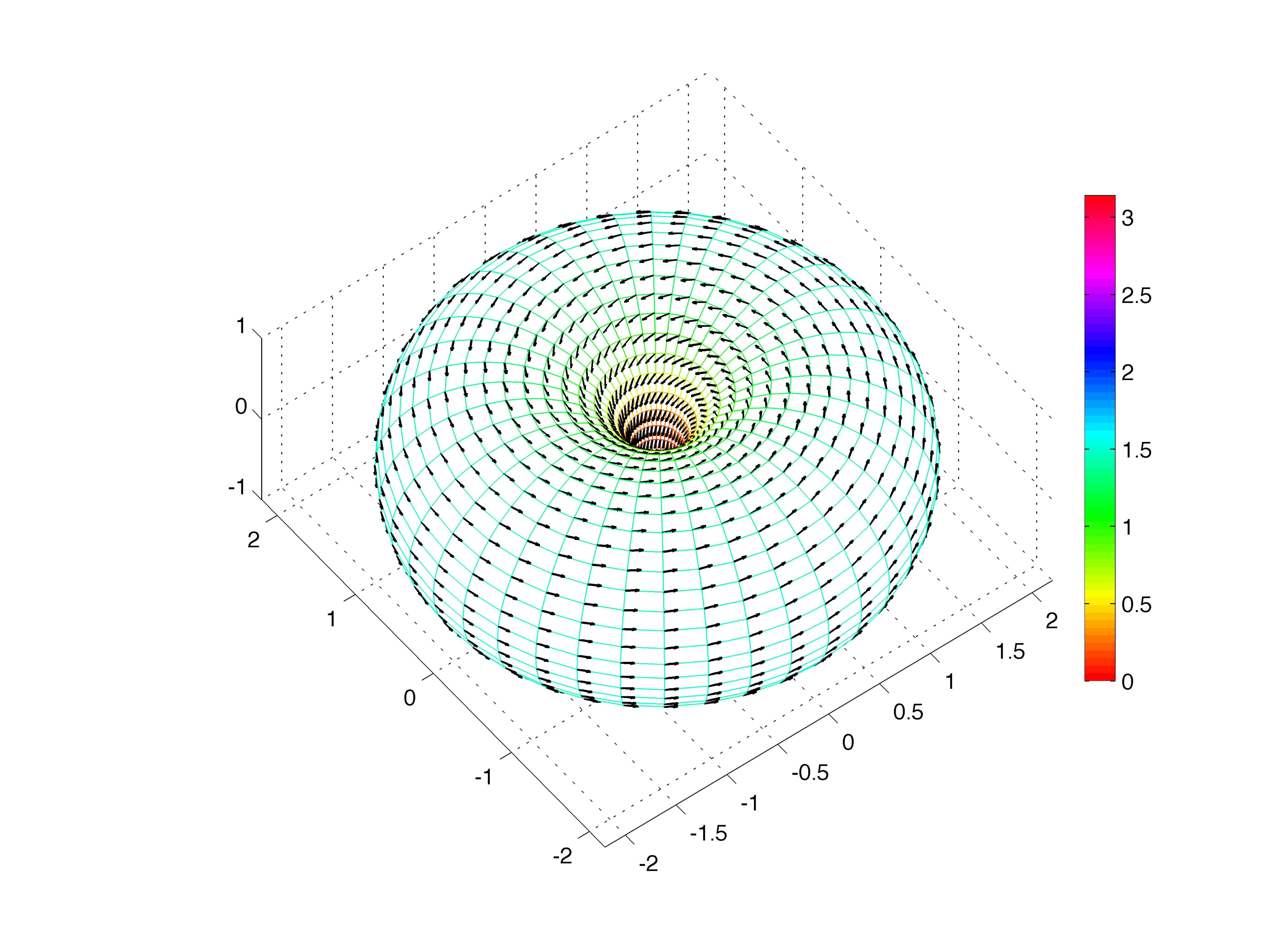}
}
\end{minipage}
\begin{minipage}{8cm}
\centering{
\includegraphics[trim = 39mm 37mm 46mm 21mm, clip, scale=1.61]{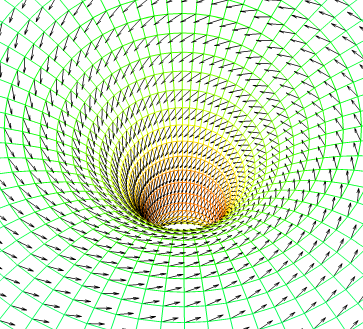}
}
\end{minipage}\vspace{-0.2cm}
\caption{Configuration of the scalar field $\alpha$ and of the vector field $\n$ of a numerical solution to the the gradient flow \eqref{eq:gradient}, in the case $R/r=1.2$ (left). Zoom-in of the central region of the same fields (right). The colour represents the angle $\alpha\in[0,\pi]$, the arrows represent the corresponding vector field $\n$.}
\label{fig7}
\end{figure}

\begin{table}
\begin{center}
  \begin{tabular}{ l || r | r | r | r }
    \hline
       h & 0 & 1 & 2 & 3 \\ \hline \hline
       0 &  10.24  &12.71   & 16.47 & 22.40\\ \hline
    		1 &  14.85 & 16.42  & 20.03  & 25.93\\ \hline
		2 &  25.80   & 27.04     & 30.57 & 36.44\\ \hline
		3 & 43.33  & 44.54  &  48.06 & 53.93 \\ \hline
  \end{tabular}
  \caption{Values of the numerical minimum of the energy, $R/r=2$. The $i$-th row and $j$-th column in the table correspond to index $h=(j,i)$. Values obtained running 30k time-steps, with dt =0.00025, on a 128x128 grid.}
  \label{tab:index}
\end{center} 
\end{table}

\paragraph{II. Case $h\neq(0,0)$} When the initial datum $\alpha_0$ has nonzero winding number $h$ on the torus (see Figure \ref{fig:otto}), the whole evolution takes place in the same homotopy class, approximating a local minimizer with nontrivial winding. 
In Table \ref{tab:index} we collect the numerical values of the energy $W_\kappa$ corresponding to minimizers in different homotopy classes.

\begin{figure} 
\centering{
\begin{minipage}{7cm}
\includegraphics[height=4.5cm]{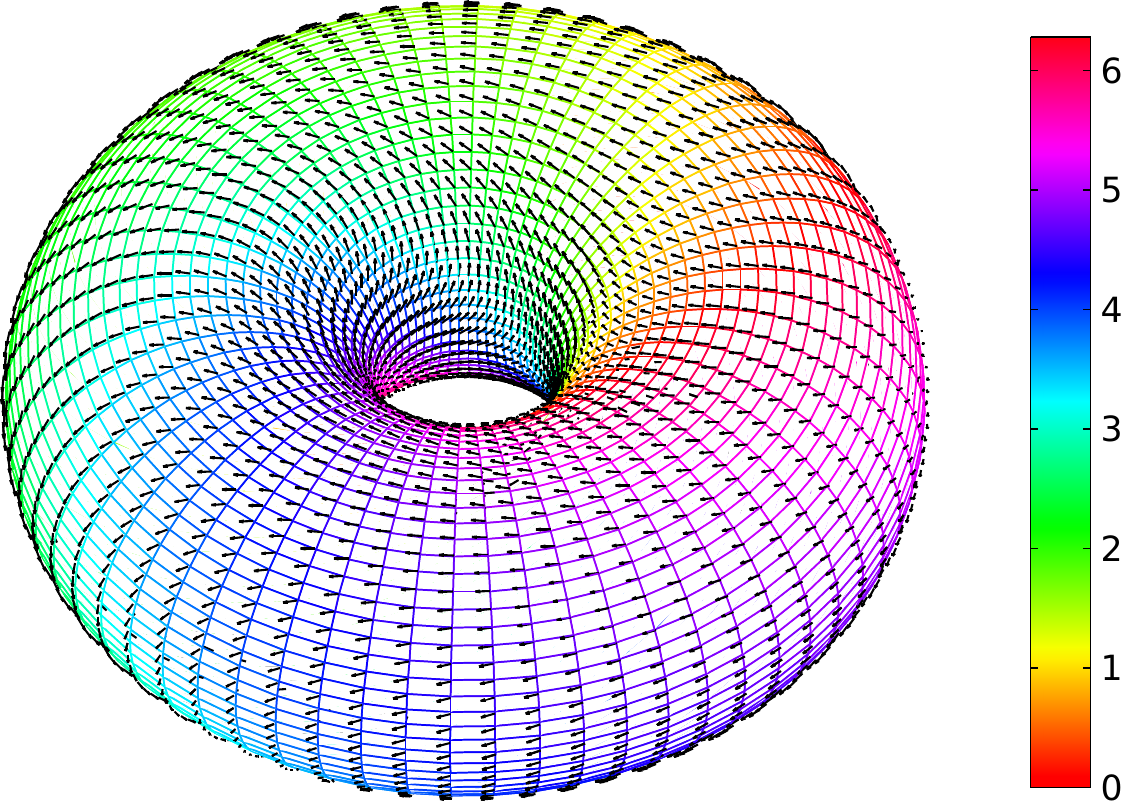}\\ 
\end{minipage}
\begin{minipage}{7cm}
\includegraphics[height=4.5cm]{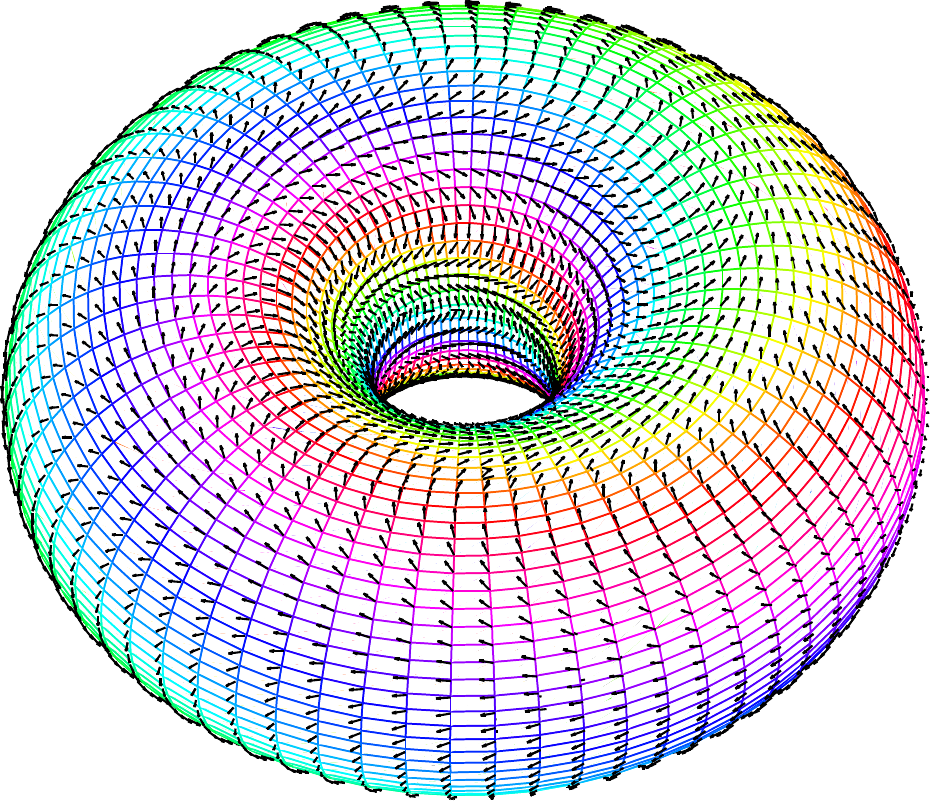}\\ 
\end{minipage}

\vspace{0.5cm}

\begin{minipage}{7cm}
\includegraphics[height=4.5cm]{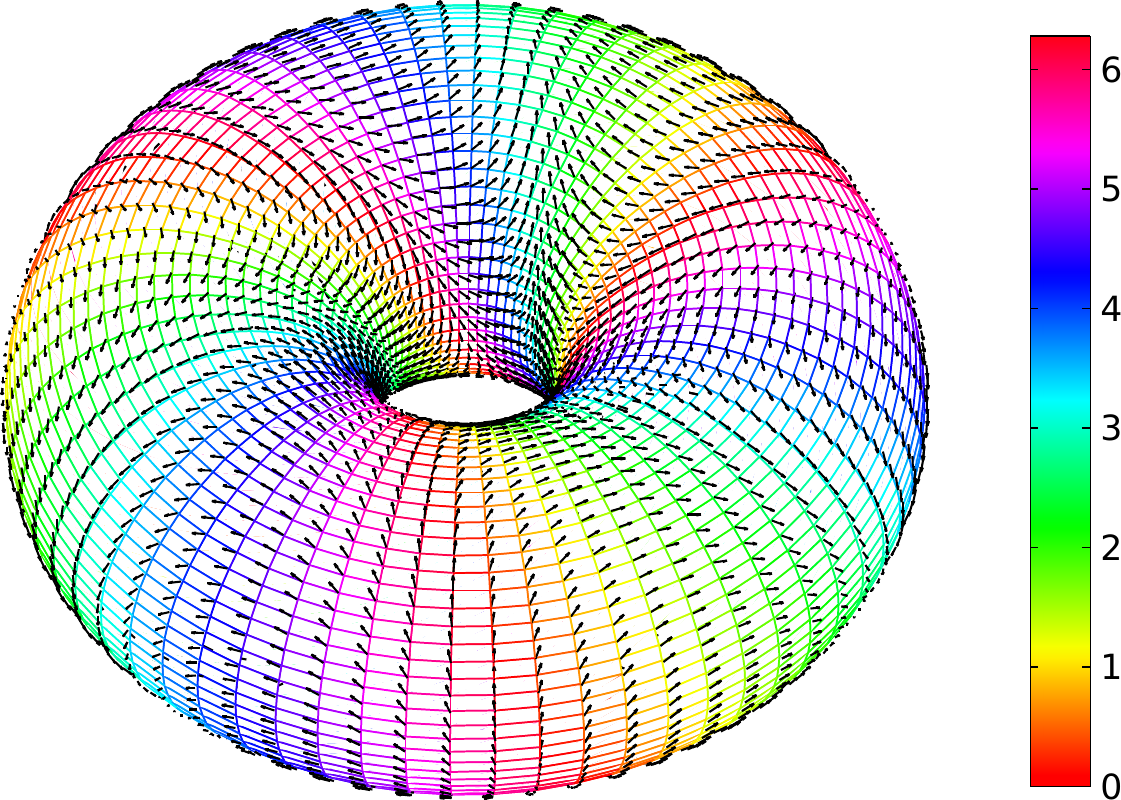}\\ 
\end{minipage}
\begin{minipage}{7cm}
\includegraphics[height=4.5cm]{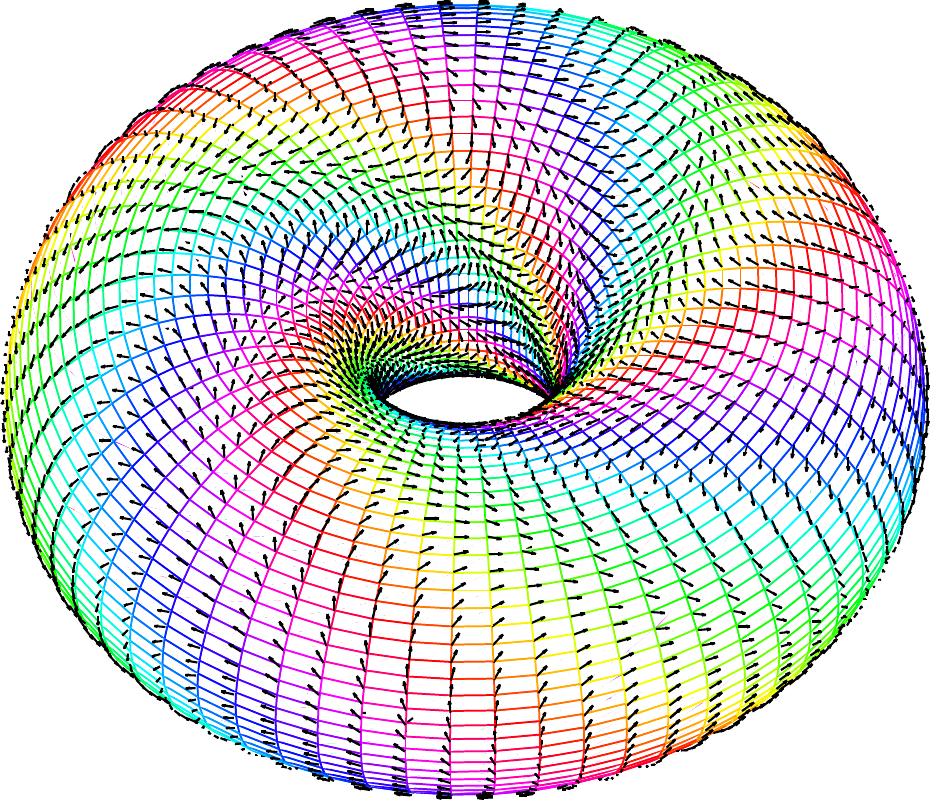}\\ 
\end{minipage}
}
\caption{Some examples of local minimizers with mixed $(\theta,\phi)$ winding numbers. In clockwise order, from top-left corner: index (1,1), (1,3), (3,3), (3,1). The colour represents the angle $\alpha\in[0,2\pi]$, the arrows represent the corresponding vector field $\n$.}
\label{fig:otto}
\end{figure}

\section*{Acknowledgments}
The authors are grateful to Francesco Bonsante, Will Cavendish, Gaetano Napoli, Jean-Cristophe Nave, and Epifanio Virga for inspiring conversations, remarks and suggestions. 
We also acknowledge the anony- mous referee for his/her careful reading of the manuscript and for his/her precise comments which surely improved the presentation of the results. This research started during a visit
of A.S. to the Department of Mathematics and Statistics of McGill University - Montreal, whose kind hospitality is acknowledged. M.S. was introduced to the subject while supported by an NSERC grant. 
A.S. would like to thank the Isaac Newton Institute for Mathematical Sciences, Cambridge, for support and hospitality during the programme Free Boundary Problems and Related Topics where this paper was completed.
Finally, A.S. and M.V. have been supported by the Gruppo Nazionale per l'Analisi Matematica, la Probabilit\`a e le loro Applicazioni (GNAMPA) of the Istituto Nazionale di Alta Matematica (INdAM).


\section*{Appendix A. Geometric quantities on the torus}
\label{app:torus}
\setcounter{equation}{0}
\renewcommand{\theequation}{A.\arabic{equation}}
\renewcommand{\thetheorem}{A.\arabic{theo}}

Let $Q:=[0,2\pi]\times [0,2\pi] \subset \R^2$, and let $X:Q\to \R^3$ be the following parametrization of an embedded torus $\bbT$
\begin{equation}
\label{app:paramtorus}
	X(\theta,\phi) = 
		\begin{pmatrix} 
			(R+r\cos \theta)\cos \phi \\ 
			(R+ r\cos \theta)\sin \phi \\ 
			r\sin \theta
		\end{pmatrix}.
\end{equation}
Using parametrization \eqref{app:paramtorus},  in the next paragraph we derive the main geometrical quantities, like tangent and normal vectors, first and second fundamental form, in order to obtain an explicit expression for the metric and the curvatures of $\mathbb T$ and for $\nabla_s \n$.

Letting
$$ X_\theta:= \frac{\partial}{\partial \theta} X,\qquad X_\phi:= \frac{\partial}{\partial \phi} X,\qquad \nb:=\frac{X_\theta \wedge X_\phi}{|X_\theta \wedge X_\phi|},$$
we have
\begin{gather*}
X_\theta = \begin{pmatrix} - r \sin \theta \cos \phi \\ -r \sin \theta \sin \phi \\ r \cos \theta \end{pmatrix},\qquad 
	X_\phi = \begin{pmatrix} -(R+r\cos \theta)\sin \phi \\ (R+r\cos \theta)\cos \phi \\ 0 \end{pmatrix}, \qquad 
	\nb = -\begin{pmatrix} \cos \theta \cos \phi \\ \cos \theta \sin \phi \\ \sin \theta \end{pmatrix},\\
X_{\theta\theta} = \begin{pmatrix} -r\cos \theta \cos \phi \\ -r\cos \theta \sin \phi \\ -r \sin \theta \end{pmatrix},\qquad
	X_{\phi\phi} = \begin{pmatrix} -(R+r\cos \theta)\cos \phi \\ -(R+r\cos \theta)\sin \phi \\ 0 \end{pmatrix},
\qquad X_{\theta\phi} = \begin{pmatrix} r\sin \theta\sin \phi \\ -r\sin \theta \cos \phi \\ 0 \end{pmatrix}.
\end{gather*}
The unit tangent vectors are
$$ 
\eu(\theta,\phi):= \frac{X_\theta}{|X_\theta|}=
	\begin{pmatrix} 
		- \sin \theta \cos \phi \\ 
		- \sin \theta \sin \phi \\  
		\cos \theta 
	\end{pmatrix},
	\qquad 
	\ev(\theta,\phi):= \frac{X_\phi}{|X_\phi|}= 
	\begin{pmatrix} 
		-\sin \phi \\ 
		\cos \phi \\ 
		0 
	\end{pmatrix}.
$$
Note that this choice of tangent vectors yields an \emph{inner} unit normal $\nb$.  The first and second fundamental forms are  
$$	
g = 
\begin{pmatrix} 
	r^2 & 0 \\ 
	0 & (R+r\cos\theta)^2 
\end{pmatrix},\qquad
	II = 
\begin{pmatrix} 
	\frac 1r  & 0\\ 
	0 & \frac{\cos\theta}{R+r\cos\theta} 
 \end{pmatrix}.
$$
We have $\sqrt{\bar g}=r(R+r\cos\theta)$, $g^{ii}:=(g_{ii})^{-1}$. 
Thus, the shape operator $\mathfrak{B}$ has the form
\begin{equation}
\left\{
	\begin{array}{ll}
		\mathfrak{B}\eu = \frac{1}{r} \eu 	& \\
		\mathfrak{B}\ev = \frac{\cos\theta}{R+r\cos\theta}\ev
	\end{array}
\right.
\end{equation}
from which we have that $\eu$ and $\ev$ are the principal directions.
Then, the principal curvatures are
\begin{equation}
\label{eq:ci}
	c_1= \frac 1r,\qquad c_2=\frac{\cos\theta}{R+r\cos\theta}.
\end{equation}
Now, we compute $(\nabla \mathbf e_i)\mathbf e_j$. Deriving the relation $\ei \cdot \ej=\delta_{ij}$ we see that
\begin{equation}
\label{eq:rel_ei}
	(\nabla \eu)^T\eu = (\nabla \ev)^T\ev  =0\qquad \text{and}\qquad  (\nabla \eu)^T\ev = -(\nabla \ev)^T\eu.
\end{equation}
To differentiate along $\eu$, let 
\[
	\begin{cases} 
		\theta(t) = \frac{t}{r} + \theta_0 \\ 
		\phi(t) = \phi_0 
	\end{cases},
\]	
and set $\gamma(t) = X(\theta(t),\phi(t)).$ We have $\gamma(0) = X(\theta_0,\phi_0)$ and $\gamma'(0) = \frac{1}{r}X_\theta(\theta_0,\phi_0) = \eu(\theta_0,\phi_0)$. Thus, the directional derivatives of $\eu$ and $\ev$ along $\eu$ are given by
\begin{align*}
	(\nabla \eu)\eu 
		= \frac{\d}{\d{t}}\bigg|_{t=0} \frac{1}{r}X_\theta(\theta(t),\phi(t)) 
		= \frac{1}{r^2}X_{\theta\theta}, \qquad
	(\nabla \ev)\eu 
		= \frac{\d{}}{\d{t}}\bigg|_{t=0}  \ev(\theta(t),\phi(t)) 
		= \mathbf 0.
\end{align*}
To differentiate along $\ev$, we set
\[
	\begin{cases} 
		\theta(t) = \theta_0 \\ 
		\phi(t) = \frac{t}{R+r\cos\theta_0} + \phi_0
	\end{cases},
\]
and take $\gamma(t) = X(\theta(t),\phi(t))$, so that $\gamma(0) = X(\theta_0,\phi_0)$ and $\gamma'(0) = \frac{1}{R+r\cos\theta_0}X_\phi(\theta_0,\phi_0) = \ev(\theta_0,\phi_0).$
Thus,
\begin{align*}
	(\nabla \eu)\ev 
		&= 	\frac{\d}{\d{t}}\bigg|_{t=0} \frac{1}{r} X_\theta(\theta(t),\phi(t)) 
			= \frac{1}{r(R+r\cos\theta_0)}X_{\theta\phi}, \\
	(\nabla \ev)\ev  
		&=	\frac{\d{}}{\d{t}}\bigg|_{t=0}  \frac{1}{R+r\cos\theta(t)}X_\phi(\theta(t),\phi(t)) 
			= \frac{1}{(R+r\cos\theta_0)^2} X_{\phi\phi}.
\end{align*}
The geodesic curvatures $\kappa_1$ and $\kappa_2$ of the principal lines of curvature can thus be obtained by
\begin{align*}
	\kappa_1 =\ev(\nabla \eu)\eu &= \frac{1}{R+r\cos\phi}X_\phi \cdot \frac{1}{r^2} X_{\theta\theta} = 0,\\
	\kappa_2 = \ev(\nabla \eu)\ev &= \frac{1}{r(R+r\cos\theta)^2} X_\phi \cdot X_{\theta\phi} = \frac{-\sin\theta}{R+r\cos\theta}.
\end{align*}
By the definition of spin connection $\mathbb A$ in subsection \ref{ssec:formulas}, we also read
\[
	\mathbb A^1= (\eu,D_{\eu}\ev)_{\R^3} \stackrel{\eqref{eq:rel_ei}}{=}-\kappa_1=0,\qquad 
	\mathbb A^2= (\eu,D_{\ev}\ev)_{\R^3} \stackrel{\eqref{eq:rel_ei}}{=}-\kappa_2=\frac{\sin\theta}{R+r\cos\theta}.
\]
The explicit forms of the surface differential operators on the torus are
\begin{align} 
	  \nabla_s \alpha= g^{ii}\partial_i\alpha 
	  	&= \frac{1}{r^2}(\partial_\theta \alpha) X_\theta + \frac{1}{(R+r\cos\theta)^2}(\partial_\phi \alpha) X_\phi \nonumber\\
		&= \frac{1}{r}(\partial_\theta \alpha) \eu + \frac{1}{R+r\cos\theta}(\partial_\phi \alpha) \ev,   \nonumber \\
	\Delta =\frac{1}{\sqrt{\bar g}}\partial_i(\sqrt{\bar g} g^{ij}\partial_j) 
		&=\frac{1}{\sqrt{\bar g}}\left( \partial_\theta\left(\sqrt{\bar g} \frac{1}{r^2}\partial_\theta\right) 
			+\partial_\phi\left(\sqrt{\bar g}  \frac{1}{(R+r\cos\theta)^2}\partial_\phi\right)\right)\nonumber\\
		&= \frac{1}{r^2}\partial^2_{\theta\theta} -\frac{\sin\theta}{r(R+r\cos\theta)}\partial_\theta 
			+ \frac{1}{(R+r\cos\theta)^2}\partial^2_{\phi\phi}.\label{eq:deltas}
\end{align}	
For $ \n = \cos\alpha \eu + \sin \alpha \ev$, the explicit expression of the surface gradient $\nabla_s \n$ in terms of the deviation angle $\alpha$, with respect to the Darboux frame $(\n,\tbf,\nb)$ is
\[
	\nabla_s \n = 
	\begin{pmatrix}
		0 & 0 & 0 \\
		\frac{\alpha_\theta}{r} \cos\alpha + \left(\frac{\alpha_\phi}{R+r\cos \theta} 
			- \frac{\sin\theta}{R+r\cos\theta}\right)\sin\alpha
		& -\frac{\alpha_\theta}{r} \sin\alpha + \left(\frac{\alpha_\phi}{R+r\cos \theta} 
			- \frac{\sin\theta}{R+r\cos\theta}\right)\cos\alpha 	& 0 \\
		\frac{1}{r}\cos^2\alpha +  \frac{\cos\theta}{R+r\cos\theta}\sin^2\alpha
		&  \left( \frac{\cos\theta}{R+r\cos\theta}-\frac{1}{r}\right)\sin\alpha\, \cos\alpha  	& 0 
\end{pmatrix}.
\]

By setting $ \n = \cos\alpha \eu + \sin \alpha \ev$, it is also possible to directly obtain \eqref{eq:gradient} from \eqref{eq:gfn}. In the case of a general surface it holds
\begin{align}
	\partial_t \n & = -(\partial_t \alpha)\sin \alpha\, \eu + (\partial_t \alpha)\cos\alpha\, \ev = -(\partial_t \alpha)\tbf, \label{eq:roughtoneverneverland}\\
	\Delta_g \n &  \stackrel{\eqref{eq:smoothrough}}{=}  -\n |\nabla  \alpha - \mathbb A|^2  - \tbf \Delta \alpha,		\\
	\Bd \n & = \B(c_1 \cos \alpha\, \eu +c_2 \sin \alpha\, \ev) =c_1^2 \cos \alpha\, \eu +c_2^2 \sin \alpha\, \ev \nonumber\\
			& = c_1^2 \cos \alpha (\cos \alpha\, \n +\sin \alpha\, \tbf) + c_2^2 \sin \alpha (\sin \alpha\, \n -\cos \alpha\, \tbf) \nonumber\\
			& = \left(c_1^2\cos^2\!\alpha + c_2^2\sin^2\!\alpha \right)\n + (c_1^2-c_2^2)(\cos\alpha\,\sin\alpha) \tbf, \\
		\vert D\n\vert^2 \n + \vert \B\n\vert^2 \n & = |\nabla_s \n|^2 \n,	
		 \stackrel{\eqref{eq:normsgradalp}}{=} \left(  |\nabla \alpha - \mathbb{A}|^2 	+ c_1^2\cos^2\!\alpha + c_2^2\sin^2\!\alpha \right)\n.\label{eq:cov3}
\end{align}
Substituting \eqref{eq:roughtoneverneverland}--\eqref{eq:cov3} in 
\[
	\partial_t \n -\Delta_g \n + \Bd \n = \vert D\n\vert^2 \n + \vert \B\n\vert^2 \n
\]
yields
\begin{align*}
	-( \partial_t \alpha) \tbf + (\Delta \alpha)\tbf &+\n |\nabla  \alpha - \mathbb A|^2 
			+ 	\left(c_1^2\cos^2\!\alpha + c_2^2\sin^2\!\alpha \right)\n \\
			&+ (c_1^2-c_2^2)(\cos\alpha\,\sin\alpha) \tbf= \left(  |\nabla \alpha - \mathbb{A}|^2 
		+ c_1^2\cos^2\!\alpha + c_2^2\sin^2\!\alpha \right)\n,
\end{align*}
that is,
\[
	\left(\Delta \alpha - \partial_t \alpha + \frac{c_1^2-c_2^2}{2}\sin (2\alpha)\right) \tbf =0.
\]


\section*{Appendix B. Some integration formulas}
\label{app:formulas}
\setcounter{equation}{0}

\renewcommand{\theequation}{B.\arabic{equation}}
\renewcommand{\thetheorem}{B.\arabic{theo}}

Let $b>1$, it holds
\begin{align}
	\int_0^{2\pi} \frac{\sin^2\theta}{b+\cos\theta}d\theta 
			& = 2\pi \left(b - \sqrt{b^2-1}\right), \label{eq:integr1} \\
	\int_0^{2\pi} \frac{\cos^2\theta}{b+\cos\theta}d\theta 
			&=2\pi b\left(\frac{b}{\sqrt{b^2-1}} -1\right). \label{eq:integr2} 
\end{align}

For $\theta\in [0,\pi)$, $b>1$
		\begin{align*} 
			\int \frac{1}{b+\cos \theta}d\theta &= \frac{2}{\sqrt{b^2-1}}\arctan\left( \frac{(b-1)\sin \theta}{\sqrt{b^2-1}(1+\cos\theta)}\right) +c\\
				&=  \frac{2}{\sqrt{b^2-1}}\arctan\left( \frac{(b-1)}{\sqrt{b^2-1}}\tan\left(\frac \theta 2\right)\right) +c.
		\end{align*}	
		Thus,
\[
	\int_0^{2\pi} \frac{1}{b+\cos \theta}d\theta = 2\lim_{s \to \pi^+}\int_0^s \frac{1}{b+\cos \theta}d\theta = \frac{2\pi}{\sqrt{b^2-1}}.
 \]


 \section*{Appendix C. The general Euler-Lagrange equation}
\label{app:euler}
\setcounter{equation}{0}
\renewcommand{\theequation}{C.\arabic{equation}}
\renewcommand{\thetheorem}{C.\arabic{theo}}

The Euler-Lagrange equation of \eqref{eq:energyb} is
\begin{equation}
\label{el:mess}
\begin{split}
	-K_1 \divs\big[((\nabla_s \alpha -\mathbb A)\cdot \tbf) \tbf\big]
		-K_3 \divs\big[((\nabla_s \alpha -\mathbb A)\cdot \n) \n\big] +(K_3-K_1) 
		((\nabla_s \alpha -\mathbb A)\cdot \tbf)((\nabla_s \alpha -\mathbb A)\cdot \n) \\
		- K_3  \frac{(c_1^2-c_2^2)}{2}\sin(2\alpha) 
			+(K_2-K_3) \frac{(c_1-c_2)^2}{4}\sin(4\alpha)=0.
\end{split}		
\end{equation}

\begin{proof} Let $\beta\in C^{\infty}_c(\Sigma)$, we study the first variation of \eqref{eq:energyb} in the direction $\beta$, i.e.
$ \frac{\d}{\d t}W(\n_t)\big|_{t=0},$ where 
\[
	\n_t :=\cos(\alpha+t\beta)\eu +\sin(\alpha+t\beta)\ev,\qquad \tbf_t :=
		-\sin(\alpha +t\beta)\eu +\cos(\alpha+t\beta)\ev. 
\]
It holds $	\frac{\d}{\d t} \n_t = \beta\, \tbf_t,\qquad \frac{\d}{\d t} \tbf_t = -\beta\, \n_t.$  We split the energy $W$ into four terms 
\begin{align*}
	W_1(t) &= \frac {K_1}{2} \int_\Sigma  ((\nabla_s (\alpha+t\beta) -\mathbb A)\cdot \tbf_t)^2 \d S,\\
	W_2(t) &= \frac{K_2}{2} \int_\Sigma(c_1-c_2)^2\sin^2(\alpha+t\beta)\,\cos^2(\alpha+t\beta)\, \d S,\\
		 &= \frac{K_2}{2} \int_\Sigma(c_1-c_2)^2\left(\frac{\sin(2\alpha+2t\beta)}{2}\right)^2\, \d S,\\
	W_{3a}(t) &= \frac{K_3}{2} \int_\Sigma ((\nabla_s (\alpha+t\beta) -\mathbb A)\cdot \n_t)^2 \d S,\\
	W_{3b}(t) &= \frac{K_3}{2} \int_\Sigma(c_1\cos^2(\alpha+t\beta)+ c_2\sin^2(\alpha+t\beta) )^2\, \d S\\
		&= \frac{K_3}{2} \int_\Sigma \left( \frac{c_1 +c_2}{2}+ \frac{c_1-c_2}{2}\cos(2\alpha+2t\beta) \right)^2\, \d S.
\end{align*}
We compute the first variation of each term 
\begin{align*}
	\frac{\d}{\d t}W_1(t)\Big|_{t=0} &=K_1 \int_\Sigma  ((\nabla_s (\alpha+t\beta) -\mathbb A)\cdot \tbf_t)(\nabla_s\beta \cdot \tbf_t + (\nabla_s (\alpha+t\beta) -\mathbb A)\cdot (-\beta\,\n_t) ) \d S \Big|_{t=0}\\
	&= K_1 \int_\Sigma  ((\nabla_s \alpha -\mathbb A)\cdot \tbf)(\nabla_s\beta \cdot \tbf -\beta (\nabla_s \alpha -\mathbb A)\cdot \n)  \d S\\
	&= -K_1 \int_\Sigma  \Big\{\divs\big[((\nabla_s \alpha -\mathbb A)\cdot \tbf) \tbf\big]+ ((\nabla_s \alpha -\mathbb A)\cdot \tbf)((\nabla_s \alpha -\mathbb A)\cdot \n)\Big\}\beta\,  \d S,\displaybreak[2]\\
	\frac{\d}{\d t}W_2(t)\Big|_{t=0} &=\frac{K_2}{4} \int_\Sigma (c_1-c_2)^2\sin(4\alpha)\beta\,  \d S, \displaybreak[2]\\	
\frac{\d}{\d t}W_{3a}(t)\Big|_{t=0} &=	-K_3 \int_\Sigma  \Big\{\divs\big[((\nabla_s \alpha -\mathbb A)\cdot \n) \n\big]- ((\nabla_s \alpha -\mathbb A)\cdot \tbf)((\nabla_s \alpha -\mathbb A)\cdot \n)\Big\}\beta\, \d S,\displaybreak[2]\\
	\frac{\d}{\d t}W_{3b}(t)\Big|_{t=0} &=- \frac{K_3}{2} \int_\Sigma \Big\{(c_1^2-c_2^2)\sin(2\alpha) +\frac{(c_1-c_2)^2}{2}\sin(4\alpha)\Big\}\beta\, \d S.
\end{align*}	
Collecting the four terms, we obtain \eqref{el:mess}
\end{proof}

\bibliographystyle{abbrv}
\bibliography{lib_torus}

\end{document}